\documentclass[final,3p,times]{elsarticle}

\usepackage{amsmath}
\usepackage{amssymb}
\usepackage{amsfonts}
\usepackage{amsthm}
\usepackage{stmaryrd}
\usepackage{float}
\usepackage[caption = false]{subfig}

\newcommand{\eqr}[1]{Eq.\thinspace(#1)}
\newcommand{\pfrac}[2]{\frac{\partial #1}{\partial #2}}

\newcommand{\pfraca}[1]{\frac{\partial}{\partial #1}}

\newcommand{\mvec}[1]{\mathbf{#1}}
\newcommand{\gvec}[1]{\boldsymbol{#1}}

\newcommand{\gcs}{\nabla}
\newcommand{\gvs}{\nabla_{\mvec{v}}}
\newcommand{\gps}{\nabla_{\mvec{z}}}
\newcommand{\dtv}{\thinspace d\mvec{v}}
\newcommand{\dtx}{\thinspace d\mvec{x}}

\newtheorem{proposition}{Proposition}
\newtheorem{corollary}{Corollary}
\newtheorem{lemma}{Lemma}
\newtheorem{remark}{Remark}

\journal{Journal of Computational Physics}

\begin{document}

\begin{frontmatter}



\title{Discontinuous Galerkin algorithms for fully kinetic plasmas}

\author[label1]{J. Juno}
\author[label2]{A. Hakim}
\author[label1,label3]{J. TenBarge}
\author[label2]{E. Shi}
\author[label1]{W. Dorland}
\address[label1]{IREAP, University of Maryland, College Park, MD, 20742}
\address[label2]{Princeton Plasma Physics Laboratory, Princeton, NJ 08543-0451}
\address[label3]{Princeton University, Princeton, NJ, 08544}

\begin{abstract}
We present a new algorithm for the discretization of the non-relativistic Vlasov-Maxwell system of equations for the study of plasmas in the kinetic regime. Using the discontinuous Galerkin finite element method for the spatial discretization, we obtain a high order accurate solution for the plasma's distribution function. Time stepping for the distribution function is done explicitly with a third order strong-stability preserving Runge-Kutta method. Since the Vlasov equation in the Vlasov-Maxwell system is a high dimensional transport equation, up to six dimensions plus time, we take special care to note various features we have implemented to reduce the cost while maintaining the integrity of the solution, including the use of a reduced high-order basis set. A series of benchmarks, from simple wave and shock calculations, to a five dimensional turbulence simulation, are presented to verify the efficacy of our set of numerical methods, as well as demonstrate the power of the implemented features. 
\end{abstract}


\end{frontmatter}
\tableofcontents

\section{Introduction}
\label{}
\begin{par}
Plasmas are ubiquitous in nature, and the study of plasmas has application to a wide variety of problems, from the development of nuclear fusion, to understanding the dynamic interaction between the solar wind and the Earth's magnetosphere, to elucidating the mysteries of large scale astrophysical phenomena such as binary star collisions or the accretion disks of black holes. Unfortunately, many plasmas of interest are only weakly collisional and far from equilibrium, making the system best described by kinetic theory, in this case, the Vlasov equation. The use of kinetic theory significantly complicates the theoretical analysis and simulation of the plasma's dynamics due to the increased dimensionality of the corresponding equations, which are solved in a combined position and velocity phase space, along with the large collection of waves and instabilities that the kinetic system supports. This complication is usually mitigated in one of two ways: either a reduction of the system via asymptotic expansions of the corresponding equations in appropriate limits for the problem of interest or direct numerical simulation of the plasma by approximating it as a collection of ``macroparticles," and employing the particle-in-cell (PIC) algorithm \citep{birdsallbook}. At scales much larger than the ion gyroradius, fluid models for the plasma are excellent tools for analyzing the macroscopic evolution of the plasma. Likewise, reductions of the system which retain kinetic effects, such as the gyrokinetic expansion of the Vlasov equation \citep{Frieman:1982, Brizard:2007}, which averages over the fast cyclotron motion of the particles in the plasma to eliminate one velocity space dimension, have been quite lucrative in both fusion and astrophysical contexts. The PIC algorithm has a long history of success in revealing many novel kinetic features of the plasma's dynamics because, outside of the approximation of integrating along ``macroparticle" trajectories as opposed to individual particle trajectories, the algorithm makes no approximations of the plasma's dynamics and retains the vast array of kinetic phenomena intrinsic to the system.
\end{par}
\begin{par}
However, these two approaches are not without their disadvantages. Fluid models derived from asymptotic expansions of the Vlasov equation often require sub-grid models for the microphysics being approximated, and, while physically motivated, these closures are not always well-suited for the plasma being modeled. Even kinetic reductions which are closed, sometimes rigorously so, such as the aforementioned gyrokinetics, are almost certainly not valid in all contexts; for example, large amplitude fluctuations and waves above the cyclotron frequency observed in the solar wind and solar corona are beyond the scope of traditional gyrokinetic formulations. Thus, a tool like the PIC algorithm is ideal for studying physics in regimes where no tractable and physically reasonable asymptotic reduction exists, but the approximation of the plasma as a collection of macroparticles introduces counting noise into the resulting solutions, which can pollute the results quite severely, or at the very least complicate the analysis. One can always mollify this concern by increasing the number of particles in the simulation, but the counting noise decreases like $1/\sqrt{N}$, where $N$ is the number of particles per grid cell. This scaling can make the study of physics where the signal-to-noise ratio is low especially challenging. For example, Camporeale \emph{et al.} \cite{Camporeale:2016} have demonstrated that a large number of particles-per-cell is required to correctly identify wave-particle resonances and compare well with linear theory. One could employ the delta-f PIC method \citep{Kunz:2014b}, but noise mitigation techniques such as the delta-f PIC method can break down if the distribution function deviates significantly from its initial value.
\end{par}
\begin{par}
Therefore, if the problem of interest requires a full Vlasov-Maxwell description and is characterized by a low signal-to-noise ratio, an alternative approach which directly solves the Vlasov equation and gives access to a noise-free solution is desired. While computationally infeasible in the past due to the need to solve a six-dimensional system (three position and three velocity), plus time, to accurately capture the dynamics of the plasma, in recent years, directly solving the Vlasov equation has become a more popular alternative approach to illuminating the microphysics of the plasma's evolution. Previous Vlasov work has focused almost exclusively on the so-called hybrid framework, treating the electrons as a fluid, usually massless and isothermal, to save significantly on computational cost \citep{Valentini:2007, Valentini2010, Greco2012, Perrone2013, Servidio2014, Valentini2016, Kempf:2012, Kempf2013, Pokhotelov2013}, with some exceptions \citep{Wettervik2016, Vencels:2016}. In this paper, we present the addition of a multi-species Vlasov-Maxwell solver to \texttt{Gkeyll}, a modular framework in which a variety of plasma physics and fluid dynamics solvers are currently being built  \citep{Hakim:2006, Hakim:2008, Ng:2015, Shi:2015, Wang:2015, Shi2017, Cagas:2017, Cagas:2017b}. Within \texttt{Gkeyll}, the Vlasov-Maxwell system is discretized in space with a discontinuous Galerkin finite element method, and discretized in time with a strong stability preserving Runge-Kutta method to form a fully explicit update. The discontinuous Galerkin method combines the power of finite element methods, including high order accuracy and the ability to handle complicated geometries, with the advantages of finite volume methods, including the introduction of limiters to enforce stability and physicality of the solution, and locality of data, for efficient parallelization \citep{Cockburn:1998vt, Cockburn:2001vr, Hesthaven2008}. The high order accuracy and locality of data in particular make our approach especially advantageous; higher order polynomials provide a level of accuracy equivalent to refining the grid at a fraction of the cost, and the locality of data significantly reduces the amount of communication required in the update, enabling the algorithm to scale to well on both standard computing architecture and many-core devices. In fact, the discontinuous Galerkin algorithm has been gaining increased attention as a means of discretizing high dimensional transport equations and various flavors of the Vlasov equation, including Vlasov-Poisson, Vlasov-Ampere, and the aforementioned Vlasov-Maxwell \citep{Cheng:2011, Cheng2014630, Cheng:2014}.
\end{par}
\begin{par}
The paper is organized as follows: Section \ref{sec:Vlasov} provides an overview of the relevant plasma kinetic equation, with discussion of relevant conservation properties in the continuous system. Section \ref{sec:DGVlasov} describes the Runge-Kutta discontinuous Galerkin discretization and proves relevant conservation properties of the discrete system in the continuous time limit. Section \ref{sec:Serendipity} covers the details of the implementation of the numerical algorithm, including the choice of basis functions to mitigate the computational cost in higher dimensions. Section \ref{sec:Results} demonstrates the functionality of the algorithm, including scaling studies and performance analysis of the code on high dimensional benchmarks. Finally, Section \ref{sec:Summary} summarizes our findings and provides future directions for further numerical improvements while discussing current problems in plasma physics which are now within grasp with this new numerical tool.
\end{par}
\section{The Vlasov-Maxwell system} \label{sec:Vlasov}
\subsection{Basic equations}
\begin{par}
The time evolution for the distribution function $f_s(t, \mvec{x}, \mvec{v})$ for species $s$ in a plasma is given by the Vlasov equation,
\begin{align}
 \pfrac{f_s}{t} + \gcs\cdot (\mvec{v} f_s) + \gvs\cdot (\mvec{F}_s f_s) = 0, \label{eqn:vm}
\end{align}
where $\mvec{F}_s=q_s/m_s (\mvec{E}+\mvec{v}\times\mvec{B})$ is the Lorentz force. In the Lorentz force, $q_s$ and $m_s$ are the charge and mass of species $s$ respectively, and $\mvec{E}$ and $\mvec{B}$ are the electric and magnetic fields. The operators $\gcs$ and $\gvs$ are the gradient operators in configuration and velocity space respectively. Here we are considering two simultaneous limits: the non-relativistic\footnote{The special relativistic limit of the Vlasov equation can be obtained by substitution of the Lorentz boost factors where appropriate and a change of variables from velocity to momentum,
\begin{align}
\pfrac{f_s}{t} + \gcs \cdot \left ( \frac{\mvec{p}}{m_s \gamma } f_s \right ) + \nabla_{\mvec{p}} \cdot \left (q_s \left (\mvec{E} + \frac{\mvec{p}}{m_s \gamma} \times \mvec{B} \right ) f_s \right ) = 0 \label{eq:relativisticVlasov}, 
\end{align}
where $\gamma$ is the Lorentz boost factor,
\begin{align}
\gamma = \frac{1}{\sqrt{1 - \frac{|\mvec{p}|^2}{m^2 c^2}}} \label{eq:LorentzBoost}.
\end{align}
}
and the collisionless\footnote{The weakly collisional or collisional limits can be obtained by the addition of the Landau-Fokker-Planck \citep{Rosenbluth1957} collision operator to the right hand side of \eqr{\ref{eqn:vm}},
\begin{align}
\left( \pfrac{f_s}{t} \right )_{\textrm{collisions}} = \sum_{s'} \nu_{s,s'} \gvs \cdot \int d\mvec{v}' \overleftrightarrow{\mvec{U}}(\mvec{v}, \mvec{v}') \cdot \left ( f_{s'}(\mvec{v}') \gvs f_s(\mvec{v}) - \frac{m_s}{m_{s'}} f_s(\mvec{v}) \nabla_{\mvec{v}'} f_{s'}(\mvec{v}') \right ) \label{eq:FPOperator},
\end{align}
where
\begin{align}
\nu_{s,s'} = \frac{q^2_s q^2_{s'} \ln(n_s \lambda_{Ds}^3)}{8 \pi m_s \epsilon_0}
\end{align}
is the collision frequency of species $s$ colliding with species $s'$, and $\overleftrightarrow{\mvec{U}}(\mvec{v}, \mvec{v}')$ is the Landau tensor,
\begin{align}
\overleftrightarrow{\mvec{U}}(\mvec{v}, \mvec{v}') = \frac{1}{|\mvec{v} - \mvec{v}'|} \left (\mvec{I} - \frac{(\mvec{v} - \mvec{v}')(\mvec{v} - \mvec{v}')}{|\mvec{v} - \mvec{v}'|^2} \right ).
\end{align}
Here, $\mvec{I}$ is the identity tensor. In the definition of the collision frequency, $\lambda_{Ds} = \sqrt{\frac{\epsilon_0 T_s}{n_s q_s^2}}$ is the Debye length. Often, the term in the logarithm, $n_s \lambda_{Ds}^3$, is abbreviated as $\Lambda$ and called the plasma parameter. Reductions of the collision operator, such as the Bhatnagar-Gross-Krook (BGK) \citep{Bhatnagar:1954} collision operator,
\begin{align}
\left( \pfrac{f_s}{t} \right )_{\textrm{collisions}} = \nu (f_{eq} - f_s) \label{eq:BGKOperator},
\end{align}
where $f_{eq}$ is some equilibrium distribution function to which the distribution function should relax, are also sometimes considered the weakly collisional or collisional limit. However, given the form of the Landau-Fokker-Planck collision operator, reduced collision operators which maintain a Fokker-Planck structure, i.e., a drag and diffusion term, such as the Lenard-Bernstein \citep{Kirkwood1947, Lenard1958, Dougherty1967} collision operator,
\begin{align}
\left( \pfrac{f_s}{t} \right )_{\textrm{collisions}} = \nu \gvs \cdot ( (\mvec{v} - \mvec{u}_s) f_s + v_{th_s}^2 \nabla_{\mvec{v}} f_s ) \label{eq:LBOperator},
\end{align}
where $\mvec{u}_s$ and $v_{th_s}^2$ are related to the first and second moments of the distribution function respectively, are often preferred. Note that, while the Landau-Fokker-Planck operator naturally generalizes to a multi-species plasma, as well as the relativistic Vlasov equation, reduced collision operators such as the BGK and Lenard-Bernstein collision operators may require significant modification to include the effects of cross-species collisions and relativistic effects.
} limits. 
\end{par}
\begin{par}
Since the motion of charged particles creates currents and electromagnetic fields, the electric and magnetic fields in the Vlasov equation evolve self-consistently by the coupling of the Vlasov equation to Maxwell's equations,
\begin{align}
  \frac{\partial \mvec{B}}{\partial t} + \nabla\times\mvec{E} &= 0, \label{eq:dbdt} \\
  \epsilon_0\mu_0\frac{\partial \mvec{E}}{\partial t} - \nabla\times\mvec{B} &= -\mu_0\mvec{J}, \label{eq:dedt} \\
  \nabla\cdot\mvec{E} &= \frac{\varrho_c}{\epsilon_0}, \label{eq:divE} \\
  \nabla\cdot\mvec{B} &= 0, \label{eq:divB}
\end{align}
where $\varrho_c$ and $\mvec{J}$ are the charge and current densities of the plasma. Charge and current densities are determined by moments of the distribution function
\begin{align}
  \varrho_c = \sum_s  q_s \langle 1 \rangle_s,
\end{align}
and
\begin{align}
  \mvec{J} = \sum_s q_s \langle \mvec{v} \rangle_s,
\end{align}
where the moment operator for any function $\varphi(\mvec{v})$ is defined as
\begin{align}
  \langle \varphi(\mvec{v}) \rangle_s
  \equiv
  \int_{-\infty}^{\infty} \varphi(\mvec{v}) f_s(t,\mvec{x},\mvec{v}) \dtv.
  \label{eq:momentDefinition}
\end{align}
If there are external fields, $\mvec{B}_0$ and $\mvec{E}_0$, e.g., created by coils or electrodes, we can replace $\mvec{B} \rightarrow \mvec{B}+\mvec{B}_0$ and $\mvec{E} \rightarrow \mvec{E}+\mvec{E}_0$ in the Lorentz force term. External fields do not appear in Maxwell's equations.
\end{par}
\subsection{Conservation Properties of the Vlasov-Maxwell system}
\begin{par}
As we mentioned previously, we are interested in deriving a discretization of the non-relativistic, collisionless, Vlasov-Maxwell system. Henceforth, when referring to the Vlasov equation, or the Vlasov-Maxwell system of equations, we emphasize that we are referring to the form of the Vlasov equation given in in \eqr{\ref{eqn:vm}}. Before proceeding to the derivation of the discretization, there are several key conservation properties of the continuous Vlasov-Maxwell system which we review here for the purposes of exploring these same conservation relations in the discrete system. To prove various properties of the Vlasov-Maxwell system we need to assume that $f(t,\mvec{x},\mvec{v}\rightarrow\pm\infty)\rightarrow 0$ faster than $\mvec{v}^n$ for finite $n$. Also, we will assume that either the configuration space is periodic, or that the distribution function vanishes on the boundaries.
\begin{proposition}
The Vlasov-Maxwell system conserves particles.
\end{proposition}
\begin{proof}
Upon integration of \eqr{\ref{eqn:vm}} over all of phase-space and summation over all species, we obtain
  \begin{align}
    \frac{d}{dt}\int_\Omega \sum_s \langle 1 \rangle_s \dtx = 0.
  \end{align}
This relation is a straightforward consequence of the fact that the Vlasov-Maxwell system is written as a conservation law in phase-space. We further note that this relation holds independently for each species because in the absence of recombination and ionization, there is no mechanism for conversion of one particle species to another.
\end{proof}
\begin{proposition}
The collisionless Vlasov-Maxwell system conserves the $L_2$ norm of the distribution function, i.e.,
  \begin{align}
    \frac{d}{dt}\int_K \frac{1}{2} f_s^2\thinspace d\mvec{z} = 0,
  \end{align}
  where the integration is taken over the complete phase-space, $K$.
\end{proposition}
\begin{proof}
  The Vlasov-Maxwell system can be written as a non-linear advection equation in phase-space. To do this, we introduce the phase-space velocity vector $\gvec{\alpha}_s\equiv(\mvec{v},\mvec{F}_s)$ and the phase-space gradient operator $\gps\equiv(\gcs,\gvs)$. The flow in phase-space is incompressible, i.e, $\gps\cdot\gvec{\alpha}_s = 0$, because $\mvec{v}$ is the velocity coordinate and thus has no configuration space dependence, and likewise, the Lorentz force $\mvec{F}_s=q_s/m_s (\mvec{E}+\mvec{v}\times\mvec{B})$ has no velocity space divergence. The latter point is slightly subtle because while the electric field $\mvec{E}$ has no velocity space dependence, the $\mvec{v} \times \mvec{B}$ force has no velocity space divergence by properties of the cross product. \eqr{\ref{eqn:vm}} can then be written as
\begin{align}
  \pfrac{f_s}{t} + \gps\cdot(\gvec{\alpha_s} f_s) = 0, \label{eqn:incompressVM}
\end{align}
where $f_s(t,\mvec{z})$ is the distribution function. If we multiply \eqr{\ref{eqn:incompressVM}} by $f_s$ and integrate over all of phase-space we obtain
  \begin{align}
    \frac{d}{dt}\int_K \frac{1}{2} f_s^2\thinspace d\mvec{z} - \int_K \gps f_s \cdot \gvec{\alpha_s} f_s \thinspace d\mvec{z} = 0.
  \end{align}
We have integrated by parts and used boundary conditions to eliminate the surface term. Now, using incompressibility of phase-space, we can write $\gps f_s \cdot \gvec{\alpha_s} f_s = \gps\cdot(\gvec{\alpha_s} f_s^2)/2$. Using this relation in the above expression and converting the volume integral into a surface integral, combined with the use of boundary conditions, gives the desired conservation law.
\end{proof}
\begin{proposition}
The collisionless Vlasov-Maxwell system conserves the entropy $S = -f \ln(f)$ of the system\footnote{The inclusion of the minus sign in the definition of the entropy is in concordance with the traditional physics definition of the entropy. It is common in the theory of hyperbolic conservation laws to drop the minus sign. This definition has the effect that the physicist's entropy is a non-decreasing quantity, while the mathematician's entropy is a non-increasing quantity.},
\begin{align}
    \frac{d}{dt}\int_K -f_s \ln(f_s) \thinspace d\mvec{z} = 0.
\end{align}
\end{proposition}
\begin{proof}
We multiply \eqr{\ref{eqn:incompressVM}} by $-\ln(f_s)$, integrate over all of phase space, and after a bit of algebra, we obtain,
  \begin{align}
    \frac{d}{dt}\int_K -f \ln(f) \thinspace d\mvec{z} + \int_K \pfrac{f_s}{t} + \gps \ln(f_s) \cdot \gvec{\alpha_s} f_s \thinspace d\mvec{z} = 0.
  \end{align}
In the first two terms, we have used the fact that $-\ln(f) \pfrac{f}{t} = \pfrac{f}{t} - \pfrac{f \ln(f)}{t}$ by the product rule, and in the last term, we have used integration by parts and eliminated the surface term. Now, since $\gps \ln(f_s) = \frac{\gps f_s}{f_s}$, we can see that the latter two terms just become the Vlasov equation written in the form of \eqr{\ref{eqn:incompressVM}}. Since in the collisionless limit the Vlasov equation equals 0, the latter two terms vanish and we are left with our desired conservation relation.
\end{proof}
\begin{proposition}
The Vlasov-Maxwell system conserves total (particles plus field) momentum.
\end{proposition}
\begin{proof}
We multiply \eqr{\ref{eqn:vm}} by $m_s\mvec{v}$, integrate over all of phase-space, and sum over all species to obtain,
  \begin{align}
    \pfrac{}{t} \int_\Omega \sum_s \langle m_s\mvec{v} \rangle_s \dtx + \underbrace{\int_\Omega \gcs\cdot \sum_s \langle m_s \mvec{v}\mvec{v} \rangle_s \dtx}_{=0\ \textrm{on integration by parts and use of BCs}} 
    - \int_\Omega (\varrho_c\mvec{E} + \mvec{J}\times\mvec{B}) \dtx = 0.  \label{eq:mom-proof-1}
  \end{align}
In the last term, we have used integration by parts and $\gvs\mvec{v}=\mvec{I}$. To make further progress, we consider Maxwell's equations. Taking the cross-product of \eqr{\ref{eq:dbdt}} with $\epsilon_0\mvec{E}$, and the cross-product of \eqr{\ref{eq:dedt}} with $\mvec{B}/\mu_0$ and subtracting the resulting equations we obtain
  \begin{align}
    \epsilon_0\pfraca{t} (\mvec{E}\times\mvec{B}) + \epsilon_0 \underbrace{\mvec{E}\times(\nabla\times\mvec{E})}_{(\nabla\mvec{E})\cdot\mvec{E} - (\mvec{E}\cdot\nabla)\mvec{E}}
    + \frac{1}{\mu_0} \underbrace{\mvec{B}\times(\nabla\times\mvec{B})}_{(\nabla\mvec{B})\cdot\mvec{B} - (\mvec{B}\cdot\nabla)\mvec{B}}
    =
    -\mvec{J}\times\mvec{B}.
  \end{align}
Now, for any vector field $\mvec{A}$ we have $(\nabla\mvec{A})\cdot\mvec{A}=\nabla |\mvec{A}|^2/2$ and $(\mvec{A}\cdot\nabla)\mvec{A} = \nabla\cdot(\mvec{A}\mvec{A})-\mvec{A}\nabla\cdot\mvec{A}$. Using these vector identities and the divergence Eqns.\thinspace(\ref{eq:divE}) and (\ref{eq:divB}) gives
  \begin{align}
    \epsilon_0\pfraca{t} (\mvec{E}\times\mvec{B}) + \nabla\left( \frac{\epsilon_0}{2}|\mvec{E}|^2 + \frac{1}{2\mu_0}|\mvec{B}|^2 \right)
    - \nabla\cdot\left( \epsilon_0\mvec{E}\mvec{E} + \frac{1}{\mu_0}\mvec{B}\mvec{B} \right)
    + \varrho_c\mvec{E} = -\mvec{J}\times\mvec{B}. \label{eq:mom-proof-2}
  \end{align}
Finally, inserting \eqr{\ref{eq:mom-proof-2}} into \eqr{\ref{eq:mom-proof-1}}, integrating by parts, and using configuration space boundary conditions gives
  \begin{align}
    \frac{d}{dt} \int_\Omega \left( \sum_s \langle m_s\mvec{v} \rangle_s + \epsilon_0 \mvec{E}\times\mvec{B} \right) \dtx = 0.
  \end{align}
The first term is the particle momentum, and the second term is the field momentum.
\end{proof}

\begin{proposition}
The Vlasov-Maxwell system conserves total (particles plus field) energy. \label{prop:continuous-energy-proof}
\end{proposition}
\begin{proof}
We multiply \eqr{\ref{eqn:vm}} by $m_s |\mvec{v}|^2/2$, integrate over all phase-space, and sum over all species to obtain,
\begin{align}
\pfrac{}{t} \int_\Omega \sum_s \langle \frac{1}{2} m_s |\mvec{v}|^2 \rangle_s \dtx
    +
    \underbrace{\int_\Omega \gcs\cdot \sum_s \langle \frac{1}{2}m_s |\mvec{v}|^2 \mvec{v} \rangle_s \dtx}_{= 0\ \textrm{on integration by parts and use of BCs}}
    -
    \int_\Omega \mvec{E}\cdot\mvec{J} \dtx
    =
    0, \label{eq:fenergy}
  \end{align}
where we have used integration by parts and $\gvs({m_s|\mvec{v}|^2/2})\cdot\mvec{F}_s = q_s \mvec{v} \cdot \mvec{E}$ to evaluate the third term on the left hand side. To make further progress, we again examine Maxwell's equations. Taking the dot product of \eqr{\ref{eq:dedt}} with $\mvec{E}/\mu_0$ and the dot product of \eqr{\ref{eq:dbdt}} with $\mvec{B}/\mu_0$ and adding the resulting equations gives us
  \begin{align}
    \pfraca{t}\left( \frac{\epsilon_0}{2}|\mvec{E}|^2 + \frac{1}{2\mu_0}|\mvec{B}|^2 \right)
    +
    \frac{1}{\mu_0}\underbrace{ \left( \mvec{B}\cdot(\nabla\times\mvec{E}) - \mvec{E}\cdot(\nabla\times\mvec{B}) \right) }_{= \nabla\cdot(\mvec{E}\times\mvec{B})}
    =
    -\mvec{E}\cdot\mvec{J}.
  \end{align}
Using this result in \eqr{\ref{eq:fenergy}}, integrating by parts, and using configuration space boundary conditions gives the total energy conservation law
\begin{align}
   \frac{d}{dt} \int_\Omega \left(\sum_s \langle \frac{1}{2} m_s |\mvec{v}|^2 \rangle_s + \frac{\epsilon_0}{2} |\mvec{E}|^2 + \frac{1}{2\mu_0} |\mvec{B}|^2 \right) \dtx = 0.
\end{align}
The first term in this expression is the total particle energy, and the second two terms make up the electromagnetic field energy.
\end{proof}
\end{par}
\section{The Discrete Vlasov-Maxwell system} \label{sec:DGVlasov}
\subsection{The semi-discrete Vlasov equation}
\begin{par}
To derive the semi-discrete Vlasov equation using a discontinuous Galerkin algorithm, we begin with the Vlasov equation in the form of \eqr{\ref{eqn:incompressVM}}. To discretize this equation, we introduce a phase-space mesh $\mathcal{T}$ with cells $K_j \in \mathcal{T}$, $j=1,\ldots,N$ and introduce the following piecewise polynomial approximation space for the distribution function $f(t,\mvec{z})$
\begin{align}
  \mathcal{V}_h^p = \{ v : v|_{K_j} \in \mvec{P}^p, \forall K_j \in
  \mathcal{T} \},
\end{align}
where $\mvec{P}^p$ is some space of polynomials. The problem can now be stated as finding $f_h\in \mathcal{V}_h^p$ such that, for all $K_j\in \mathcal{T}$,
\begin{align}
  \int_{K_j} w\pfrac{f_h}{t} \thinspace d\mvec{z} + 
  \oint_{\partial K_j}w^- \mvec{n}\cdot\hat{\mvec{F}}  \thinspace dS 
  - \int_{K_j} \gps w \cdot \gvec{\alpha}_h f_h \thinspace d\mvec{z} = 0 \label{eq:dis-weak-form}
\end{align}
for all test functions $w\in \mathcal{V}_h^p$. In this \emph{discrete weak-form}, $\mvec{n}\cdot\hat{\mvec{F}} = \mvec{n}\cdot\hat{\mvec{F}}( \gvec{\alpha}_h^- f_h^-, \gvec{\alpha}_h^+ f_h^+)$ is a numerical flux function, where $\mvec{n}$ is an outward unit vector on the surface of the cell $K_j$. Further, the subscript $h$ indicates the discrete solution and the notation $w^-$ ($w^+$) indicates that the function is evaluated just inside (outside) the location on the surface $\partial K_j$. The particular polynomial space $\mvec{P}^p$ we choose is the so-called Serendipity element space \citep{Arnold:2011eu}. With the choice of polynomial space, we can then say the discrete solution $f_h$ is,
\begin{align}
f_h(t, \mvec{z}) = \sum_i f_i(t) \lambda_i(\mvec{z}), \label{eq:distf-expansion}
\end{align}
where $\lambda_i(\mvec{z})$ are a set of polynomials chosen such that they lie in the aforementioned Serendipity element space. For the numerical flux, we use the local Lax-Friedrichs flux,
\begin{align}
  \mvec{n}\cdot\hat{\mvec{F}}(\gvec{\alpha}_h^-
  f_h^-, \gvec{\alpha}_h^+ f_h^+)
  =
  \frac{1}{2}\mvec{n}\cdot (\gvec{\alpha}_h^+f_h^+ + \gvec{\alpha}_h^-f_h^-)
  -
  \frac{c}{2}(f_h^+-f_h^-),
  \label{eq:num-flux}
\end{align}
where $c =\max(|\mvec{n}\cdot\gvec{\alpha}_h^+|,|\mvec{n}\cdot\gvec{\alpha}_h^-|)$. This particular numerical flux function is a common choice for scalar hyperbolic equations due to its simplicity and efficiency. An important note is that the Lax-Friedrichs flux can be reduced to a pure central flux function by setting $c = 0$, where the average of the cell interface values is used to evaluate the surface integrals. A similar reduction can be done with the appropriate choice of $c$ to create the standard upwind numerical flux, which uses the sign of $\gvec{\alpha}_h$ to determine the direction of the flow of the distribution function across a cell interface. The surface and the volume integrals in \eqr{\ref{eq:dis-weak-form}} are replaced by Gaussian quadrature of an appropriate order to ensure that the discrete integrals are performed exactly.
\end{par}
\subsection{The semi-discrete Maxwell system}
\begin{par}
For Maxwell's equations, we also discretize the system with a discontinuous Galerkin algorithm. We denote the \emph{restriction} of the phase-space mesh, $\mathcal{T}$, to configuration space by $\mathcal{T}_\Omega$. The cells in configuration space are denoted by $\Omega_j\in\mathcal{T}_\Omega$, for $i=0,\ldots,N_\Omega$, where $N_\Omega$ are the number of configuration space cells. We introduce the solution space
\begin{align}
  \mathcal{X}^p_{h} = \mathcal{V}_h^p \backslash \Omega,
\end{align}
i.e, the restriction of the set $\mathcal{V}^p_h$ to configuration space $\Omega$.
\end{par}
\begin{par}
While deriving the discrete weak-form, we need to evaluate volume integrals which include terms of the form $\varphi\nabla\times\mvec{E}$, for $\varphi\in\mathcal{X}^p_{h}$. For example, consider
\begin{align}
  \int_{\Omega_j} 
  \underbrace{
    \varphi\nabla\times\mvec{E}
  }_{
    \nabla\times(\varphi\mvec{E}) - \nabla\varphi\times\mvec{E}
  }
  \dtx.
\end{align}
Gauss' law can then be used to convert one volume integral into a surface integral
\begin{align}
  \int_{\Omega_j} \nabla\times(\varphi\mvec{E}) \dtx
  =
  \oint_{\partial\Omega_j} d\mvec{s}\times(\varphi\mvec{E}),
\end{align}
where $d\mvec{s}$ is the (vector) area-element that points in the direction of the outward normal to the configuration space cell  $\Omega_j$. Using these expressions, we can now write the discrete weak-form of Maxwell's equations as
\begin{align}
  \int_{\Omega_j}\varphi \pfrac{\mvec{B}_h}{t}\dtx
  +
  \oint_{\partial\Omega_j} d\mvec{s}\times(\varphi^-\hat{\mvec{E}}_h)
  -
  \int_{\Omega_j} \nabla\varphi\times\mvec{E}_h \dtx
  =
  0, \label{eq:dis-weak-B}
\end{align}
and
\begin{align}
  \epsilon_0\mu_0\int_{\Omega_j}\varphi \pfrac{\mvec{E}_h}{t}\dtx
  -
  \oint_{\partial\Omega_j} d\mvec{s}\times(\varphi^-\hat{\mvec{B}}_h)
  +
  \int_{\Omega_j} \nabla\varphi\times\mvec{B}_h \dtx
  =
  -\mu_0\int_{\Omega_j} \varphi \mvec{J}_h\dtx, \label{eq:dis-weak-E}
\end{align}
for all $\varphi\in\mathcal{X}^p_{h}$. Here $\hat{\mvec{E}}_h$ and $\hat{\mvec{B}}_h$ are, respectively, the values of the electric and magnetic field at the cell interfaces.
\end{par}
\begin{par}
We consider two methods of obtaining the cell interface fields needed in the discrete weak-form of Maxwell's equations: central fluxes and upwind fluxes. As we will see later, both numerical flux functions have advantages and disadvantages, particularly in terms of their conservation properties. For central fluxes we use averages of values just across the interface, i.e., $\hat{\mvec{E}}_h = \llbracket \mvec{E}\rrbracket$ and $\hat{\mvec{B}}_h = \llbracket\mvec{B}\rrbracket$, where $\llbracket\cdot\rrbracket$ represents the averaging operator,
\begin{align}
  \llbracket g \rrbracket \equiv (g^+ + g^-)/2,
\end{align}
for any function $g$. On the other hand, using \emph{upwind fluxes} requires solving a Riemann problem in a coordinate system local to that face. Consider a local coordinate system $(\mvec{s},\gvec{\tau}_1,\gvec{\tau}_2)$ on the configuration space cell face, i.e., on $\partial\Omega_j$. Here, $\mvec{s}$ is a unit vector normal to $\partial\Omega_j$, and $\gvec{\tau}_1$ and $\gvec{\tau}_2$ are tangent vectors such that $\gvec{\tau}_1\times\gvec{\tau}_2=\mvec{s}$. Let $(E_1,E_2,E_3)$ and
$(B_1,B_2,B_3)$ be electric and magnetic fields in this coordinate system. Then, assuming variations only along direction $\mvec{s}$, Maxwell's equations reduce to $\partial B_1/\partial t = 0$, $\partial E_1/\partial t = 0$, and the following uncoupled set of two equations for the tangential field components
\begin{align}
  \pfrac{B_2}{t} - \pfrac{E_3}{x_1} = 0;\quad \pfrac{E_3}{t} - c^2\pfrac{B_2}{x_1} = 0,
\end{align}
and
\begin{align}
  \pfrac{B_3}{t} + \pfrac{E_2}{x_1} = 0; \quad \pfrac{E_2}{t} + c^2\pfrac{B_3}{x_1} = 0.
\end{align}
Multiplying the first of each pair by $c$ and adding and subtracting from the second of that pair we obtain a set of four uncoupled advection equations
\begin{align}
  \pfraca{t}(E_3+cB_2) - c \pfraca{x_1}(E_3+cB_2) &= 0, \\
  \pfraca{t}(E_3-cB_2) + c \pfraca{x_1}(E_3-cB_2) &= 0,
\end{align}
and
\begin{align}
  \pfraca{t}(E_2+cB_3) + c \pfraca{x_1}(E_2+cB_3) &= 0, \\
  \pfraca{t}(E_2-cB_3) - c \pfraca{x_1}(E_2-cB_3) &= 0.
\end{align}
Hence, the solution to the Riemann problem with initial conditions
\begin{align}
  (E_2,E_3) = (E_2^-,E_3^-);\quad (B_2,B_3) = (B_2^-,B_3^-)
\end{align}
for $x_1<0$, and
\begin{align}
  (E_2,E_3) = (E_2^+,E_3^+); \quad (B_2,B_3) = (B_2^+,B_3^+)
\end{align}
for $x_1>0$ at $x_1=0$ is simply
\begin{align}
  \hat{E}_3 + c\hat{B}_2 &= E_3^+ + c B_2^+ \\
  \hat{E}_3 - c\hat{B}_2 &= E_3^- - c B_2^-
\end{align}
and
\begin{align}
  \hat{E}_2 + c\hat{B}_3 &= E_2^- + c B_3^- \\
  \hat{E}_2 - c\hat{B}_3 &= E_2^+ - c B_3^+.
\end{align}
Rearranging these expressions shows that the upwind fields in the local face coordinate system are
\begin{align}
  \hat{E}_2 = \llbracket E_2 \rrbracket - c\thinspace\{ B_3 \} \label{eq:r-e2} \\
  \hat{E}_3 = \llbracket E_3 \rrbracket + c\thinspace\{ B_2 \} \label{eq:r-e3}
\end{align}
and
\begin{align}
  \hat{B}_2 = \llbracket B_2 \rrbracket + \{E_3\}/c \label{eq:r-b2} \\
  \hat{B}_3 = \llbracket B_3 \rrbracket - \{E_2\}/c \label{eq:r-b3}
\end{align}
where $\{ \cdot \}$ is the jump operator
\begin{align}
  \{ g \} \equiv (g^+-g^-)/2
\end{align}
for any function $g$. The solutions to the Riemann problem given by Eqns.\thinspace(\ref{eq:r-e2})-(\ref{eq:r-b3}) are identical to those presented in previous studies of Maxwell's equations \citep{Barbas2015186}.
\end{par}
\begin{remark}
Before proceeding to the conservation properties of our semi-discrete system, we note that the discretization of Maxwell's equations given by Eqns.\thinspace(\ref{eq:dis-weak-B}) and (\ref{eq:dis-weak-E}) does not include the constraints given by Eqns.\thinspace(\ref{eq:divE}) and (\ref{eq:divB}), i.e., the divergence constraints in Maxwell's equations. Thus, our algorithm may violate these constraints over the course of the simulation. Where appropriate in section \ref{sec:Results}, we will discuss how the violation of the divergence constraints in Maxwell's equations manifests.  
\end{remark}
\subsection{Conservation properties of the semi-discrete system}
\begin{par}
We proceed as we did with the continuous system, first considering whether the discrete system conserves number density.
\begin{proposition}  \label{prop:discrete-particle-cons}
  The discrete scheme conserves total number of particles.
\end{proposition}
\begin{proof}
Choosing $w=1$ in the discrete weak-form, \eqr{\ref{eq:dis-weak-form}}, and summing over all phase-space cells $K_j$, gives the desired conservation law since the volume integral vanishes for $w = 1$ and the surface integrals are symmetric about cell interfaces.
\end{proof}
Before we move on to the $L^2$ norm, we consider the following Lemma,
\begin{lemma}
  \label{lem:discrete-phase-space-incompress}
  Phase space incompressibility holds for the discrete system, i.e.,
  \begin{align}
  \gps \cdot \gvec{\alpha}_h = 0. \label{eq:discrete-ps-incompress}
  \end{align}
\end{lemma}
\begin{proof}
For the specific discrete phase space flow in the Vlasov-Maxwell system, $\gvec{\alpha}_h = (\mvec{v}, q_s/m_s (\mvec{E}_h + \mvec{v} \times \mvec{B}_h))$. Within a cell, \eqr{\ref{eq:discrete-ps-incompress}} is zero since, as with the continuous system, $\mvec{v}$ has no configuration space dependence, and $q_s/m_s (\mvec{E}_h + \mvec{v} \times \mvec{B}_h)$ has no divergence in velocity space. The question is whether the jumps in $\gvec{\alpha}_h$ across cell interfaces in phase space are accounted for by the scheme. Integrating \eqr{\ref{eq:discrete-ps-incompress}} over a phase space cell $K_j$, employing Gauss' Law, and summing over cells,
  \begin{align}
  \sum_j \oint_{\partial K_j} \mvec{n} \cdot \gvec{\alpha}_h^- \thinspace dS= 0.
  \end{align}
This result follows for the simple reason that the phase space flow is in fact continuous with respect to the surfaces considered, allowing us to pairwise cancel the integrand upon summation. For example, consider the configuration space component of the flow $\gvec{\alpha_h}$, $\mvec{v}$. The velocity, $\mvec{v}$, is continuous across configuration space surfaces because $\mvec{v}$ has no configuration space dependence. Likewise, the velocity space component of $\gvec{\alpha_h}, q_s/m_s (\mvec{E}_h + \mvec{v} \times \mvec{B}_h)$, is continuous across velocity space surfaces because $\mvec{E}_h$ and $\mvec{B}_h$ have no velocity space dependence, and $\mvec{v}$ in the $\mvec{v} \times \mvec{B}_h$ term is the velocity coordinate, and thus is continuous. We note that this proof is specific to the phase space flow for the Vlasov-Maxwell system and in general may not hold for all systems.
\end{proof}  
\begin{proposition} \label{prop:discrete-L2-norm}
  The discrete scheme decays the $L^2$ norm of the distribution function monotonically. 
\end{proposition}
\begin{proof}
Since the distribution function itself lies in the test space, we can set $w = f_h$ in \eqr{\ref{eq:dis-weak-form}}. We then have,
  \begin{align}
    \int_{K_j} f_h \pfrac{f_h}{t} \thinspace d\mvec{z} + 
    & \oint_{\partial K_j}f_h^- \mvec{n}\cdot\hat{\mvec{F}}  \thinspace dS 
    -\int_{K_j} \gps f_h \cdot \gvec{\alpha}_h f_h \thinspace d\mvec{z} = \notag \\
    &\frac{1}{2} \int_{K_j} \pfrac{f_h^2}{t} \thinspace d\mvec{z} + 
    \oint_{\partial K_j}f_h^- \mvec{n} \cdot \left(\hat{\mvec{F}} - \gvec{\alpha}_h^- \frac{f_h^-}{2} \right ) \thinspace dS,
  \end{align}
where we have used Lemma\thinspace\ref{lem:discrete-phase-space-incompress} to rewrite $\gps f_h \cdot \gvec{\alpha}_h f_h = \frac{1}{2} \gps \cdot \left (\gvec{\alpha_h} f_h^2 \right )$ since phase space is incompressible, even in our discrete system. Now, summing over all cells and inserting our numerical flux, with a bit of algebra, we obtain,
  \begin{align}
  \sum_j \frac{1}{2} \int_{K_j} \pfrac{f_h^2}{t} \thinspace d\mvec{z} +   \sum_j \frac{1}{2} \oint_{\partial K_j} |\mvec{n} \cdot \gvec{\alpha}_h^\pm | f_h^{-^2} +  |\mvec{n} \cdot \gvec{\alpha}_h^\pm | f_h^{+^2} - f_h^+ f_h^- (  |\mvec{n} \cdot \gvec{\alpha}_h^+ | +  |\mvec{n} \cdot \gvec{\alpha}_h^- | ) \thinspace dS = 0,
  \end{align}
where the $\pm$ is determined by the value of $c = \max(|\mvec{n}\cdot\gvec{\alpha}_h^+|,|\mvec{n}\cdot\gvec{\alpha}_h^-|)$ in the Lax-Friedrichs flux. Note that the sum over surfaces is a sum over pair-wise surfaces, and we have included the contributions on both the left and right side of the surface. Recall that for the proof of phase space incompressibility for the discrete system, we exploited the fact that the phase space flow is continuous with respect to the surfaces considered, i.e., $\mvec{v}$ is continuous across configuration space surfaces and $q_s/m_s (\mvec{E}_h + \mvec{v} \times \mvec{B}_h)$ is continuous across velocity space surfaces. We can use this fact again to simplify our expression,
  \begin{align}
  \sum_j \frac{d}{dt} \int_{K_j} f_h^2 \thinspace d\mvec{z} 
  \le
  -\sum_j \oint_{\partial K_j} \max(|\mvec{n}\cdot\gvec{\alpha}_h^+|,|\mvec{n}\cdot\gvec{\alpha}_h^-|) (f_h^+ - f_h^-)^2 \thinspace dS.
  \end{align}
So, the time evolution of the $L^2$ norm of the distribution function is bounded from above by a negative-definite expression. Thus, the $L^2$ norm of the distribution function is a monotonically decreasing function, and our scheme is $L^2$ stable.
\end{proof}
\begin{corollary} \label{cor:discrete-entropy}
If the discrete distribution function $f_h$ remains positive definite, then the discrete scheme grows the discrete entropy monotonically\footnote{The monotonic growth of the discrete entropy is due to our convention in the definition of the entropy. If one drops the minus sign in the definition of the entropy, then the discrete entropy is a monotonically \textbf{decreasing} function, instead of an increasing function, if the discrete distribution function $f_h$ remains positive definite.},
\begin{align}
\sum_j \frac{d}{dt} \int_{K_j} -f_h \ln(f_h) \thinspace d\mvec{z} \geq 0
\end{align}
\end{corollary}
\begin{proof}
Using the well known bound,
\begin{align}
\ln(x) \leq x - 1, \label{eq:log-bound}
\end{align}
we can see that $\ln(f_h) \leq f_h - 1$, so long as $f_h$ remains a positive definite quantity, and thus $\ln(f_h)$ is well-defined. Multiplying by $-f_h$ then gives us the inequality,
\begin{align}
-f_h \ln(f_h) \geq -f_h^2 + f_h.
\end{align}
But the left-hand side is just the discrete entropy. Integrating over a phase space cell $K_j$, summing over cells, and taking the time-derivative of both sides gives us an expression for the time evolution of the discrete entropy in our scheme,
\begin{align}
\sum_j \frac{d}{dt} \int_{K_j} -f_h \ln(f_h) \thinspace d\mvec{z} \geq \sum_j \frac{d}{dt} \int_{K_j} -f_h^2 + f_h \thinspace d\mvec{z}.
\end{align}
Now, we note that in Proposition \ref{prop:discrete-L2-norm} we have already proved that the $L^2$ norm of the discrete distribution function is a monotonically decaying function, and thus the negative of the $L^2$ norm is a monotonically increasing function, and by Proposition \ref{prop:discrete-particle-cons}, the discrete system conserves particles. Thus, the discrete entropy is a monotonically increasing function. The significance of this is simply that, while we will demonstrate shortly that the semi-discrete scheme derived here conserves energy exactly, the scheme also has a small amount of ``regularization'' associated with it analogous to numerical diffusion present in the discretization of other equation systems via similar methods. Although, in this case, the regularization occurs in the complete phase space as opposed to the numerical diffusion present in discretizations of fluid equations. 
\end{proof}
We will show below that the discrete DG scheme also conserves total energy exactly if \emph{central fluxes} are used for Maxwell's equations. The energy conservation is independent of the flux used in the Vlasov equations. We will first consider the conservation of energy for our discretization of Maxwell's equations.
\begin{lemma}
  \label{lem:em-e-cons}
  The semi-discrete scheme for Maxwell's equations conserves electromagnetic energy exactly when using central fluxes, and is bounded when using upwind fluxes,
  \begin{align}
    \sum_j \frac{d}{dt} \int_{\Omega_j} \left( \frac{\epsilon_0}{2}|\mvec{E}_h|^2 + \frac{1}{2\mu_0}|\mvec{B}_h|^2 \right) \dtx
    \le
    -\sum_j \int_{\Omega_j} \mvec{J}_h\cdot\mvec{E}_h \dtx. \label{eq:em-e-cons}
  \end{align}
By bounded, we mean that when using upwind fluxes, when the right hand side is positive, the electromagnetic energy will increase less than $\left |\sum_j \int_{\Omega_j} \mvec{J}_h\cdot\mvec{E}_h \dtx \right |$, and when the right hand side is negative the electromagnetic energy will decay more than $- \left |\sum_j \int_{\Omega_j} \mvec{J}_h\cdot\mvec{E}_h \dtx \right |$.
\end{lemma}
\begin{proof}
From the discrete weak-form of Maxwell's equations, we need to compute equations for $|\mvec{E}_h|^2$ and $|\mvec{B}_h|^2$. Since each component of the field lies in the selected test space, we take the $i^{th}$-component of \eqr{\ref{eq:dis-weak-B}} and use $B_{hi}$ as a test function, e.g., choose $\varphi=B_{hx}$. Summing these three equations will give us an expression for the time-derivative of $|\mvec{B}_h|^2$. We follow the same procedure for \eqr{\ref{eq:dis-weak-E}}, which gives an expression for the time-derivative of $|\mvec{E}_h|^2$. With a bit of algebra, we obtain,
  \begin{align}
    \frac{d}{dt} \int_{\Omega_j} \frac{1}{2} |\mvec{B}_h|^2 \dtx
    +
    \oint_{\partial\Omega_j} d\mvec{s}\cdot\hat{\mvec{E}}_h\times \mvec{B}^-_h
    +
    \int_{\Omega_j} \mvec{E}_h\cdot\nabla\times\mvec{B}_h \dtx 
    = 0, \label{eq:h-bnum}
  \end{align}
  and
  \begin{align}
    \epsilon_0\mu_0\frac{d}{dt} \int_{\Omega_j} \frac{1}{2} |\mvec{E}_h|^2 \dtx
    -
    \oint_{\partial\Omega_j} d\mvec{s}\cdot\hat{\mvec{B}}_h\times \mvec{E}^-_h
    -
    \int_{\Omega_j} \mvec{B}_h\cdot\nabla\times\mvec{E}_h \dtx
    = -\int_{\Omega_j} \mvec{J}_h\cdot\mvec{E}_h \dtx. \label{eq:h-enum}
  \end{align}
We now multiply both equations by $1/\mu_0$ and add them. Since $\mvec{E}_h\cdot\nabla\times\mvec{B}_h-\mvec{B}_h\cdot\nabla\times\mvec{E}_h = \nabla\cdot(\mvec{B}_h\times\mvec{E}_h)$, we can combine the third terms of Eqns.\thinspace(\ref{eq:h-bnum}) and (\ref{eq:h-enum})
  \begin{align}
    \int_{\Omega_j} \nabla\cdot(\mvec{B}_h\times\mvec{E}_h) \dtx
    = 
    \oint_{\partial\Omega_j} d\mvec{s}\cdot\mvec{B}^-_h\times \mvec{E}^-_h.
  \end{align}
In the above result, note that upon integration by parts, we must use the field \emph{just inside} the face. Hence, the evolution of the electromagnetic energy in a single cell becomes
  \begin{align}
    \frac{d}{dt} \int_{\Omega_j} \left( \frac{\epsilon_0}{2}|\mvec{E}_h|^2 + \frac{1}{2\mu_0}|\mvec{B}_h|^2 \right) \dtx
    +
    \oint_{\partial \Omega_j} d\mvec{s}\cdot \left(\hat{\mvec{E}}_h\times \mvec{B}^-_h + \hat{\mvec{E}}^-_h\times \hat{\mvec{B}}_h - {\mvec{E}}^-_h\times \mvec{B}^-_h \right)
    =
    -\int_{\Omega_j} \mvec{J}_h\cdot\mvec{E}_h \dtx. \label{eq:er-em-gen}
  \end{align}
{\bf Exact Energy Conservation With Central Flux.} Using central-fluxes to determine the interface fields, i.e., setting $\hat{\mvec{E}}_h = \llbracket \mvec{E}\rrbracket$ and $\hat{\mvec{B}}_h = \llbracket\mvec{B}\rrbracket$, gives us,
  \begin{align}
    \frac{d}{dt} \int_{\Omega_j} \left( \frac{\epsilon_0}{2}|\mvec{E}_h|^2 + \frac{1}{2\mu_0}|\mvec{B}_h|^2 \right) \dtx
    +
    \frac{1}{2} \oint_{\partial \Omega_j} d\mvec{s}\cdot \left( \mvec{E}^+_h\times \mvec{B}^-_h + \mvec{E}^-_h\times \mvec{B}^+_h \right)
    =
     -\int_{\Omega_j} \mvec{J}_h\cdot\mvec{E}_h \dtx.
  \end{align}
Upon summing over all configuration space cells, we see that the surface term vanishes as it has opposite signs for the two cells sharing an interface, leading to the desired discrete electromagnetic energy conservation equation,
  \begin{align}
    \sum_j \frac{d}{dt} \int_{\Omega_j} \left( \frac{\epsilon_0}{2}|\mvec{E}_h|^2 + \frac{1}{2\mu_0}|\mvec{B}_h|^2 \right) \dtx
    = -\sum_j \int_{\Omega_j} \mvec{J}_h\cdot\mvec{E}_h \dtx.
  \end{align}
{\bf Energy Bound With Upwind Flux.} To see what happens when using upwind fluxes, we transform the fields appearing in surface integral into the $(\mvec{s},\gvec{\tau}_1,\gvec{\tau}_2)$ coordinate system. We can then write the third term in \eqr{\ref{eq:er-em-gen}} as $d\mvec{s}\cdot(\hat{\mvec{E}}_h\times \mvec{B}^-_h + \dots) = ds (\hat{E}_2B_3^- - \hat{E}_3B_2^-) + \dots$ Using Eqns.\thinspace(\ref{eq:r-e2})-(\ref{eq:r-b3}) for the interface fields and summing over all configuration space cells, we then obtain
  \begin{align}
    \sum_j \frac{d}{dt} \int_{\Omega_j} &\left( \frac{\epsilon_0}{2}|\mvec{E}_h|^2 + \frac{1}{2\mu_0}|\mvec{B}_h|^2 \right) \dtx
    = \notag \\
    &-\sum_j \int_{\Omega_j} \mvec{J}_h\cdot\mvec{E}_h \dtx + \sum_j \oint_{\partial \Omega_j} ds \left( \{E_2\}E_2^-/c+ \{E_3\}E_3^-/c + c \{B_2\}B_2^-+ c \{B_3\}B_3^- \right).
  \end{align}
Note that due to the symmetry of the terms, the central flux terms in Eqns.\thinspace(\ref{eq:r-e2})-(\ref{eq:r-b3}) have vanished on summing over all cells. Now consider the contribution of the term $\{E_2\}E_2^-$ to the two cells adjoining some face. This term will be $(E_2^+-E_2^-)E_2^-/2$ and $(E_2^--E_2^+)E_2^+/2$. On summing over the two cells, this contribution will become $-(E_2^+-E_2^-)^2/2$. Likewise for the other electric and magnetic field coordinates. Hence, the surface terms, on summation, contribute non-positive quantities to the right-hand side, implying that
  \begin{align}
    \sum_j \frac{d}{dt} \int_{\Omega_j} &\left( \frac{\epsilon_0}{2}|\mvec{E}_h|^2 + \frac{1}{2\mu_0}|\mvec{B}_h|^2 \right) \dtx < -\sum_j \int_{\Omega_j} \mvec{J}_h\cdot\mvec{E}_h \dtx.
  \end{align}
\end{proof}

\begin{lemma}
  \label{lem:kin-e-cons}
If $|\mvec{v}|^2$ belongs to the approximation space $\mathcal{V}_h^p$, then the semi-discrete scheme satisfies
  \begin{align}
    \sum_{j} \sum_s \frac{d}{dt} \int_{K_j} \frac{1}{2} m|\mvec{v}|^2 f_h \thinspace d\mvec{z}
    -
    \sum_{j} \int_{\Omega_j} \mvec{J}_h\cdot\mvec{E}_h \dtx
    =
    0. \label{eq:lemma-2}
  \end{align}
Note that the species index is implied, the sum over $j$ in the first term is over all phase space cells, and the sum over $j$ in the second term is over all configuration space cells.
\end{lemma}
\begin{proof}
As we have assumed that $|\mvec{v}|^2\in\mathcal{V}_h^p$, we can set $w=m|\mvec{v}|^2/2$ in \eqr{\ref{eq:dis-weak-form}}. This assumption gives us,
  \begin{align}
    \int_{K_j} \frac{1}{2}m|\mvec{v}|^2 \pfrac{f_h}{t} \thinspace d\mvec{z} 
    + 
    \oint_{\partial K_j} \frac{1}{2}m|\mvec{v}|^2 \mvec{n}\cdot\hat{\mvec{F}}\thinspace dS 
    - 
    \int_{K_j} \underbrace{\gps \left( \frac{1}{2}m|\mvec{v}|^{2} \right) \cdot \gvec{\alpha}_h}_{q\mvec{v}\cdot \mvec{E}_h} f_h \thinspace d\mvec{z} 
    =
    0,    
  \end{align}
where we have used the continuity of $|\mvec{v}|^2$ in the second term, so there is no need to evaluate the basis function just inside the cell. Summing over all species and all phase-space cells, the surface term vanishes as the interface flux contributes terms with opposite signs exactly like our density conservation proof, and we obtain the required identity
  \begin{align}
    \sum_{j} \sum_s \frac{d}{dt} \int_{K_j} \frac{1}{2} m|\mvec{v}|^2 f_h \thinspace d\mvec{z}
    -
    \sum_{j} \int_{\Omega_j} \mvec{J}_h\cdot\mvec{E}_h \dtx
    =
    0.
  \end{align}
Note that we have carried out the sum in velocity space to compute the current density, leaving an integration and sum over configuration space cells. This identity would hold if we also reduced the Lax-Friedrichs flux to either a central flux or upwind flux numerical flux function. We emphasize that this proof only requires $|\mvec{v}|^2\in \mathcal{V}_h^p$ because that is all that is required to substitute $m|\mvec{v}|^2/2$ for the test function $w$. This may appear deceptive as the continuous proof in Proposition \ref{prop:continuous-energy-proof} appears to contain terms which are higher order than $|\mvec{v}|^2$. However, we note that the higher order terms in the continuous proof come from the substitution of the explicit expressions for $\gvec{\alpha}$, the phase space flow, whereas in the discrete proof presented here, we have left the discrete phase space flow $\gvec{\alpha}_h$ as is to stress the fact that $\gvec{\alpha}_h$ has its own basis function expansion. Thus, the higher order terms which are explicit in the continuous energy conservation proof are implicit here in the discrete energy conservation proof. If $|\mvec{v}|^2\in \mathcal{V}_h^p$, then $\mvec{v}\in \mathcal{V}_h^p$ as well, and terms in the discrete phase space flow such as $\mvec{v}$, the configuration space component of the phase space flow, can be exactly represented in terms of our basis function expansion.
\end{proof}
With these lemmas, the proof of total energy conservation is straightforward.
\begin{proposition} \label{prop:DiscreteVlasovEnergyProof}
If central-fluxes are used for Maxwell equations, and if $|\mvec{v}|^2\in \mathcal{V}_h^p$, the semi-discrete scheme conserves total (particles plus field) energy exactly.
\end{proposition}
\begin{proof}
  Using Lemma\thinspace\ref{lem:em-e-cons} and \thinspace\ref{lem:kin-e-cons} demonstrates that the semi-discrete scheme conserves total energy exactly,
  \begin{align}
    \frac{d}{dt} \sum_{j} \sum_s \int_{K_j}\frac{1}{2} m|\mvec{v}|^2 f_h \thinspace d\mvec{z}
    +
    \frac{d}{dt} \sum_j \int_{\Omega_j} \left( \frac{\epsilon_0}{2}|\mvec{E}_h|^2 + \frac{1}{2\mu_0}|\mvec{B}_h|^2 \right) \dtx
    =
    0
  \end{align}
\end{proof}
We remark that if upwind fluxes are used for Maxwell's equations, then the energy will decay monotonically. This is because the bounding of the electromagnetic energy when using upwind fluxes results in a small amount of energy being lost from the system, i.e., when $\mvec{J}_h\cdot\mvec{E}_h$ is positive, more energy is lost from the electromagnetic fields than goes into the particles, and when $\mvec{J}_h\cdot\mvec{E}_h$ is negative, less energy goes into the electromagnetic fields than was given by the particles. Either way, the conversion of electromagnetic energy from, and to, particle energy results in some decrease of the total energy when upwinding is employed with Maxwell's equations. We note though that other authors have demonstrated this loss of energy is small for higher order schemes such as the DG method employed here \citep{Balsara2017104}, and we confirm this fact by studying the energetics of many of the benchmarks presented in Section \ref{sec:Results}. Further, just as with Lemma\thinspace\ref{lem:kin-e-cons}, Proposition\thinspace\ref{prop:DiscreteVlasovEnergyProof} requires that $|\mvec{v}|^2\in\mathcal{V}_h^p$. By employing at least piecewise quadratic polynomials in velocity space, we can insure that $|\mvec{v}|^2$ is in our space of test functions.
\begin{remark}
We note that while employing at least piecewise quadratic polynomials is a way to guarantee $|\mvec{v}|^2\in\mathcal{V}_h^p$, we can also project $|\mvec{v}|^2$ onto piecewise linear basis functions. In that case, it will be a projected energy which is a conserved quantity and thus the semi-discrete system can conserve energy even when utilizing piecewise linear basis functions. 
\end{remark}
\end{par}
\begin{par}
{\bf Lack Of Momentum Conservation.} Note that we have made no mention of whether the discrete system conserves momentum. This omission is because the discrete system does not in fact conserve momentum. We can show momentum non-conservation by choosing $w = m\mvec{v}$ and proceeding as we did with the continuous system,
  \begin{align}
    \int_{K_j} m\mvec{v} \pfrac{f_h}{t} \thinspace d\mvec{z} 
    + 
    \oint_{\partial K_j} m\mvec{v} \mvec{n}\cdot\hat{\mvec{F}}\thinspace dS 
    - 
    \int_{K_j} \gps \left( m\mvec{v} \right) \cdot \gvec{\alpha}_h f_h \thinspace d\mvec{z} 
    =
    0.    
  \end{align}
Since $\mvec{v}$ is continuous, upon summation over all phase space cells and species, we obtain,
  \begin{align}
    \sum_j \sum_s \int_{K_j} m_s \mvec{v} \pfrac{f_h}{t} \thinspace d\mvec{z} 
    - 
    \sum_j \int_{\Omega_j} \rho_{c_h} \mvec{E}_h + \mvec{J}_h \times \mvec{B}_h  \thinspace \dtx
    =
    0. 
  \end{align}
We can proceed exactly as we did with the continuous proof, but we note a key subtlety,
  \begin{align}
    \sum_j \int_{\Omega_j} \frac{d}{dt} \left(\sum_s \langle m_s \mvec{v} \rangle_s + \epsilon_0 \mvec{E}\times\mvec{B} \right) + \nabla\left( \frac{\epsilon_0}{2}|\mvec{E}_h|^2 + \frac{1}{2\mu_0}|\mvec{B}_h|^2 \right) - \nabla\cdot\left( \epsilon_0\mvec{E}_h\mvec{E}_h + \frac{1}{\mu_0}\mvec{B}_h\mvec{B}_h \right) \thinspace \dtx
    =
    0. \label{eq:DiscreteMomentumConservation}
  \end{align} 
Since the electric and magnetic fields are discontinuous across configuration space cell interfaces, we cannot use integration by parts to eliminate the latter two terms. In other words, integration by parts holds only locally and not over the whole domain due to the jumps in the fields across surfaces. However, it is important to note that from the form of this equation that in fact momentum conservation depends only weakly on velocity space resolution, and since the size of the discontinuities in the electric and magnetic fields decrease with increasing configuration space resolution, we can more strongly conserve momentum by increasing configuration space resolution.
\end{par}
\subsection{The time discretization of the semi-discrete Vlasov-Maxwell system} \label{sec:TimeStepping}
\begin{par}
Now, given the semi-discrete Vlasov-Maxwell system, all that remains is to specify a means of solving the resulting system of ordinary differential equations to advance the solution in time. We have implemented a third order strong-stability Runge-Kutta (SSP-RK3) method to generate a fully explicit update \citep{shu2002}. We find for a multi-stage method, SSP-RK3 is a good optimization of memory footprint considerations and accuracy for our time-discretization. For higher order time discretizations, the likely avenue is not more stages in a Runge-Kutta type scheme, but high order single step methods such as an Arbitrary DERivative (ADER) scheme like those employed in Balsara \emph{et al.} \citep{Balsara2016357}. 
\end{par}
\begin{par}
Since the update is fully explicit, it is worth briefly mentioning what constraints on our time step exist, i.e., what Courant-Friedrich-Lewy (CFL) conditions exist for this system of equations. Given that we are solving Maxwell's equations, a constraint in configuration space is the speed of light. In every problem of interest the speed of light will be the largest velocity in the system. This constraint can be written as,
\begin{align}
c \frac{\Delta t}{\Delta x} \leq \frac{1/d}{2p + 1}, \label{eq:MaxwellCFL}
\end{align}
where $c$ is the speed of light, $d$ is the number of configuration space dimensions, and $p$ is the polynomial order of our expansion. This particular CFL condition is similar to the constraint for the finite-difference-time-domain (FDTD) discretization of Maxwell's equations \citep{Yee:1966}, but with an additional safety factor for stability which depends upon the number of configuration space dimensions and polynomial order of our basis expansion. We note that because we will always choose the speed of light to be the largest velocity in the system, \eqr{\ref{eq:MaxwellCFL}} will be a more severe time step restriction than the constraint from the configuration space component of the phase space flow, i.e., the $\mvec{v} \cdot \nabla f$ term. However, \eqr{\ref{eq:MaxwellCFL}} is not the only CFL condition in a fully explicit update of the Vlasov-Maxwell system. Since we are solving for the distribution function on a combined configuration and velocity space grid, there is a corresponding maximum acceleration for a parcel of distribution function in velocity space. The CFL condition is then given by the phase space flow in velocity space,
\begin{align}
\frac{q_s}{m_s} \max_{x,y,z}(\mvec{E} + \mvec{v}_{max} \times \mvec{B}) \frac{\Delta t}{\Delta v} \leq \frac{1/d}{2p + 1},
\end{align}
where the max on the phase space flow is to determine in which direction in velocity space is the acceleration the largest, and this time $d$ is the number of velocity space dimensions. In other words, the maximum component of the total Lorentz force, $\mvec{E} + \mvec{v}\times \mvec{B}$, will also set a constraint on the time step. It is not uncommon, depending upon simulation parameters, for this constraint to be more restrictive than the speed of light CFL. For example, the $\mvec{v} \times \mvec{B}$ force can be quite large at the edge of velocity space for moderate magnetic field amplitudes depending upon how large the velocity grid is. 
\end{par}
\section{Efficient Implementation of Kinetic Update} \label{sec:Serendipity}
\begin{par}
Since the equation system we are solving is intrinsically high-dimensional, with many problems of interest requiring the solution of a four, five, or six dimensional hyperbolic equation, plus time, special care must be taken when dealing with the so-called ``curse of dimensionality,'' which refers to the exponential scaling of solver cost with dimension size. In a finite element framework, the ``curse of dimensionality'' is partly due to the polynomial expansion within a cell. A typical polynomial space employed is the Lagrange Tensor product space, formed by taking tensor products of a one-dimensional monomial expansion. Due to the nature of the tensor product, the number of degrees of freedom within a cell will scale as $(p+1)^d$, where $p$ is the polynomial order and $d$ is the dimension. For four, five, and six dimensional simulations, this scaling quickly becomes intractable. As mentioned previously, we employ a reduced basis set commonly referred to in the literature as the Serendipity basis set, which retains the formal convergence order of the Lagrange Tensor basis, while eliminating many monomials from the tensor product expansion in higher dimensions.\footnote{It should be noted that, while this is true for structured grids, the convergence order of the Serendipity expansion on an unstructured grid is more subtle. This is to say, arbitrary refinements of an unstructured grid will destroy the convergence order of the Serendipity expansion, but with care the convergence order of the space can be maintained \citep{Arnold2000}. All work considered herein is on structured quadrilaterals.} As such, the Serendipity basis' scaling in higher dimensions is much less severe than the Lagrange Tensor basis, scaling like
\begin{align}
N_p = \sum_{i=0}^{\min(d, p/2)} 2^{n-i} \binom{d}{i} \binom{p - i}{i},
\end{align}
where $N_p$ is the number of polynomials in the expansion within a cell, and as before, $p$ is the polynomial order and $d$ is the dimension. The essential feature of the Serendipity polynomial space is that the monomials with ``super-linear'' degree greater than $p$ are not needed to maintain the formal convergence order of the tensor product expansion. By super-linear degree, we mean the degree of the monomial without counting the linear monomials. For example, in three dimensions, with polynomial order two, the Lagrange Tensor basis will have monomials of the form $x^2 y^2 z$, but this monomial has superlinear degree four and is subsequently dropped by the Serendipity basis expansion. A comparison of the number of degrees of freedom within a cell for the two basis sets, given in Tables \ref{table:LagrangeTable} and \ref{table:SerendipityTable}, reveals a stark contrast, especially for five and six dimensions even for relatively low polynomial order.
\begin{table}
\begin{center}
\begin{tabular}{| l | l | l | l | l | l | l | l | l |}
\hline
\textbf{Lagrange Tensor} & \textbf{Polynomial Order} & 1 & 2 & 3 & 4  & 5 & 6 & 7 \\ \hline
\textbf{Dimension} & $(p+1)^d$& & & & & & & \\ \hline
2 & & 4 & 9 & 16 & 25 & 36 & 49 & 64 \\ \hline
3 & & 8 & 27 & 64 & 125 & 216 & 343 & 512 \\ \hline
4 & & 16 & 81 & 256 & 625 & 1296 & 2401 & 4096 \\ \hline
5 & & 32 & 243 & 1024 & 3125 & 7776 & 16807 & 32768 \\ \hline
6 & & 64 & 729 & 4096 & 15625 & 46656 & 117649 & 262144 \\ \hline
\end{tabular}
\caption{Number of degrees of freedom internal to a cell in the Lagrange Tensor polynomial space.}
\label{table:LagrangeTable}
\end{center}
\end{table}

\begin{table}
\begin{center}
\begin{tabular}{| l | l | l | l | l | l | l | l | l |}
\hline
\textbf{Serendipity Element} & \textbf{Polynomial Order} & 1 & 2 & 3 & 4 & 5 & 6 & 7 \\ \hline
\textbf{Dimension} & $\sum_{i=0}^{\min(d, p/2)} \binom{d}{i} \binom{p - i}{i} $& & & & & & & \\ \hline
2 & & 4 & 8 & 12 & 17 & 23 & 30 & 38 \\ \hline
3 & & 8 & 20 & 32 & 50 & 74 & 105 & 144 \\ \hline
4 & & 16 & 48 & 80 & 136 & 216 & 328 & 480 \\ \hline
5 & & 32 & 112 & 192 & 352 & 592 & 952 & 1472 \\ \hline
6 & & 64 & 256 & 448 & 880 & 1552 & 2624 & 4256 \\ \hline
\end{tabular}
\caption{Number of degrees of freedom internal to a cell in the Serendipity polynomial space.}
\label{table:SerendipityTable}
\end{center}
\end{table}
\end{par}
\begin{par}
The actual construction of the algorithm now proceeds in two steps. Having chosen our polynomial space, we must choose a criteria for the creation of the polynomial expansion. Two common choices are a \emph{modal} and \emph{nodal} expansion. These two formulations can be shown to be equivalent mathematically, i.e., there exists a linear transformation between the two choices of basis expansions. However, their comparative computational construction is different. A \emph{modal} expansion specifies only a set of polynomials, which can be as simple as $x^{p-1}$ for polynomial order $p$, to more complicated special basis sets such as Legendre or Chebyshev polynomials. Here we choose the \emph{nodal} expansion, specifying a set of nodes on a reference quadrilateral element for which a polynomial takes the value of 1 at one node, and 0 at all other nodes. The node location is somewhat arbitrary, though certain nodal configurations provide key computational advantages. Here we choose such a nodal layout, schematically given in 1D, 2D, and 3D in Figure \ref{fig:SerendipitySchematic}.
\begin{figure}
\centerline{
\includegraphics[width=\linewidth]{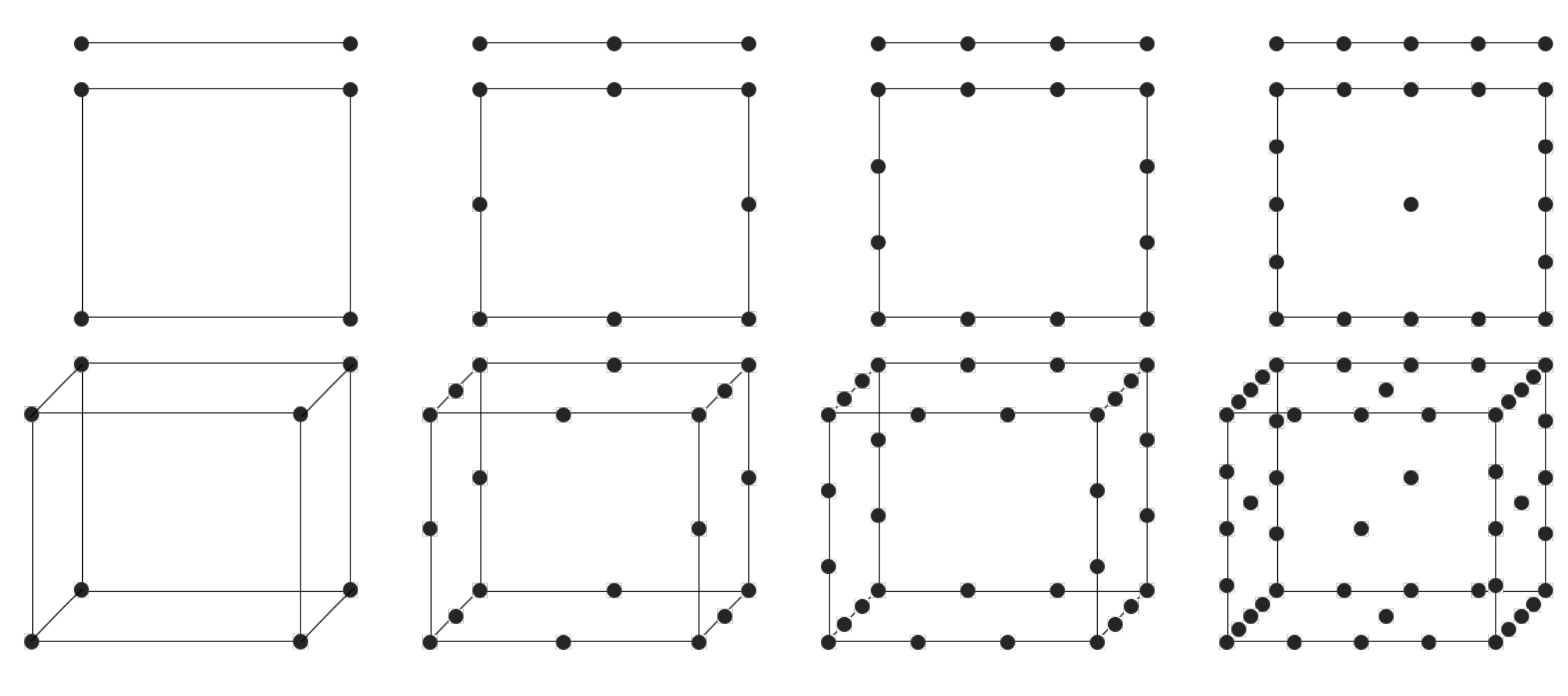}
}
\caption{Schematic drawing of the nodal locations for the Serendipity basis in 1D (top), 2D (middle), and 3D (bottom) for polynomial orders one (far left), two (middle left), three (middle right), and four (far right).} 
\label{fig:SerendipitySchematic}
\end{figure}
This particular nodal layout is constructed such that every higher dimensional reference quadrilateral element is built from the lower dimensional reference quadrilateral elements, so that the lower dimensional faces of a reference quadrilateral element also form a unisolvent expansion, i.e., the polynomials local to the face form a complete basis of the solution space. For example, consider the pictorial representation of the 3D reference element. Each 2D face of the reference 3D element is exactly the 2D reference element for that particular polynomial order. This same recursive approach can be applied to higher dimensions as well\footnote{This fact is true in general, but higher polynomial orders may modify the lower dimensional reference elements such that the recursive algorithm is not quite as obvious as the one presented here. Just as polynomial order 4 introduces an interior node to a reference 2D element, so can higher polynomial orders introduce interior nodes to higher dimensional reference elements which would have to be taken into account in the recursive generation of the reference element. For up to polynomial order 4 though, every higher dimensional object can easily be generated as described, with 2D reference elements making up the faces of a 3D reference element, 3D reference elements making up the faces of a 4D reference element, and so on. We emphasize that this is likely sufficient for our problems since the Serendipity basis in four, five, and six dimensions, with polynomial order 4, involves the solution of a large number of degrees of freedom per cell, and due to the same performance and cost considerations that motivated the use of the Serendipity basis, we seek to avoid evolving thousands of degrees of freedom per cell.}, with the reference 4D element being comprised of reference 3D elements for each of the 4D element's eight 3D faces, and a reference 5D element consisting of a reference 4D element for all ten 4D faces of a 5D element. This approach has the advantage of greatly simplifying surface integral calculations. Since every higher dimensional element is recursively generated from lower dimensional elements, every face of a higher dimensional element, the 2D faces in 3D, or the 4D faces in 5D, forms a unisolvent expansion for that surface. We thus only require the nodal information local to that face and can reduce the number of multiplications in the evaluation of the surface integrals by a somewhat sizable fraction. In 5D for instance, to advance the solution of the distribution function in time, one performs one 5D volume integral and ten 4D surface integrals. By requiring only the nodal information on a 4D surface, the number of degrees of freedom in evaluating the surface integral given in Table \ref{table:SerendipityTable} is reduced by a factor of two to three.
\end{par}
\begin{par}
The second step is then, given the choice of polynomial expansion, determining a means of evaluating the volume and surface integrals in the DG spatial discretization. Gaussian quadrature is a natural choice. The standard approach is to establish a tensor product of $n$ quadrature points in each direction for every dimension one wishes to integrate, and this approach will integrate exactly monomials of a particular order, $2n - 1$ for Gauss-Legendre or $2n - 3$ for Gauss-Lobatto for example, regardless of the dimension in which monomial varies. Here, we employ a Gauss-Legendre integration scheme, but we note that this strategy quickly becomes untenable for the same reason the Lagrange tensor polynomial space is prohibitively expensive for solving the Vlasov equation: the ``curse of dimensionality.'' For example, consider integrating the volume term in five dimensions with second order polynomials. Naively, one expects this to require the integration of monomials with degree $3p = 6$ in each dimension, because both $\gvec{\alpha_h}$ and $f_h$ have polynomial expansions, thus requiring 4 quadrature points in each dimension, or a total of $4^5 = 1024$ quadrature points, to avoid under-integrating the volume term in a cell. Given that the scaling of the computation of the volume integral in a cell is $\mathcal{O}(N_q N_p)$, the number of operations per phase space cell becomes quite large for modest polynomial orders in high dimensions. One way to reduce the number of operations dramatically would be to accept aliasing errors from under-integrating the volume and surface integrals in the DG discretization of the Vlasov equation. Accepting aliasing errors is a common approach to reducing the cost of fluid codes, where the nodes which specify the polynomial expansion in a cell will be co-located with quadrature points, and the interpolation to quadrature points operation is avoided, usually at the cost of under-integrating nonlinear terms. However, this is not an acceptable solution for numerically solving the Vlasov equation because the conservation relations derived in Section \ref{sec:DGVlasov} rely on the integrals being computed correctly, i.e., the conservation relations in the Vlasov-Maxwell system are implicit, as opposed to the explicit conservation relations which form fluid systems of equations. Under-integration of the DG discretization of the Vlasov equation can lead to poor particle and energy conservation, which can eventually lead to numerical instability. 
\end{par}
\begin{par}
Fortunately, the naive requirement to avoid under-integration of terms in our DG discretization is quite far from reality due to the structure present in the Vlasov equation. Consider the advection in velocity space,
\begin{align}
\int_{K_j} \gvs w \cdot \gvec{\alpha_h} f_h  \thinspace d\mvec{z} = \int_{K_j} \gvs w \cdot \frac{q}{m} (\mvec{E}_h + \mvec{v} \times \mvec{B}_h) f_h  \thinspace d\mvec{z}.
\end{align}
To avoid over-complicating the notation, we have dropped the species index. We can see clearly that in fact, integrating monomials of degree $3p$ is only required in configuration space, and in velocity space, we require at most integrating monomials with degree $2p + 1$. We can thus employ an anisotropic grid for performing the quadrature of the advection in velocity space, using only the minimum number of quadrature points required along each direction of integration. Table \ref{table:AnisotropicQuadratureTable} considers the impact anisotropic quadrature has on a few combinations of velocity space and configuration space dimensions. While there is no gain for polynomial order 1, the gain is a moderate improvement for other combinations compared to isotropic quadrature, as shown in Table \ref{table:AnisotropicQuadratureReduction}. A similar reduction in the number of quadrature points required can be demonstrated for the surface integrals. There is an even larger reduction in the number of operations with the configuration space component of the advection in phase space on a structured, Cartesian, grid. Since $\mvec{v}$ is a coordinate, it requires no expansion in basis functions.
\begin{table}
\begin{center}
\begin{tabular}{| l | l | l | l | l | l | l |}
\hline
& \textbf{Base Polynomial Order} & 1 & 2 & 3 & 4 & 5 \\ \hline
\textbf{Dimension} & $((3p+1)/2)^{X_{dim}}\times((2p+2)/2)^3$& & & & & \\ \hline
1X3V & & 16 & 108 & 320 & 875 & 1728 \\ \hline
2X3V & & 32 & 432 & 1600 & 6125 & 13824 \\ \hline
3X3V & & 64 & 1728 & 8000 & 42875 & 110592 \\ \hline
\end{tabular}
\caption{Number of quadrature points required for a few dimension combinations to integrate the volume term for the advection of the distribution function in velocity space.}
\label{table:AnisotropicQuadratureTable}
\end{center}
\end{table}
\begin{table}
\begin{center}
\begin{tabular}{| l | l | l | l | l | l | l |}
\hline
\textbf{Cost(Original/New)} & \textbf{Base Polynomial Order} & 1 & 2 & 3 & 4 & 5 \\ \hline
\textbf{Dimension} & & & & & & \\ \hline
1X3V & & 1 & $\sim 2.37$ & $\sim 1.95$ & $\sim 2.74$ & $\sim 2.37$ \\ \hline
2X3V & & 1 & $\sim 2.37$ & $\sim 1.95$ & $\sim 2.74$ & $\sim 2.37$ \\ \hline
3X3V & & 1 & $\sim 2.37$ & $\sim 1.95$ & $\sim 2.74$ & $\sim 2.37$\\ \hline
\end{tabular}
\caption{Relative gain in the reduction of the number of quadrature points required to integrate the volume term for the advection of the distribution function in velocity space.}
\label{table:AnisotropicQuadratureReduction}
\end{center}
\end{table}
Upon rewriting the term then as,
\begin{align}
\int_{K_j} \nabla w \cdot \mvec{v} f_h  \thinspace d\mvec{z} = \int_{K_j} \nabla w \cdot (\mvec{v} - \mvec{v}_{\textrm{center}}) f_h d\mvec{z} + \int_{K_j} \nabla w \cdot \mvec{v}_{\textrm{center}} f_h  \thinspace d\mvec{z},
\end{align}
we note that both of these terms are such that the integrals can be pre-computed. Here, $\mvec{v}_{\textrm{center}}$ is velocity in the center of the cell. In other words, upon expansion of the distribution function in basis functions given in \eqr{\ref{eq:distf-expansion}}, to update the distribution function in configuration space, we only require two matrices,
\begin{align}
G_{jk} & = \int_{K_j} \nabla w_j \cdot (\mvec{v} - \mvec{v}_{\textrm{center}}) w_k  \thinspace d\mvec{z}, \label{eq:StreamingMatrix1}\\
N_{jk} & = \int_{K_j} \nabla w_j \cdot \mvec{v}_{\textrm{center}} w_k  \thinspace d\mvec{z}, \label{eq:StreamingMatrix2}
\end{align}
along with the cell center coordinate. We remark that $N_{jk}$ is commonly referred to as the grad-stiffness matrix in finite element literature, weighted by the constant cell center velocity coordinate. In rewriting the configuration space update in this way we reduce the update from $\mathcal{O}(N_q N_p)$ operations to $\mathcal{O}(N_p^2)$ operations, where $N_q$ is the number of quadrature points required to exactly integrate either the volume or surface configuration space terms. $N_q$ for the configuration space update scales like $(p+1)^d$, so the gain in writing the update this way is identical to the gain we obtained from employing the Serendipity polynomial space instead of the Lagrange Tensor basis. We note that similar manipulations can be done for the surface integrals in the DG discretization of the Vlasov equation to reduce the number of operations in those components of the update. For example, the quadrature required for the force term surface integrals is reduced not only by the exploitation of anisotropic quadrature requirements, but also by the fact that since as mentioned previously, each set of surface nodes is itself a unisolvent expansion on that surface, we use only that nodal information in the evaluation of the surface integrals. Likewise, the streaming term surface integrals can be precomputed in a similar fashion to the volume integrals due to the fact that the configuration space component of the Lax-Friedrichs flux can be computed analytically,
\begin{align}
\mvec{n}\cdot\hat{\mvec{F}}(\gvec{\alpha}_h^- f_h^-, \gvec{\alpha}_h^+ f_h^+) = 
\begin{cases}
\mvec{n}\cdot\mvec{v} f^- \quad \textrm{if } \mvec{v} > 0, \\
-\mvec{n}\cdot\mvec{v} f^+ \quad \textrm{if } \mvec{v} < 0.
\end{cases}
\end{align}
For the Cartesian grids employed here $\mvec{n}\cdot\mvec{v}$ is simple to calculate, e.g., surface integrals computed at constant $x$ in configuration space just involve the velocity coordinate $v_x$. Given the analytic form of the Lax-Friedrichs flux, similar matrices to \eqr{\ref{eq:StreamingMatrix1}} and \eqr{\ref{eq:StreamingMatrix2}} can be derived and multiplied by $f^-$ or $-f^+$ depending upon the sign of the velocity,
\begin{align}
G_{{jk}_{surface}} & = \int_{\partial K_j} \mvec{n} \cdot (\mvec{v} - \mvec{v}_{\textrm{center}}) w^-_j w_{k_{surface}} \thinspace dS_{config}, \label{eq:StreamingSurfaceMatrix1}\\
N_{{jk}_{surface}} & = \int_{\partial K_j} \mvec{n} \cdot \mvec{v}_{\textrm{center}} w^-_j w_{k_{surface}} \thinspace dS_{config}, \label{eq:StreamingSurfaceMatrix2}
\end{align}
where $w_{k_{surface}}$ are the surface basis functions, i.e., basis functions local to a surface defined by the nodal information on a surface, and we have introduced the surface element $dS_{config}$ to emphasize these surface integrals are for evaluating the configuration space component of the phase space flow and thus surfaces at some constant configuration space dimension. Note that $w_j^-$ is the $j^{\textrm{th}}$ basis function of the full volume expansion, but evaluated at the surface of interest so that the dimensions of this matrix are $N_p \times N_{p_{surface}}$. While individually, each of these reductions in the number of operations per time step is modest, taken altogether, they provide a substantial gain over a traditional finite element approach using a tensor product basis.    
\end{par}
\begin{par}
A final note on the implementation involves the coupling of the Vlasov solver and the solver for Maxwell's equations. Since we require the first velocity moment to advance Maxwell's equations in time, it is worthwhile to consider efficient algorithms for the calculation of velocity moments. For simplicity, let us first consider how one would calculate the zeroth velocity moment, the density. The density in this case is related to the charge density defined earlier,
\begin{align}
n_s(\mvec{x}) & =  \langle 1 \rangle_s \label{eq:densityDefinition} \\
\rho_s(\mvec{x}) & = \sum_s q_s n_s(\mvec{x}). 
\end{align}
In the discrete system, the density $n_h$ belongs to the same space as the electromagnetic fields $\mathcal{X}^p_{h}$, and thus the expansion of the density involves the basis expansion in only configuration space, i.e.,
\begin{align}
n_h(\mvec{x}, t) = \sum_m n_m(t) \varphi_m(\mvec{x}).
\end{align}
The goal is then to determine $n_m(t)$. We do this by multiplying the definition of the discrete zeroth moment on both sides by $\varphi(\mvec{x})_n$ and integrating over configuration space. Note that by discrete zeroth moment, we mean that since phase space has been partitioned into a phase space mesh $\mathcal{T}$ with cells $K_j \in \mathcal{T}, j = 1,\dots, N$, \eqr{\ref{eq:densityDefinition}}, which is a global operation in velocity space, much also be partitioned. We do this by summing over the contributions of the velocity integral inside each cell $K_j$, but only for a fixed position in configuration space, i.e.,
\begin{align}
\sum_m n_m(t) \int \varphi_m(\mvec{x})  \varphi_n(\mvec{x})  \thinspace \dtx = \sum_{K_j \backslash \mathcal{V}} \sum_k f_k(t) \int w_k(\mvec{x}, \mvec{v}) \varphi_n(\mvec{x})  \thinspace \dtx \dtv,
\label{eq:momentAlgorithm}
\end{align}
where our choice of notation $K_j \backslash \mathcal{V}$ implies we will sum only over the space orthogonal to configuration space $\Omega$ for each cell $K_j$. On the left hand side of the \eqr{\ref{eq:momentAlgorithm}}, we note that the integral is what is commonly referred to in finite element literature as the mass matrix, but only in configuration space,
\begin{align}
M_{mn} = \int \varphi_m(\mvec{x})  \varphi_n(\mvec{x})  \thinspace \dtx \label{eq:configMassMatrix}.
\end{align}
The right hand side is what we define as the zeroth moment matrix,
\begin{align}
A^0_{nk} = \int w_k(\mvec{x}, \mvec{v}) \varphi_n(\mvec{x})  \thinspace \dtx \dtv.
\end{align}
We remark that, on uniform, structured grids, the zeroth moment matrix is the same in every cell and thus the algorithm defined above can be pre-computed with minimal memory storage. Non-uniform structured grids would require a small generalization in the form of a constant scale factor. Higher moments, such as the first velocity moment needed to compute the current for Maxwell's equations,
\begin{align}
n_s(\mvec{x}) \mvec{u}(\mvec{x})_s & =  \langle \mvec{v} \rangle_s \label{eq:fluxDefinition} \\
\mvec{J}_s(\mvec{x}) & = \sum_s q_s n_s(\mvec{x}) \mvec{u}_s(\mvec{x}),
\end{align}
are slightly more subtle. Note that a first moment matrix defined as,\
\begin{align}
A^1_{nk} = \int \mvec{v} w_k(\mvec{x}, \mvec{v}) \varphi_n(\mvec{x})  \thinspace \dtx \dtv \notag, 
\end{align}
would be different in every velocity space cell, even on uniform, structured grids. A similar rearrangement to the one employed for the configuration space component of the advection in phase space provides a significant savings in terms of memory,
\begin{align}
A^1_{nk} = \int (\mvec{v} - \mvec{v}_{\textrm{center}}) w_k(\mvec{x}, \mvec{v}) \varphi_n(\mvec{x}) \thinspace \dtx \dtv + \mvec{v}_{\textrm{center}} \int w_k(\mvec{x}, \mvec{v}) \varphi_n(\mvec{x})  \thinspace \dtx \dtv. 
\end{align}
On a Cartesian mesh, it is most obvious how this manifests algorithmically if we break up the update into its vector components. In other words, if the system has three velocity dimensions, then
\begin{align}
A^1_{nk_i} = & \int (v_i - v_{\textrm{center}_i}) w_k(\mvec{x}, \mvec{v}) \varphi_n(\mvec{x})  \thinspace \dtx \dtv + v_{\textrm{center}_i} \int w_k(\mvec{x}, \mvec{v}) \varphi_n(\mvec{x})  \thinspace \dtx \dtv \label{eq:fullMoment1Matrix},
\end{align}
where $i = (x, y, z)$. This result can be generalized to higher moments such as the stress tensor and energy density,
\begin{align}
\overleftrightarrow{\mvec{S}}(\mvec{x})_s & =  m_s \langle \mvec{v} \mvec{v} \rangle_s \\
\epsilon(\mvec{x})_s & = \frac{1}{2}m_s  \langle |\mvec{v}|^2 \rangle_s \label{eq:energyDefinition},
\end{align} 
and beyond.
\end{par}
\begin{par}
Before moving on to benchmarks, we summarize the results of this implementation section by synthesizing all the steps above into a coherent algorithm, i.e., how each operation described manifests in a single Runge-Kutta stage. Define $f_{old}$ and $f_{new}$ to be the values of the distribution function before and after the Runge-Kutta stage respectively.
\begin{enumerate}
\item Compute current from old value of the distribution function
\begin{itemize}
\item The current is computed by summing over the contribution from the first moment for each species,
\begin{align}
J_{m_{i_{old}}} = \sum_s q_s (M_{mn}^{-1} A^1_{nk_i} f_{k_{s_{old}}}) \quad \textrm{where } i = (x,y,z) \textrm{ for } V_{dim} = 3 \label{eq:currentIncrement},
\end{align}
where $M_{mn}^{-1}$ is the inverse of the matrix given by \eqr{\ref{eq:configMassMatrix}} and $A^1_{nk_i}$ is given by \eqr{\ref{eq:fullMoment1Matrix}}.
\end{itemize}
\item Advance the Vlasov equation one time-step for each species.
\begin{itemize}
\item Compute the volume integral component for the streaming term,
\begin{align}
f_{m_{new}} = f_{m_{old}} + \Delta t (M_{mj_{phase}}^{-1} G_{jk} f_{k_{old}} + M_{mj_{phase}}^{-1} N_{jk} f_{k_{old}}),
\end{align}
where $G_{jk}$ is given by \eqr{\ref{eq:StreamingMatrix1}} and $N_{jk}$ is given by \eqr{\ref{eq:StreamingMatrix2}}. The matrix $M_{mj}^{-1}$ in this case is related to \eqr{\ref{eq:configMassMatrix}}, but this matrix involves phase space basis functions,
\begin{align}
M_{mj_{phase}} = \int w_m(\mvec{z})  w_j(\mvec{z})  \thinspace d\mvec{z} \label{eq:phaseMassMatrix}.
\end{align}
\item To obtain the volume integral component of the force term, interpolate the old distribution function, $f_{old}$, and the old velocity component of the phase space flow, $q_s/m_s (\mvec{E}_{old} + \mvec{v} \times \mvec{B}_{old})$, onto Gauss-Legendre abscissas, but with the specified anisotropy given by Table \ref{table:AnisotropicQuadratureTable}, and evaluate the integrals using Gauss-Legendre quadrature. 
\item Increment $f_{m_{new}}$ with this additional contribution multiplied by $\Delta t M_{mj_{phase}}^{-1}$.
\item Compute the surface integral component for the streaming term,
\begin{align}
f_{m_{new}} = f_{m_{old}} + \Delta t (M_{mj_{phase}}^{-1} G_{jk_{surface}} f_{k_{old}} + M_{mj_{phase}}^{-1} N_{jk_{surface}} f_{k_{old}}),
\end{align}
where $G_{jk_{surface}}$ is given by \eqr{\ref{eq:StreamingSurfaceMatrix1}} and $N_{jk_{surface}}$ is given by \eqr{\ref{eq:StreamingSurfaceMatrix2}}. The matrix $M_{mj_{phase}}^{-1}$ is identical to the matrix used in the volume component of the streaming term, i.e., the inverse of \eqr{\ref{eq:phaseMassMatrix}}.
\item To obtain the surface integral component of the force term, we consider the nodal components of $f_{old}$ and the velocity component of the phase space flow, $q_s/m_s (\mvec{E}_{old} + \mvec{v} \times \mvec{B}_{old})$, local to the surface of interest and interpolate that nodal information onto Gauss-Legendre abscissas, again with the specified anisotropy so that the minimum number of quadrature points required are used along configuration and velocity space dimensions. We then evaluate the Lax-Friedrichs flux at each quadrature point, taking into account the sign of $q_s/m_s (\mvec{E}_{old} + \mvec{v} \times \mvec{B}_{old})$ so that the upwind direction is computed correctly, and compute the surface integrals using Gauss-Legendre quadrature.
\item Finally, as before, we increment $f_{m_{new}}$ with this result multiplied by $\Delta t M_{mj_{phase}}^{-1}$.
\end{itemize}
\item Advance Maxwell's Equations one time-step.
\begin{itemize}
\item Since Maxwell's equations are linear, we employ standard techniques for linear advection equations, with the caveat that the form of the numerical flux is either central fluxes, e.g. the simple average of the left and right fields, or upwind, as defined by Eqns.\thinspace(\ref{eq:r-e2}-\ref{eq:r-b3}).
\item We then increment the $i^{\textrm{th}}$ component of the electric field with the current derived from $f_{old}$ given by \eqr{\ref{eq:currentIncrement}}, i.e.,
\begin{align}
E_{m_{i_{new}}} = E_{m_{i_{old}}} + \Delta t \frac{J_{m_{i_{old}}}}{\epsilon_0}.
\end{align}
\end{itemize}
\end{enumerate} 
\end{par}
\begin{par}
This procedure is repeated for each Runge-Kutta stage and is extensible to an arbitrary order multi-step method, though the memory usage of a multi-stage method beyond SSP-RK3 may be prohibitive. We note that the inverses of the various mass matrices which appear do not need to be multiplied through for each time-step on the structured, Cartesian, grids considered in this paper and can be pre-factored into the form of the update step. Further, it is worth acknowledging why for the subset of meshes considered herein we do not pre-evaluate all the integrals present, since polynomials on structured grids can be pre-computed easily. The simple answer is two-fold. The first component of the answer is that for the nodal basis expansions considered, the pre-computation of all the integrals does not provide significant savings in terms of the number of operations required to perform a Runge-Kutta stage. Roughly, this algorithms most expensive component, the non-linear force term, scales like $\mathcal{O}(N_q N_p)$ where $N_q$ is the number of quadrature points, and $N_p$ is the number of basis functions in the volume expansion. Pre-computation of this component of the update reduces the number of multiplications required to $\mathcal{O}(N_c N_p^2)$ where $N_c$ is the number of basis functions in configuration space, i.e., the number of basis functions employed in the expansion of Maxwell's equations and the moments of the distribution function. Upon comparison of these two scalings, we find at best for the high dimensional problems considered a factor of 2 savings in the number of operations. However, the second component of the answer is that the memory requirements of the pre-computation update are more severe due to the need to store a large $N_c N_p^2$ rank three tensor for each velocity dimension. For the form of the basis expansion we employ in this paper, we found that the additional memory requirements were not worth the small reduction in the number of multiplications. This may not be true in general, as a modal basis expansion in orthogonal polynomials could have a significant reduction in the memory footprint due to the resulting sparsity of the tensor. We will consider extensions of the algorithm to different basis expansions in a subsequent publication.  
\end{par}
\section{Applications}\label{sec:Results}
\begin{par}
In the following sections we present a number of benchmarks and studies as verification of our discretization of the Vlasov-Maxwell system. The following set of simulations is by no means exhaustive. Further verification of the Vlasov-Maxwell solver can be found in \citep{Cagas:2017, Cagas:2017b}, which employ the Vlasov-Maxwell solver to study the plasma sheath and the current filamentation instability (CFI), a drift induced instability caused by counter-streaming beams, which can serve to amplify magnetic fields. We note in particular the agreement \citep{Cagas:2017} finds between linear theory of the CFI for long wavelengths and the linear growth of the instability in simulations, as well as the comparable nonlinear saturation amplitudes between two-fluid and Vlasov simulations of the instability. Further details of the nonlinear regime of the CFI, including a detailed analysis of the nonlinear saturation of the CFI can be found in \citep{Cagas:2017b}. 
\end{par}
\subsection{Scaling studies}
\begin{par}
We briefly note a few scaling benchmarks of the Vlasov-Maxwell solver. Using a 1X3V setup, meaning one configuration space dimension $(x)$ and three velocity dimensions $(v_x, v_y, v_z)$, with $256 \times 16^3$ cells and polynomial order 4 with the Serendipity space, 136 degrees of freedom per cell, a strong scaling study demonstrates the solver's ability to scale to thousands of processors. We can likewise perform a weak scaling study, starting from a box size of $64 \times 8^3$ box on 16 processors and increasing to a the box size of $1024 \times 16^3$ on 2048 processors. Each simulation is run for hundreds of time steps to minimize the ratio of time spent initializing the simulation to time spent taking time steps. Results for these two scaling studies are given in Figure \ref{fig:ScalingResults}.
\begin{figure}
\centerline{
\includegraphics[width=0.5\linewidth]{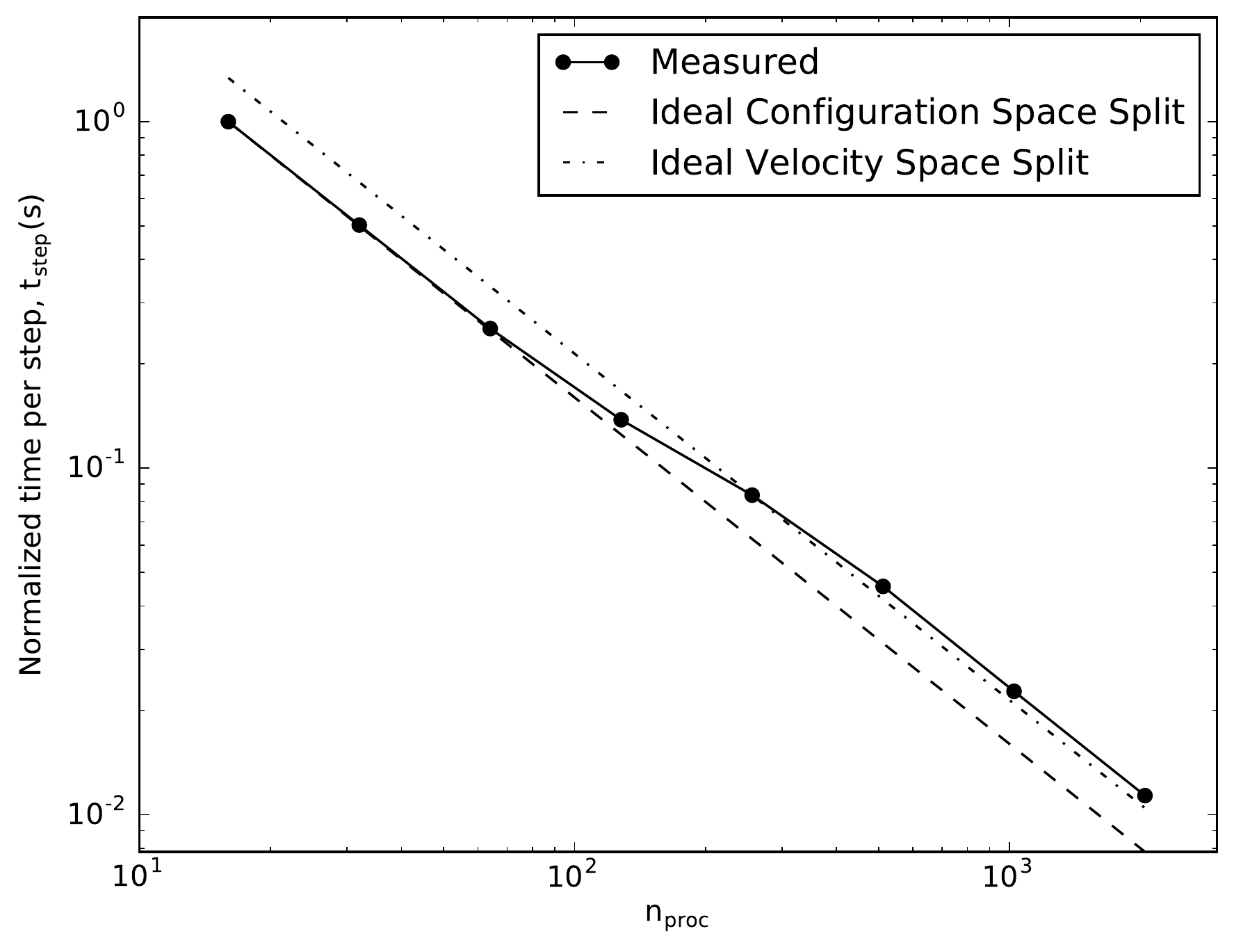}
\includegraphics[width=0.5\linewidth]{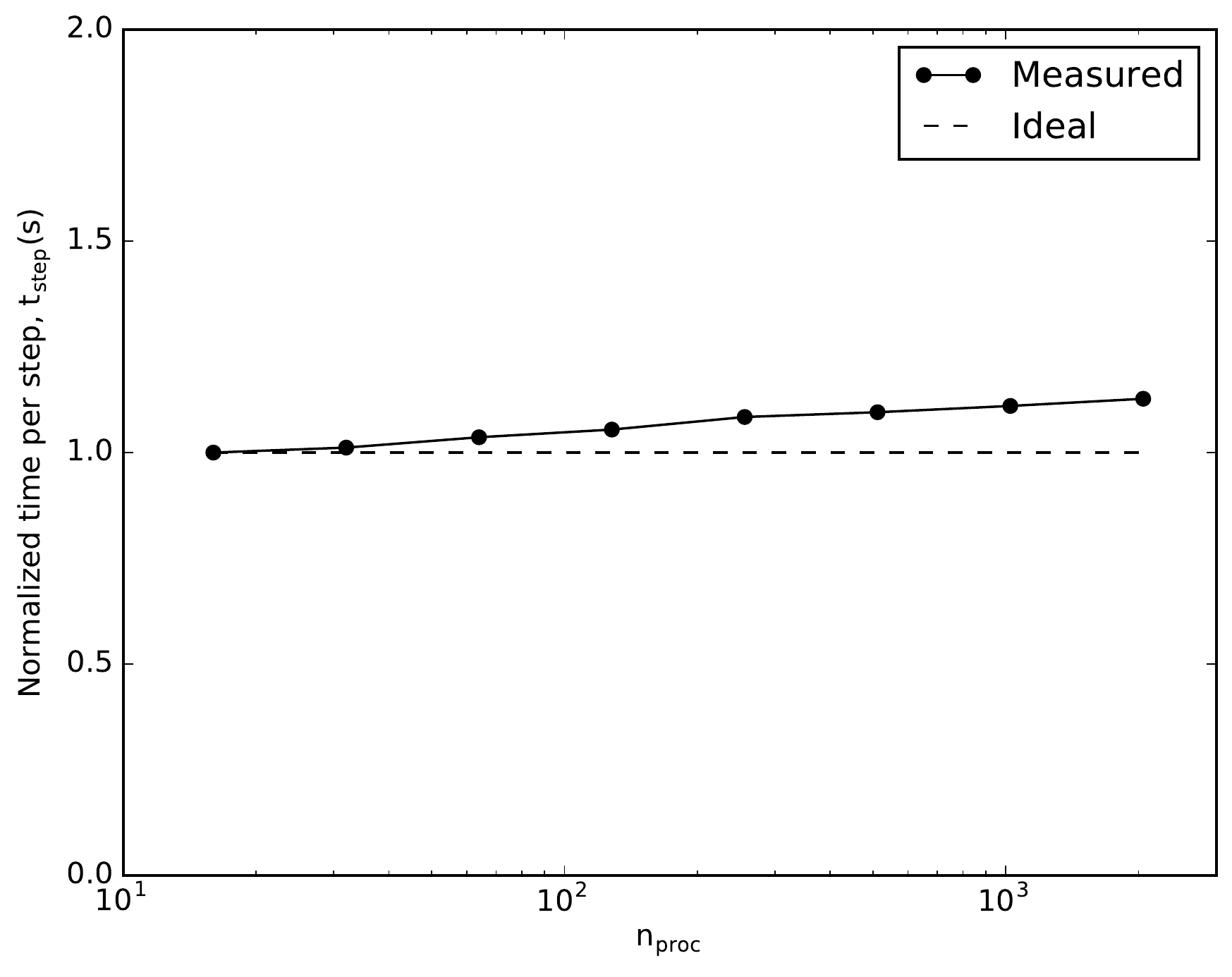}
}
\caption{Strong (left) and weak (right) scaling of the Vlasov-Maxwell solver for a prototypical 1X3V simulation. The y-axes are normalized to the time of the first simulation in the scaling, i.e., the $256 \times 16^3$ box on 16 cores for the strong scaling study, and the $64 \times 8^3$ box on 16 cores for the weak scaling study.} 
\label{fig:ScalingResults}
\end{figure}
\end{par}
\subsection{Electrostatic applications}
\begin{par}
Our first tests of the algorithm will be on systems for which the magnetic field is not dynamically important. While we solve Maxwell's equations, Eqns.\thinspace(\ref{eq:dbdt})-(\ref{eq:divB}), in their entirety, it is not incorrect to draw analogy between the tests presented below, and the so-called Vlasov-Ampere system, where Maxwell's equations reduce to,
\begin{align}
  \pfrac{\mvec{E}}{t} = \frac{\mvec{J}}{\epsilon_0}.
\end{align}
For all the following benchmarks, the distribution function is initialized as a drifting Maxwellian,
\begin{align}
f_s (\mvec{x}, \mvec{v}, t = 0) = \frac{n_s(\mvec{x})}{\sqrt{2 \pi v_{th_s}(\mvec{x})^2}} \exp\left(-\frac{|\mvec{v} - \mvec{u}_s(\mvec{x})|^2}{2 v_{th_s}(\mvec{x})^2} \right), \label{eq:initMaxwell}
\end{align}
where $n_s(\mvec{x})$ is the number density defined in \eqr{\ref{eq:densityDefinition}}, and  $\mvec{u}_s(\mvec{x}), v_{th_s}(\mvec{x})$ are the flow and thermal velocity of species $s$ respectively. The flow naturally follows from \eqr{\ref{eq:fluxDefinition}},
\begin{align}
\mvec{u}_s(\mvec{x}) = \frac{\langle \mvec{v} \rangle_s}{n_s(\mvec{x})},
\end{align}
and the thermal velocity can be compute from the second moment of the distribution function,
\begin{align}
v_{th_s}(\mvec{x})^2 = \frac{\langle |\mvec{v}|^2 \rangle_s}{V_{dim} n_s(\mvec{x})} - |\mvec{u}_s(\mvec{x})|^2,
\end{align}
where the constant $V_{dim}$ is determined by the number of velocity dimensions of the simulation, either one, two, or three. Alternatively, the thermal velocity can be defined in terms of a specified temperature,
\begin{align}
v_{th_s} = \sqrt{\frac{T_s}{m_s}},
\end{align}
where $T_s$ is the temperature of the species.
\end{par}
\subsubsection{Momentum and energy conservation tests}\label{sec:ConservationTests}
\begin{par}
We first consider a simple problem to demonstrate numerically the conservation laws derived earlier. We will initialize a distribution function with strong asymmetric flows to determine how well momentum and energy are conserved in our implementation of the energy, but not momentum, conserving DG algorithm. In one configuration space dimension ($x$) and one velocity space dimension ($v_x$), or 1X1V, we initialize a Maxwellian distribution function with density
\begin{align}
n(x, t = 0) & = n_0 (1 + \exp(-\beta_l (x - x_m)^2)) \qquad x < x_m, \notag \\
& =  n_0 (1 + \exp(-\beta_r (x - x_m)^2)) \qquad x>x_m \label{eq:conservationTestDensityProfile},
\end{align}
for both the protons and electrons. The length of the box is $L_x = c/\omega_{pp} = d_p$, where $\omega_{pp} = \sqrt{n_0 e^2/\epsilon_0 m_p}$ is the proton plasma frequency. We choose $x_m = 0.25 L_x$, $\beta_l = 0.75/d_e^2, \beta_r = 0.075/d_e^2$, and $n_0 = 1.0$, where $d_e = d_p\sqrt{m_e/m_p}$ is the electron inertial length. The drift is given by $u(x, t =0) = v_{th_e}$ for both species. Other relevant parameters are $v_{th_e}/c = 0.1$, $m_p/m_e = 1836$, and $T_p/T_e = 1.0$. The velocity space extents for the electrons are $\pm 8 v_{th_e}$ and the proton velocity space extents are $v_{th_e} \pm 8 v_{th_p}$. All simulations are run with $N_v = 64$. A number of CFL conditions and configuration space resolutions are employed to determine the convergence order of our time-stepping scheme and the convergence of the momentum errors. The boundary conditions in configuration space are periodic, while the boundary conditions in velocity space are zero flux to guarantee that there are no boundary losses for these tests of energy and momentum conservation. While it may appear that the density profile in \eqr{\ref{eq:conservationTestDensityProfile}} is not periodic, the values of $\beta_l$ and $\beta_r$ localize the density profile around $x=x_m$ since the profile varies on electron scales, while the domain size is with respect to proton scales. This scale separation reduces the violation of periodic boundary conditions to a vanishingly small number, $\mathcal{O}(10^{-34})$. The simulations are run to $1000 \pi \omega_{pe}^{-1}$ where $\omega_{pe}=\omega_{pp}\sqrt{m_p/m_e}$ is the electron plasma frequency. This length of time corresponds to approximately ten thousand time steps for the runs with the largest time steps. Results are plotted in Figure \ref{fig:ConservationResults}. Since we are either halving the time-step or doubling the resolution, we define the order of convergence as
\begin{align}
\mathcal{C}(\mathcal{E}_1, \mathcal{E}_2) = \frac{\log(\mathcal{E}_1) - \log(\mathcal{E}_2)}{\log(2)},
\end{align}
where $\mathcal{E}_i$ is the relative change of a quantity, either momentum or energy, in a particular simulation. In a set of simulations which are ``converging,'' $\mathcal{E}_2 < \mathcal{E}_1$ so the convergence order is a positive number. We note that the energy conservation result can be obtained by either decreasing the size of the time step in an $N_x = 8$ simulation, as shown, or by refining the grid. The energy errors for $N_x = 16, 32, 64$ are roughly equivalent to the energy errors shown for the smaller time-step. For the energy conservation test, we find the orders of convergence are 2.12, 2.70, and 2.75 for each successive halving of the CFL number for an $N_x = 8$ simulation. These orders of convergence are comparable with the expected third order scaling from our third order strong-stability-preserving Runge-Kutta method. For the momentum conservation test, we find the convergence orders of the polynomial order 2 simulations are 1.46, 2.27, and 3.76, and the convergence orders of the polynomial order 3 simulations are 1.99, 3.08, and 6.04, for a refining of the grid from $N_x = 8$ to $N_x = 64$ cells in configuration space. We can thus say the momentum conservation tests demonstrates super-linear scaling, i.e., greater than $p$, for continual refinement of the grid. 
  \begin{figure}
        \subfloat{\includegraphics[width=0.5\linewidth]{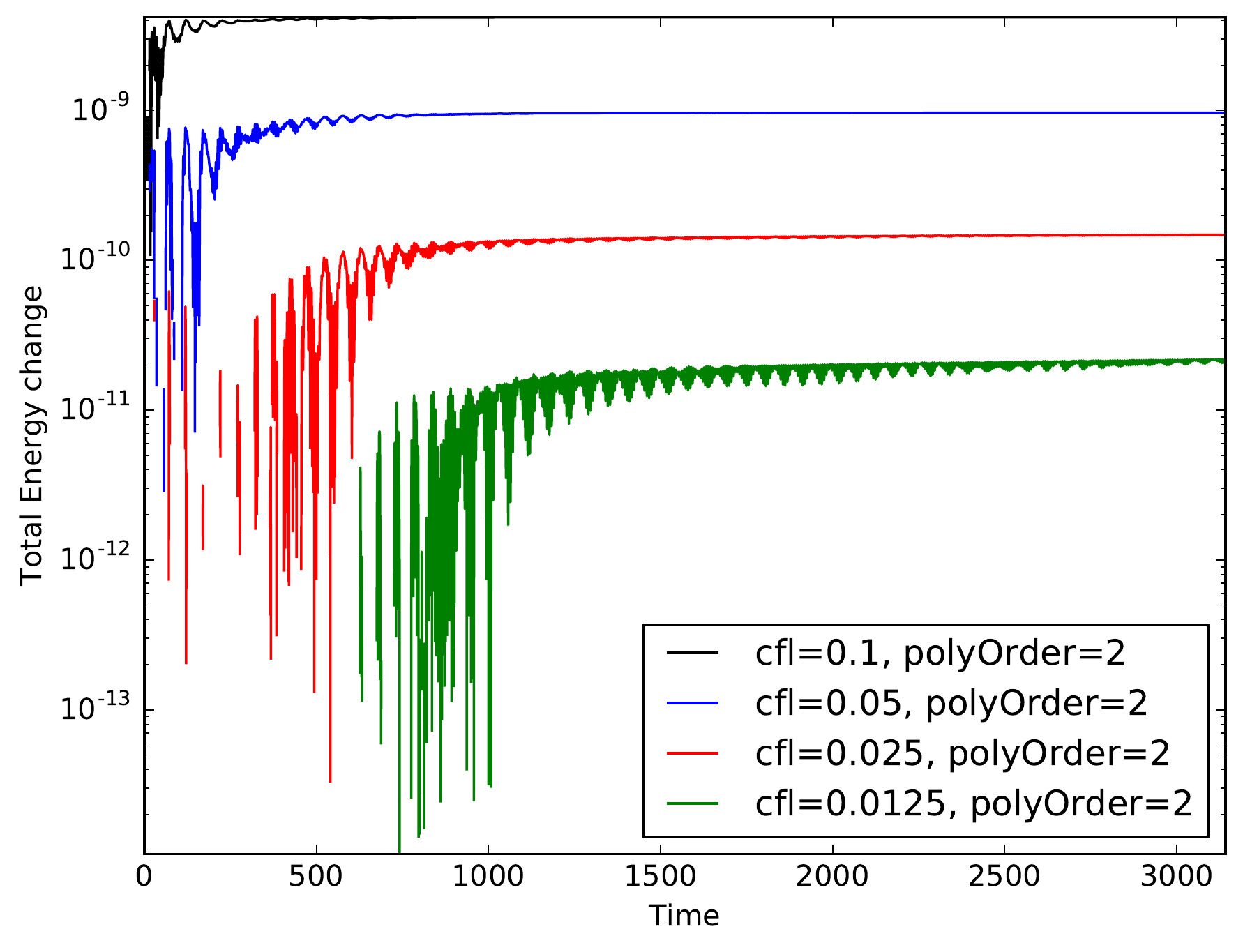}}
        \subfloat{\includegraphics[width=0.5\linewidth]{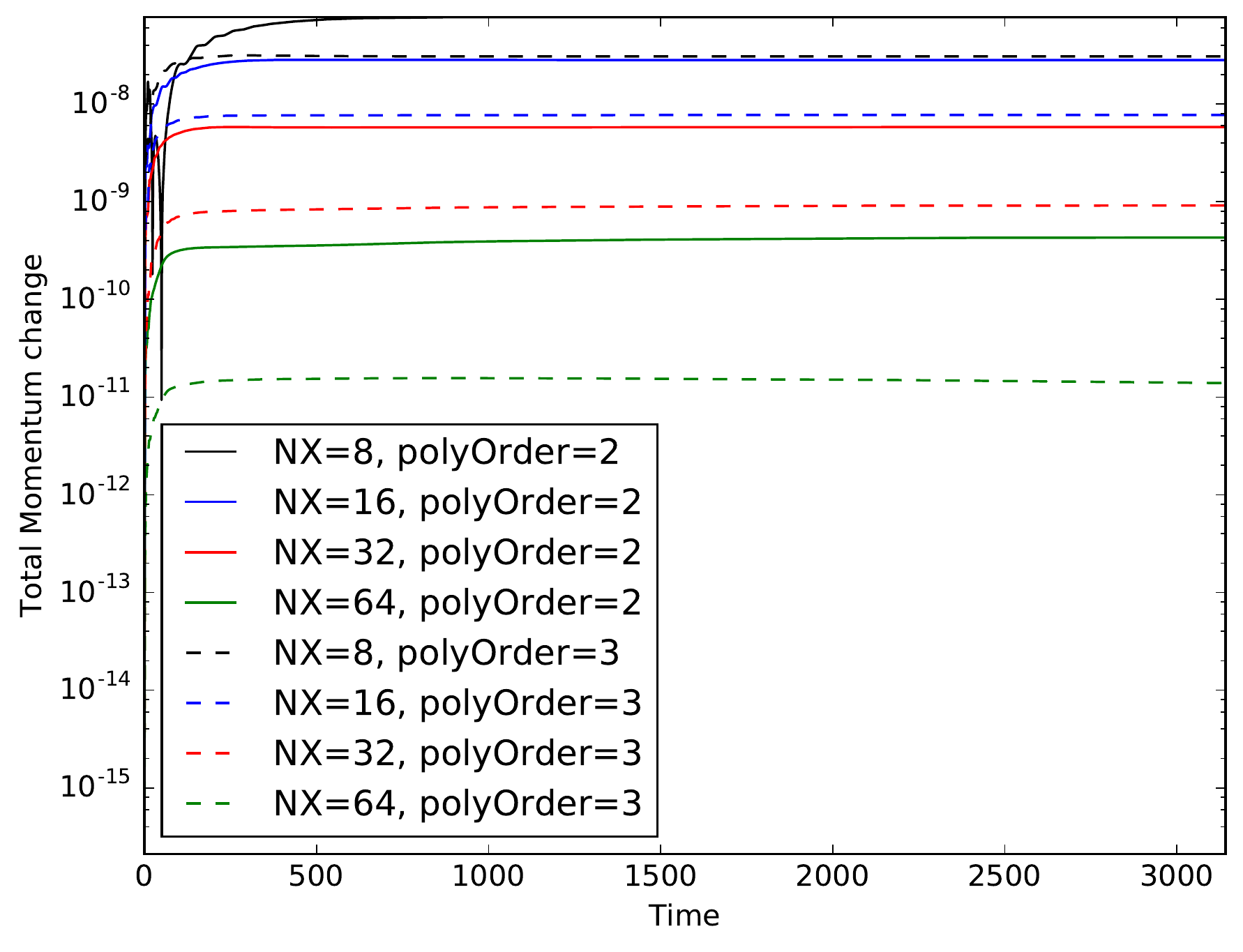}}
        \caption{The left plot is the relative change in the total (proton, electron, and electromagnetic) energy $\Delta E = (E(t) - E(t=0))/E(t=0)$ for an $8 \times 64$, polynomial order 2 simulation with gradually more restrictive CFL conditions. The orders of convergence are 2.12, 2.70, and 2.75 for each successive halving of the CFL number. The right plot is the relative change in the total (proton and electron) momentum $\Delta p_x = (p_x(t=0)-p_x(t))/p_x(t=0)$ for a set of simulations varying polynomial order and configuration space resolution. The convergence orders of the polynomial order 2 simulations are 1.46, 2.27, and 3.76, and the convergence orders of the polynomial order 3 simulations are 1.99, 3.08, and 6.04. We note the sign change in the definition of the relative momentum change compared to the relative energy change. While the energy monotonically decays for the entire length of the run, the momentum errors manifest as an increase in the total momentum.}
        \label{fig:ConservationResults}
  \end{figure}
\end{par}
\begin{par}
We can also use this same initial condition to examine the behavior of quantities we expect to decay, or grow, throughout the simulation, such as the $L^2$ norm of the distribution function and the discrete entropy $S_h = -f_h \ln(f_h)$. However, we should point out that a crucial assumption in our proof of the growth of the discrete entropy was that the distribution function $f_h$ remain positive definite. Due to the nature of our DG scheme, this cannot be guaranteed. In fact, for highly resolved grids, while the positivity violations decrease in magnitude, they cannot be completely eliminated. The reason for this is the magnitude of the distribution function is so small at the edges of velocity space, that the generation of round-off error positivity violations is inevitable. Fortunately, as can be seen in Figure \ref{fig:ConservationResults}, small positivity violations do not affect the stability or robustness of our algorithm. This is not surprising, as even moderate positivity violations are tolerable so long as the moments of the distribution function remain positive. It would take substantial negativity in the distribution function to drive moments such as the density or energy negative and create the same numerical stability issues often found in fluid codes. Nonetheless, since we would like to at least comment on a quantity such as the discrete entropy even with the presence of positivity violations, we have run the same initial conditions with more resolution and defined an ``entropy-like'' quantity, $S_h^* = -|f_h| \ln(|f_h|)$, which we expect will behave similarly to the discrete entropy in the limit that positivity violations go to zero. We have plotted the $L^2$ norm of the distribution function and $S_h^*$ integrated over all of phase space for both protons and electrons in Figure \ref{fig:EntropyL2Plot} for highly resolved grids, with polynomial order 2, demonstrating, at least approximately, the behavior we expect based on Proposition \ref{prop:discrete-L2-norm} and Corollary \ref{cor:discrete-entropy}.
\begin{figure}
    \begin{center}
        \subfloat{\includegraphics[width=0.5\linewidth]{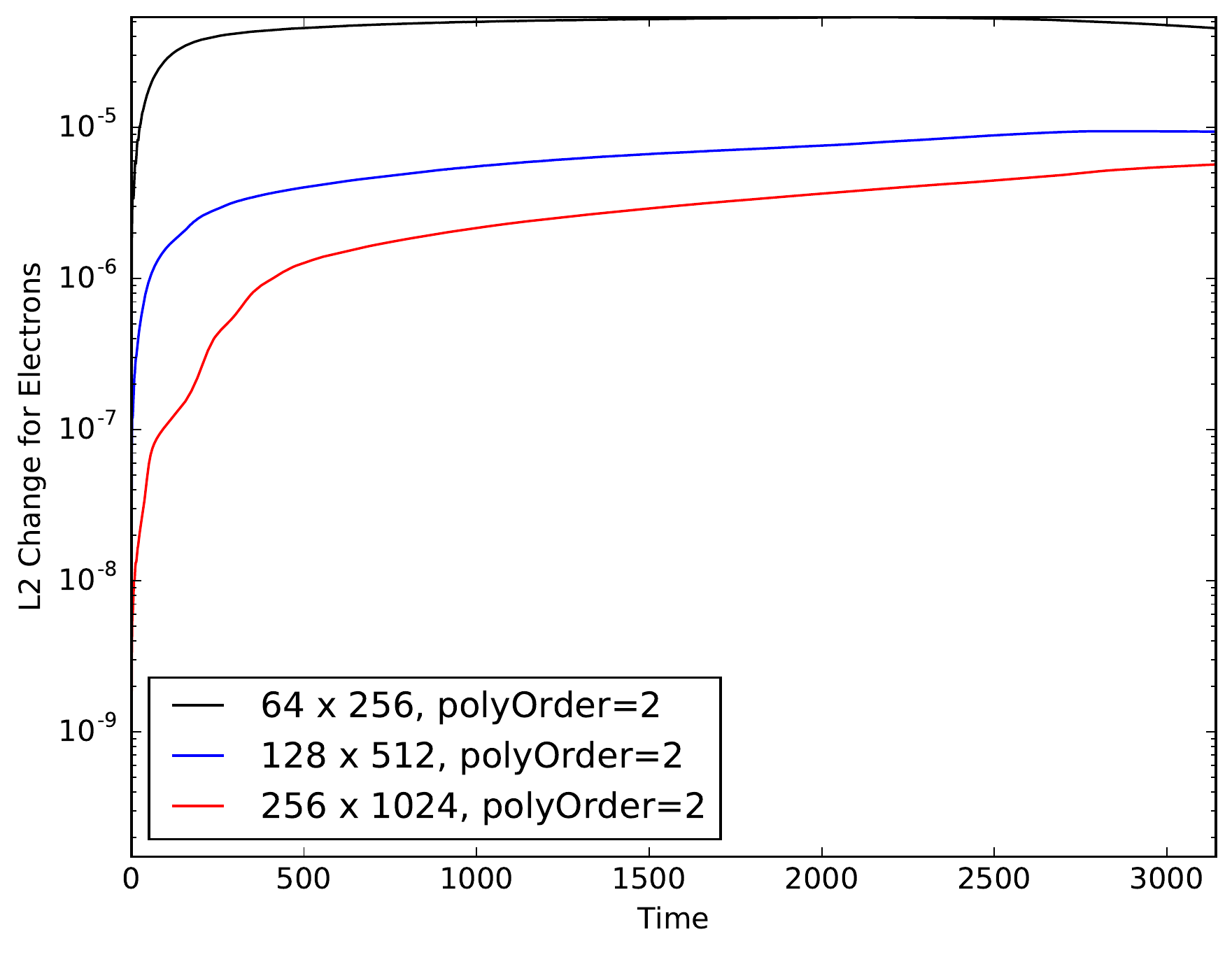}}
        \subfloat{\includegraphics[width=0.5\linewidth]{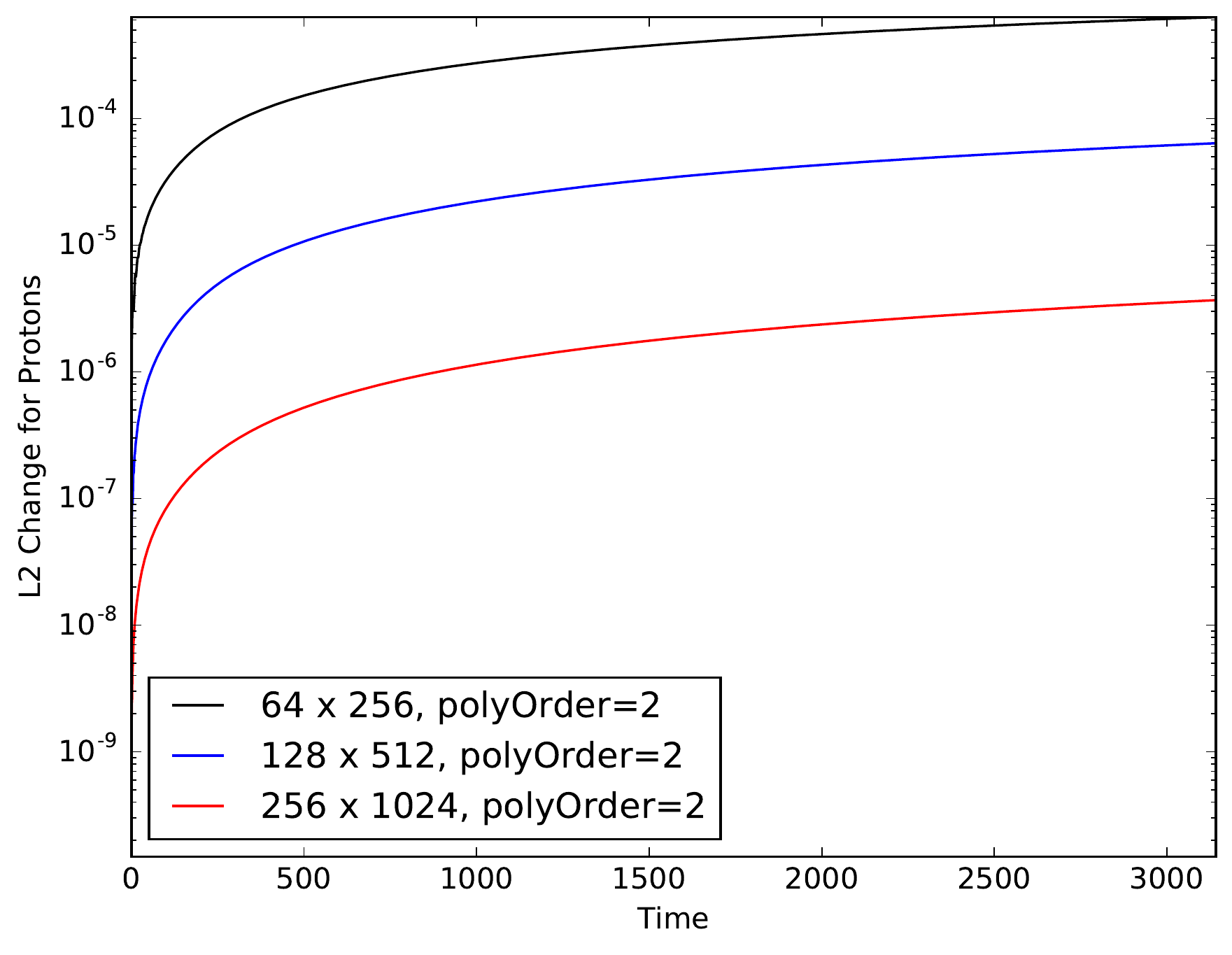}}\\
        \subfloat{\includegraphics[width=0.5\linewidth]{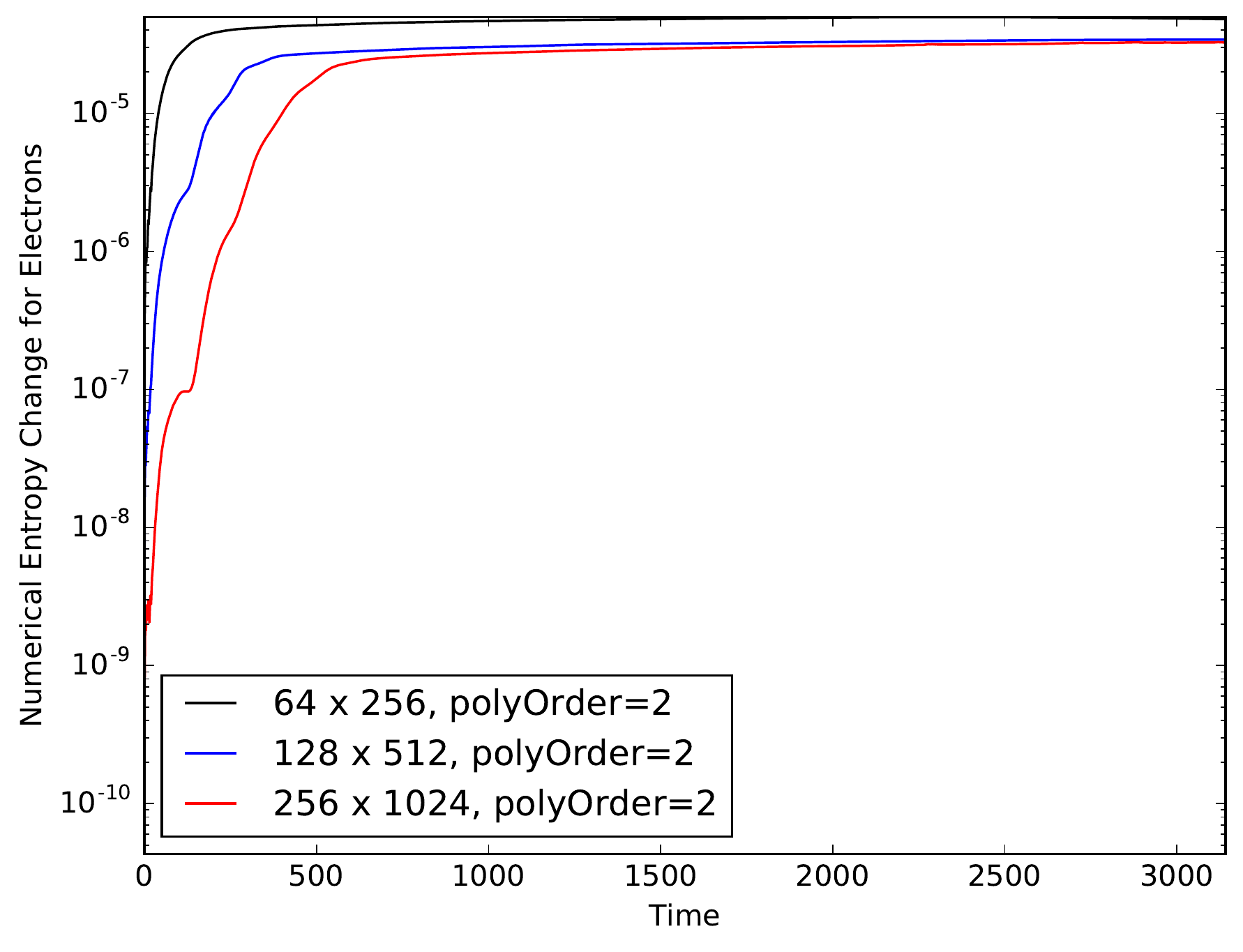}}
        \subfloat{\includegraphics[width=0.5\linewidth]{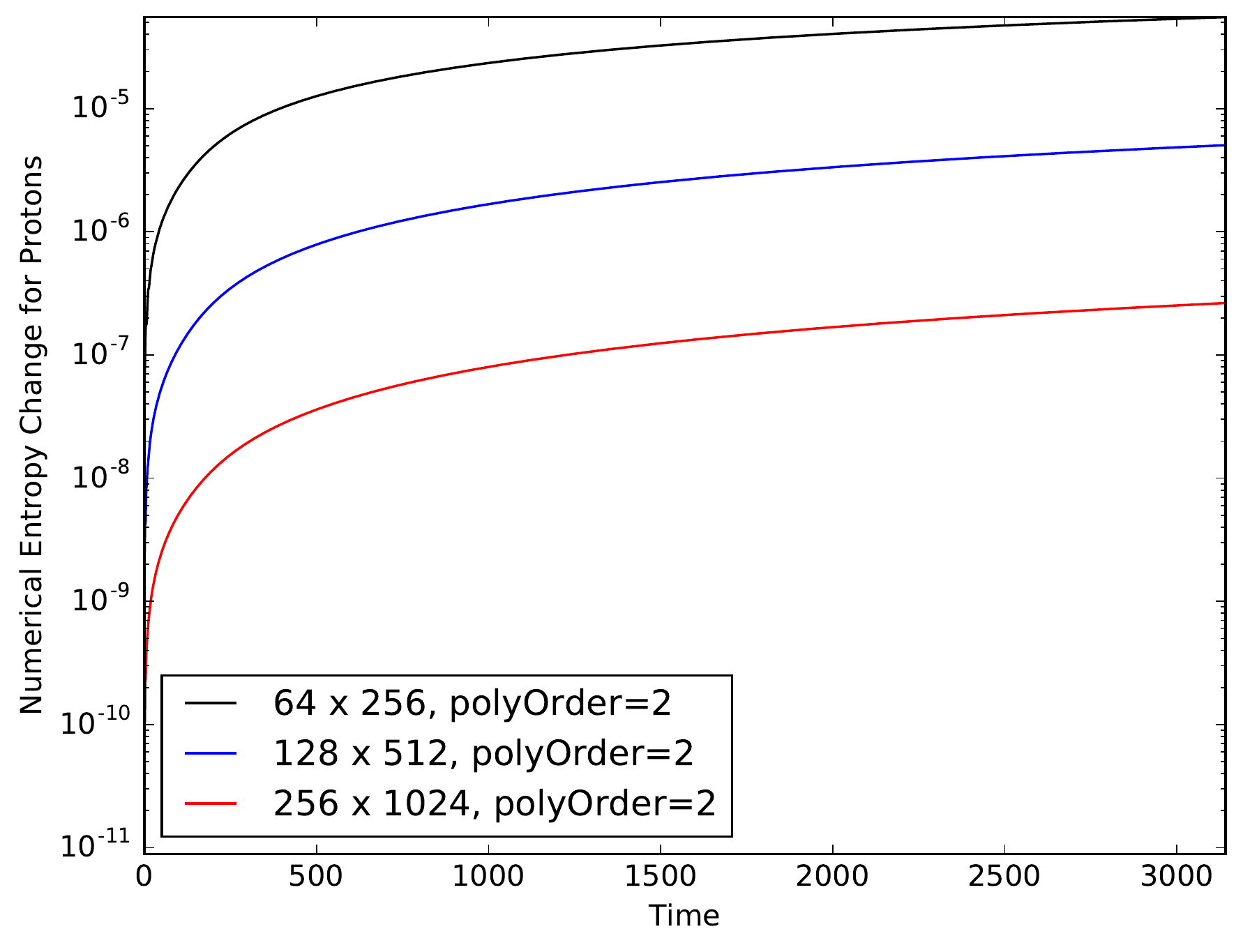}}        
    \end{center}
        \caption{The $L^2$ norm for electrons (top left), protons (top right), and modified entropy variable $S_h^*$ for electrons (bottom left), and protons (bottom right). Similar to how we plotted energy and momentum conservation in Figure  \ref{fig:ConservationResults}, we have plotted the relative change in the $L^2$ norm and this modified entropy $S_h^*$ from their starting value, with the caveat that the $L^2$ norm is a monotonically decreasing quantity, while our modified entropy variable is a monotonically increasing quantity, so the sign of the relative change is different for the two.}
        \label{fig:EntropyL2Plot}
\end{figure}   
\end{par}
\begin{par}
Finally, we reiterate that our discretization of Maxwell's equations is not guaranteed to preserve the divergence relations, in this particular case \eqr{\ref{eq:divE}}, and thus we would like to examine how well this simple initial condition preserves this constraint on the electric field. Since our system only has one dimension in configuration space, in a discrete sense we are concerned with how well the relation,
\begin{align}
\pfrac{E_{x_h}(x)}{x} - (n_{p_h}(x) - n_{e_h}(x)) = 0, \label{eq:divEError}
\end{align}
is satisfied at each point in configuration space. We have plotted \eqr{\ref{eq:divEError}} in Figure \ref{fig:divEErrorComp}, normalized to the maximum value of the divergence of the electric field, at $1000 \pi \omega_{pe}^{-1}$ for a number of resolutions in configuration space with polynomial order 2. We have also plotted the value of \eqr{\ref{eq:divEError}}, integrated over all of configuration space, versus time, in Figure \ref{fig:divEErrorComp}. The $x$ derivative of the electric field is computed using the analytic derivatives of the basis function expansion within a cell. The above normalization is somewhat arbitrary, but it does illustrate that the error in the divergence constraint is small compared to the largest variation in the electric field, and thus the errors in the divergence relation are sub-dominant to the dynamics. Further, it is a comfort to note that the errors decrease dramatically with increasing resolution, even given the complicated coupling at play since one of the sources of the divergence errors comes from the computation of the moments of the distribution function, and the charge density $\rho_{c_h}(x) = n_{p_h}(x) - n_{e_h}(x)$ does not appear anywhere in the equation system we are explicitly evolving.
\begin{figure}
    \begin{center}
        \subfloat{\includegraphics[width=0.5\linewidth]{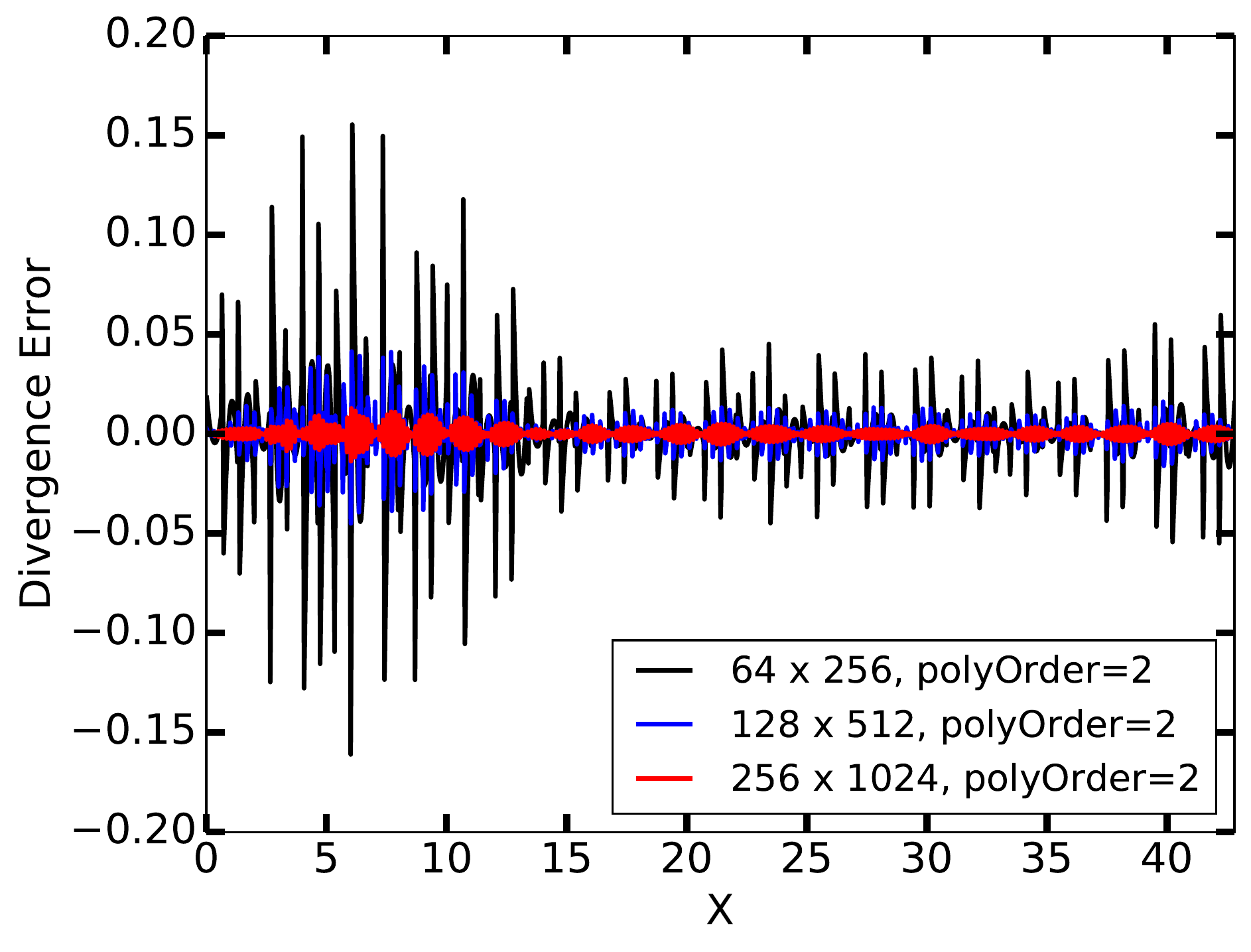}}
        \subfloat{\includegraphics[width=0.5\linewidth]{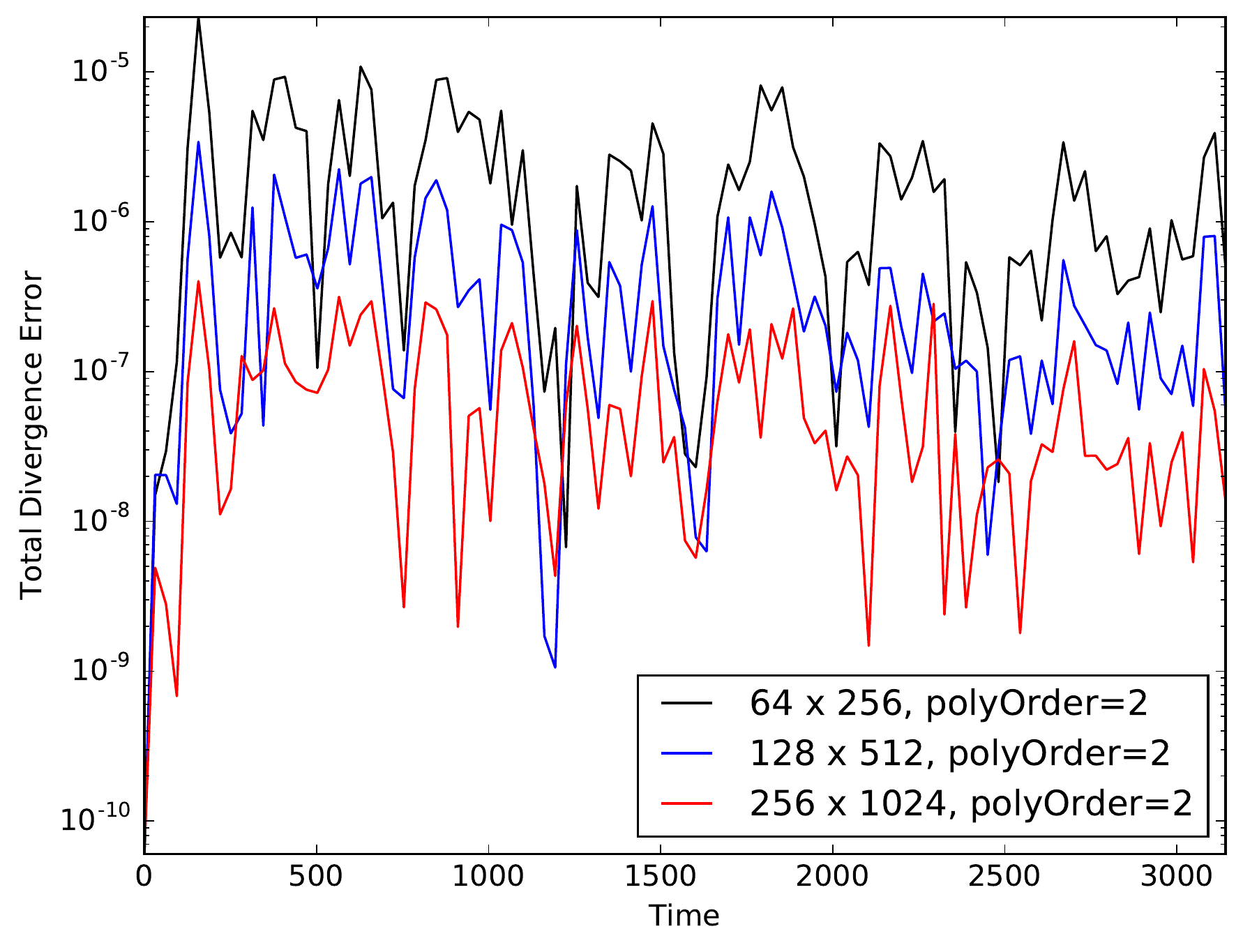}}
    \end{center}
        \caption{\eqr{\ref{eq:divEError}}, normalized to the maximum value of $\partial E_x/\partial x$, at $t = 1000 \pi \omega_{pe}^{-1}$ (left) and the integral of \eqr{\ref{eq:divEError}}, with no normalization, over all of configuration space, in time (right). While the behavior of the divergence errors in time is not monotonic like other errors, such as the energy conservation errors in Figure \ref{fig:ConservationResults} or the $L^2$ error in Figure \ref{fig:EntropyL2Plot}, it is nonetheless decreasing with increasing resolution, a comfort given the complicated interplay between the sources calculated from the Vlasov equation and our solver for Maxwell's equations.}
        \label{fig:divEErrorComp}
\end{figure}   
\end{par}
\subsubsection{Landau damping of Langmuir waves}
\begin{par}
Consider a plasma, or Langmuir, wave propagating in a plasma of protons and electrons whose distribution functions are given by Maxwellians. Langmuir waves are dispersive waves, with a dispersion relation given by
\begin{align}
1 - \frac{1}{2 k^2 \lambda_{De}^2} Z'\left(\frac{\omega}{\sqrt{2} v_{th_e} k} \right ) = 0, \label{eq:plasmaDisp}
\end{align} 
in the limit that the proton mass is much larger than the electron mass and the protons can thus be considered immobile. Here, $v_{th_e}$ is the electron thermal velocity, $\omega_{pe}$ is the electron plasma frequency, and $\lambda_{De} = v_{th_e}/\omega_{pe}$ is the electron Debye length. $Z(\zeta)$ is the plasma dispersion function, defined as
\begin{align}
Z(\zeta) = \frac{1}{\sqrt{\pi}} \int_{-\infty}^\infty \frac{e^{-x^2}}{x - \zeta} dx,
\end{align}
with the derivative of the plasma dispersion function given by
\begin{align}
Z'(\zeta) = -2(1 + \zeta Z(\zeta)).
\end{align}
An application of complex integration techniques shows that depending on the sign of the largest imaginary component of the frequency $\omega = \omega_r + i \gamma$, the wave is either unstable and will grow with time, or will damp away, a phenomenon known as Landau damping. For Langmuir waves propagating in a Maxwellian plasma of protons and electrons, the waves quickly damp away. Using a 1X1V setup, we can initialize Langmuir waves in the Vlasov-Maxwell system with a small density perturbation and the corresponding electric field to support this density perturbation,
\begin{align}
n_e(x) & = n_0(1 + \alpha \cos(kx)) \\
n_p(x) & = n_0 \\
E_x(x) & = -|e| \alpha \frac{\sin(kx)}{\epsilon_0 k},
\end{align}
where $n_0 = 1.0$, $\alpha$ is the size of the perturbation, and $k$ is the wavenumber of the wave. The electric charge $e$ and permittivity of free space $\epsilon_0$ are included in the electric field to satisfy \eqr{\ref{eq:divE}}. Choosing $\alpha \ll 1$ allows us to compare with the analytic theory described above. The box size is set to $L_x = 2\pi/k$ so exactly one wavelength fits in the domain. Specific parameters for these runs are: $\alpha = 10^{-4}$, $m_p/m_e = 1836$, $T_p/T_e = 1.0$, and  $v_{th_e}/c = 0.1$. For the proton species, the velocity space extents are $\pm 6 v_{th_p}$, and for the electrons, the velocity space extents are $\pm 6 v_{th_e}$. The boundary conditions in configuration space are periodic, while the boundary conditions in velocity space are zero flux to preserve density and energy conservation in the semi-discrete system. The resolution is chosen for each simulation to adequately resolve the Debye length in configuration space and to mitigate numerical recurrence in velocity space. By numerical recurrence, we refer to the process by which the collisionless system artificially ``un-mixes'' if the distribution function forms structure at the velocity space grid scale. Numerical recurrence is inevitable with finite velocity resolution for this particular problem as the Landau damping of the wave will create smaller and smaller velocity space structure through the phase-mixing of the wave. We could completely eliminate this issue with a diffusive process in velocity space, such as a collision operator. Here, we choose ample velocity resolution so that the wave damps enough for us to extract a clean damping rate and frequency for the wave initialized. We find for the longest wavelengths, using polynomial order 2, a resolution of 64 points in configuration space adequately resolves the Debye length, and 128 points in velocity space permits the wave to phase-mix sufficiently to extract damping rates. The evolution of the electromagnetic energy, as well as the other components of the energy, in a prototypical simulation is given in Figure~\ref{fig:LangmuirWaveEnergy}. Comparisons of a number of Vlasov-Maxwell simulations with theory for both the damping rates and the frequencies of the waves are given in Figure~\ref{fig:LangmuirWaveTheory}. For the theory, we solve \eqr{\ref{eq:plasmaDisp}} using a root-finding technique. We emphasize that we solve the Vlasov-Maxwell system in its entirety, including the nonlinear term, and for both a proton and electron species. With the above simulation parameters, the plasma waves damp entirely on the electron species, so the approximation that the protons are essentially immobile in our dispersion relation holds to a high precision. We also wish to note that the resolution of 64 points in configuration space is not required for every simulation. For example, the prototypical simulation presented in Figure~\ref{fig:LangmuirWaveEnergy} uses only 16 points in configuration space, or approximately one grid cell per Debye length. As long as the gradients are properly resolved, the Vlasov-Maxwell discretization is extremely robust. Subsequent simulations have a large amount of variation in their resolution compared to the Debye length. While Section \ref{sec:ESShock} uses a similar resolution of approximately one grid cell per Debye length, a grid cell in the simulation presented in Section \ref{sec:OT} contains 40 Debye lengths. 
  \begin{figure}
    \begin{center}
        \subfloat{\includegraphics[width=0.5\linewidth]{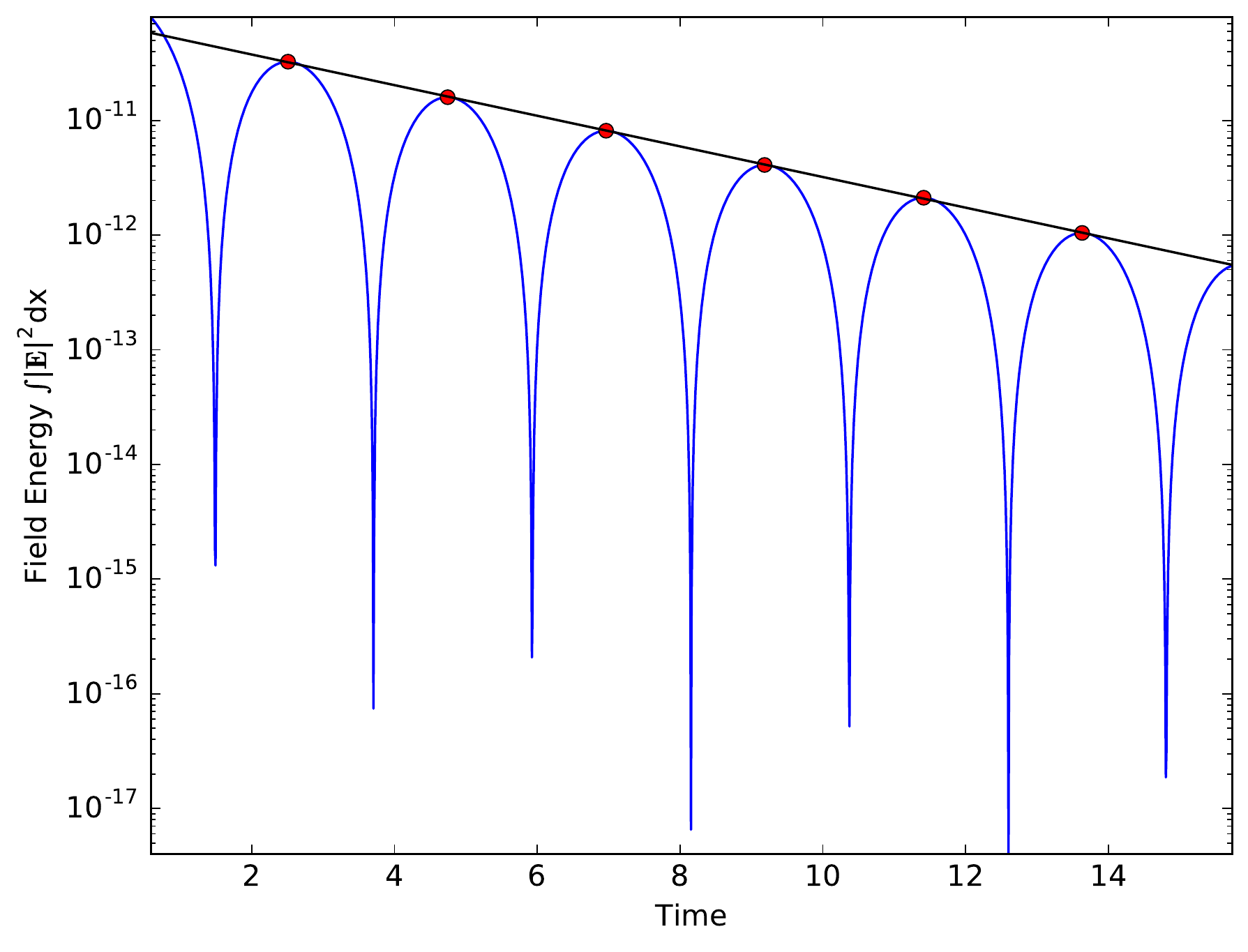}}
        \subfloat{\includegraphics[width=0.5\linewidth]{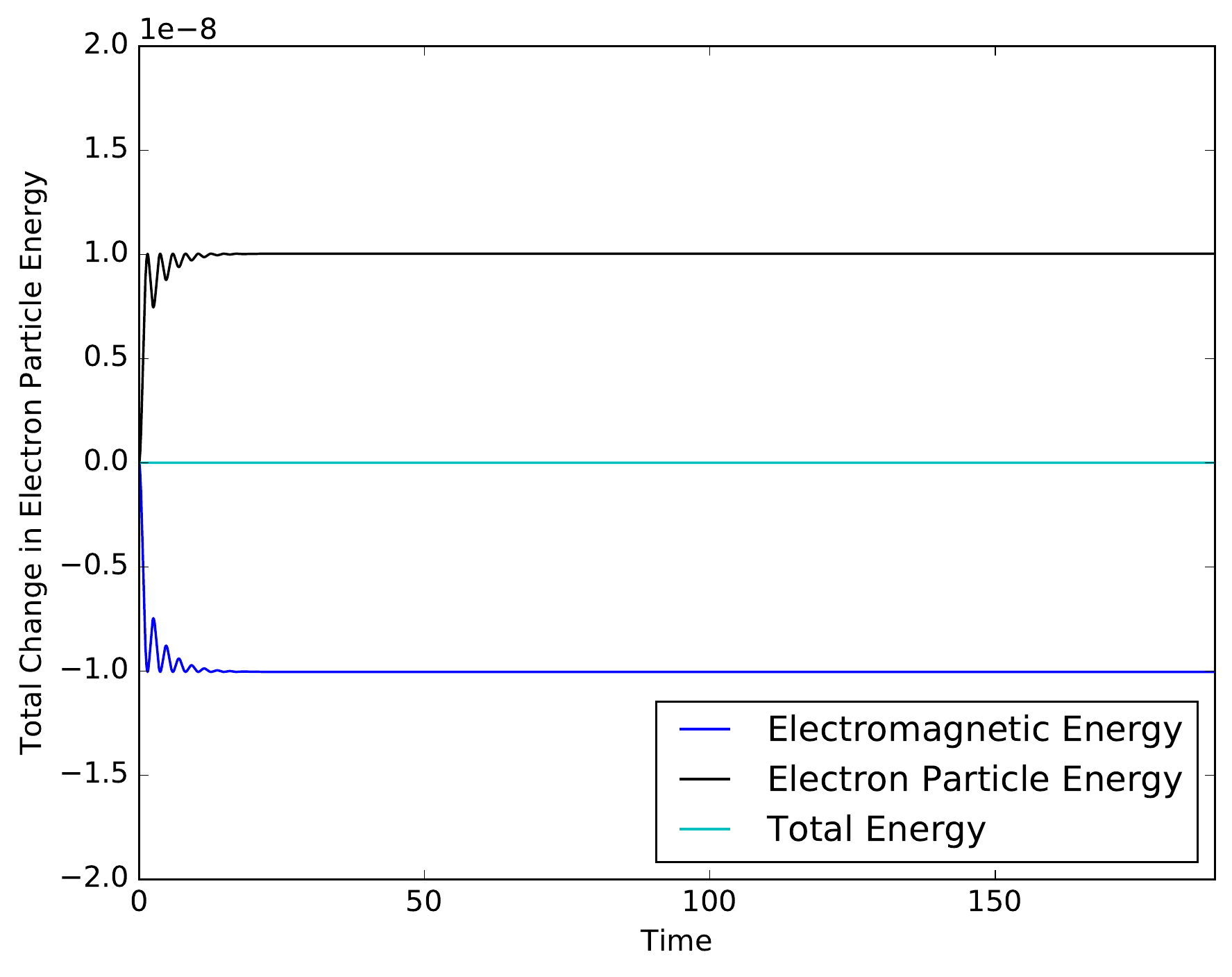}}
     \end{center}
        \caption{Prototypical evolution of the electromagnetic energy, $\frac{\epsilon_0}{2} \int |\mvec{E}|^2 dx$, for the damping of a Langmuir wave, in this case $k \lambda_D = 0.5$ for a number of plasma periods (left), and evolution of various components of the energy for the full length of the simulation (right). The right plot is the relative change in the energy component compared to the total energy at $t=0$, i.e, $\Delta E_{comp}/E_0$. The local maxima (red circles) of the evolution in the left plot are used to determine both the damping rate and frequency of the excited wave via linear regression. We note that energy is very well conserved and, as expected, the plasma waves damp on the electrons, converting electromagnetic energy to electron thermal energy.}
        \label{fig:LangmuirWaveEnergy}
  \end{figure}
  \begin{figure}
        \subfloat{\includegraphics[width=0.5\linewidth]{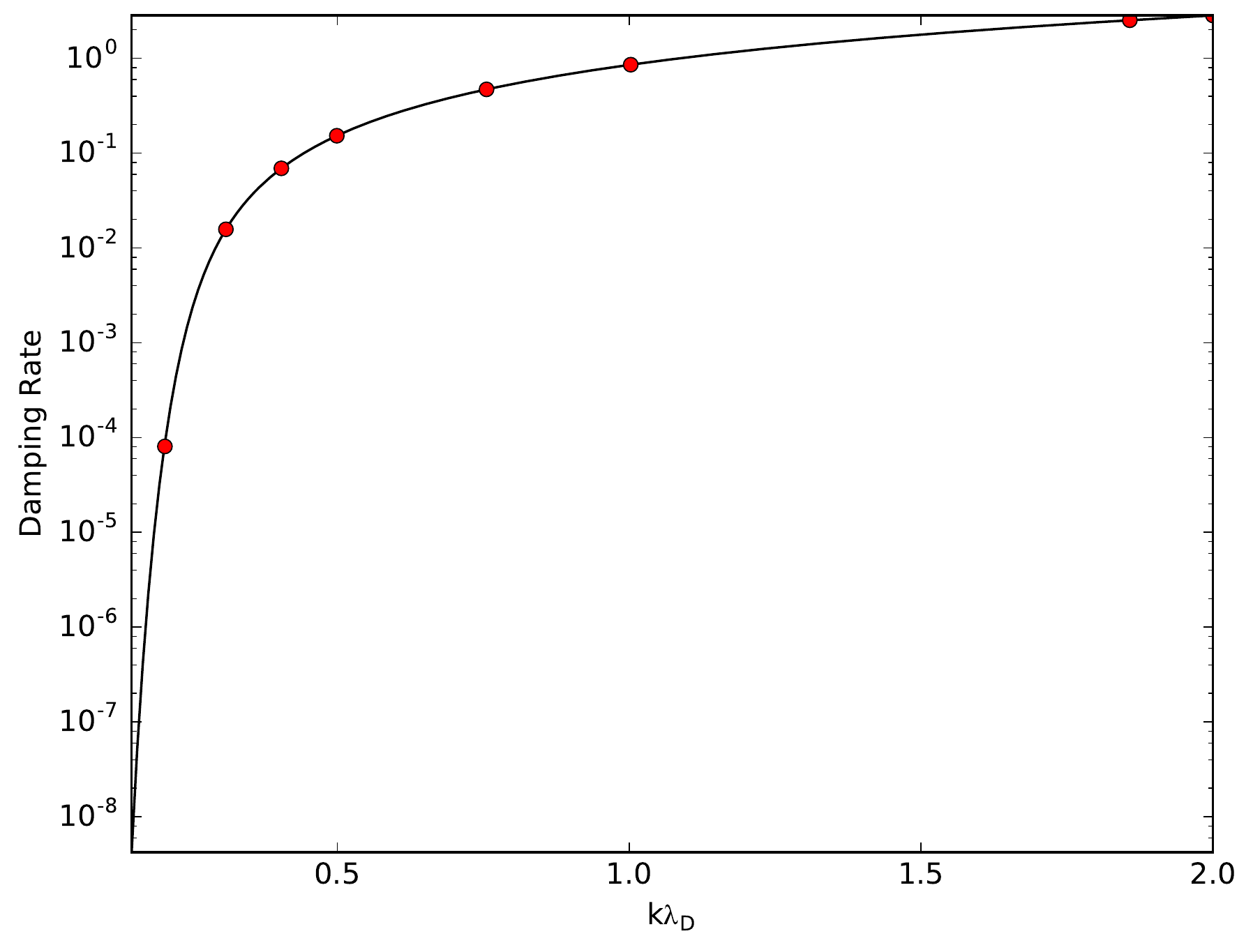}}
        \subfloat{\includegraphics[width=0.5\linewidth]{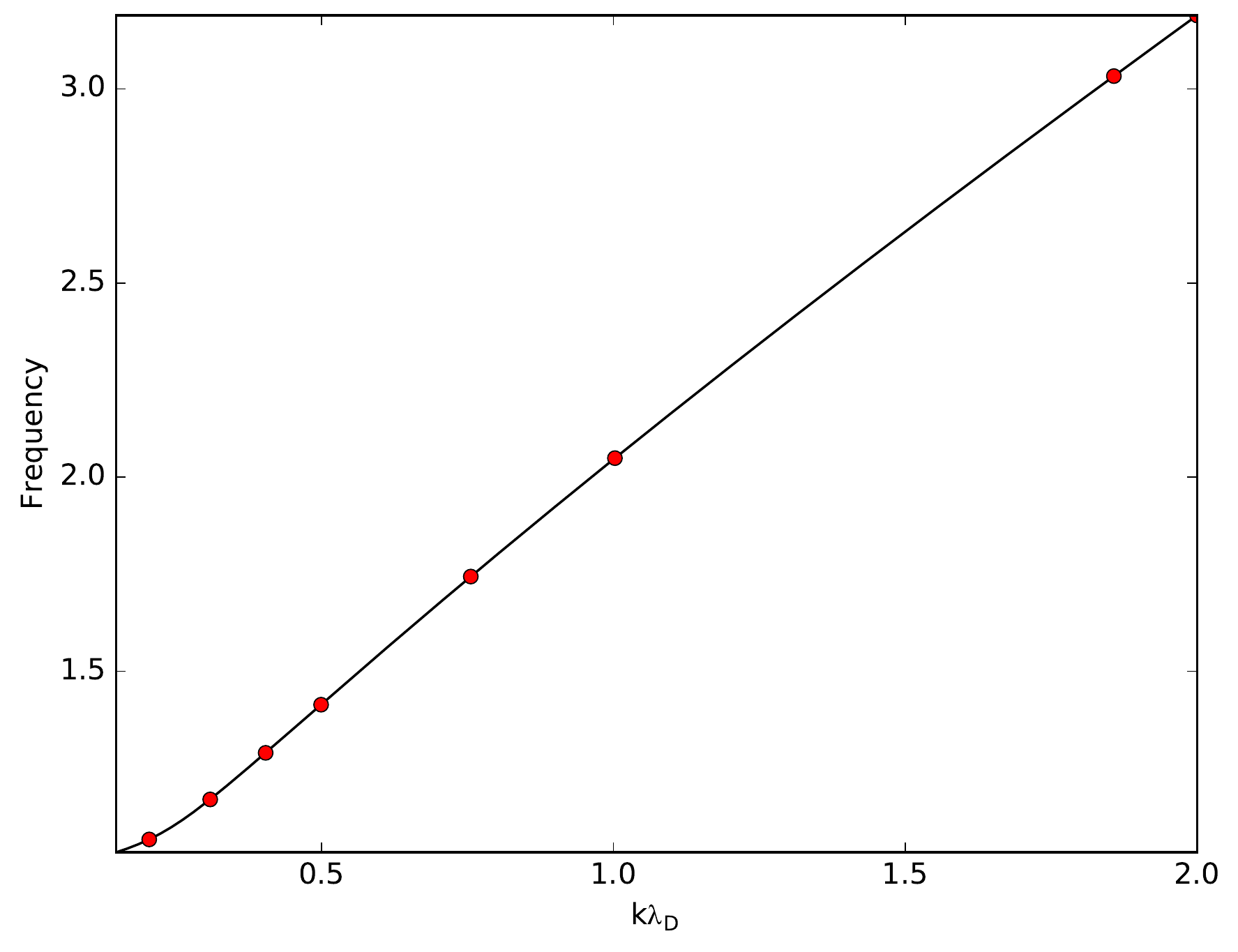}}
        \caption{Damping rates (left) and frequencies (right) of Langmuir waves from theory (solid line) and for a number of Vlasov-Maxwell simulations (red circles). The solid lines are obtained using a root finding technique applied to \eqr{\ref{eq:plasmaDisp}}. The x-axis is normalized to the Debye length $\lambda_D$ and the y-axis of the frequency plot is normalized to the plasma frequency $\omega_{pe}$.}
        \label{fig:LangmuirWaveTheory}
  \end{figure}
\end{par}
\subsubsection{Electrostatic shocks}\label{sec:ESShock}
\begin{par}
When two counter propagating, but otherwise identical, slabs of plasma collide at speeds greater than the proton sound speed $c_s = \sqrt{T_e/m_p}$, a potential well can be created, leading to trapping of electrons and the formation of shocks from nonlinear steepening of wave fronts. In the absence of collisions, shock formation only occurs up to a critical mach number $M_c = u/c_s \sim 3$, where $u$ is the drift velocity \citep{Forslund:1970}. To demonstrate that the collisionless Vlasov-Maxwell system behaves as expected, we perform two simulations, one with $M_c = 1.5$ and one with $M_c = 5$, and compare the results to analogous five moment two fluid simulations given by the algorithm described in Hakim 2006 \citep{Hakim:2006}. The parameters of the two simulations are as follows: $T_e/T_p = 9.0$, $m_p/m_e = 1836$, $v_{th_e}/c = 0.1$, and $L_x = 1000 \lambda_{De}$. The drifts of both the protons and electrons are equal, so as not to initialize any current. The velocity space extents of the electrons are $\pm 6 v_{th_e}$, while the proton velocity space grid is $[-|u_p| -10 v_{th_p}, |u_p| + 10 v_{th_p}]$, where $|u_p|$ is the magnitude of the specified drift. The proton velocity grid is made larger by the inclusion of the drift in its extents, but including the drifts ensures that the moments of the proton distribution function are calculated accurately. Including the drifts in the electron velocity space extents are unnecessary since the drifts are much smaller than the electron thermal velocity. The boundary conditions in velocity space are again zero flux, but in this case the configuration space boundary conditions are open to avoid any pollution of the region downstream of the shock due to waves launched by the initial condition propagating towards the wall. The resolution of the Vlasov-Maxwell simulations is $1024 \times 128$ with polynomial order 2. The left half of the domain is initialized with the distribution function propagating to the right, while the distribution function on the right propagates to the left. The drifting plasmas thus collide in the middle of the domain. At $t = 1000 \omega_{pe}^{-1}$, we note the following comparisons between the two Vlasov runs with the two fluid simulations in Figure~\ref{fig:FluidMach}. Essential features of the shock for $M_c = 1.5$ such as the mass density and momentum for both species compare well, while the particle energy, especially for the protons, shows poor agreement between the two models. In addition, while the $M_c = 5$ fluid simulation still clearly shocks, the broadening of the densities in the $M_c = 5$ Vlasov run demonstrates that no shock forms for this high of a Mach number in the collisionless system, in agreement with the aforementioned previous studies \citep{Forslund:1970}. Plots of the distribution function in Figure~\ref{fig:ESShockDistributionFunction} further illustrate some of the key differences between the fluid and kinetic models. For example, the protons in the $M_c = 1.5$ are very far from thermal equilibrium, having experienced a significant amount of particle acceleration, thus easily explaining why the proton particle energy does not compare very well.
  \begin{figure}
      \begin{center}
      \includegraphics[width=.75\linewidth]{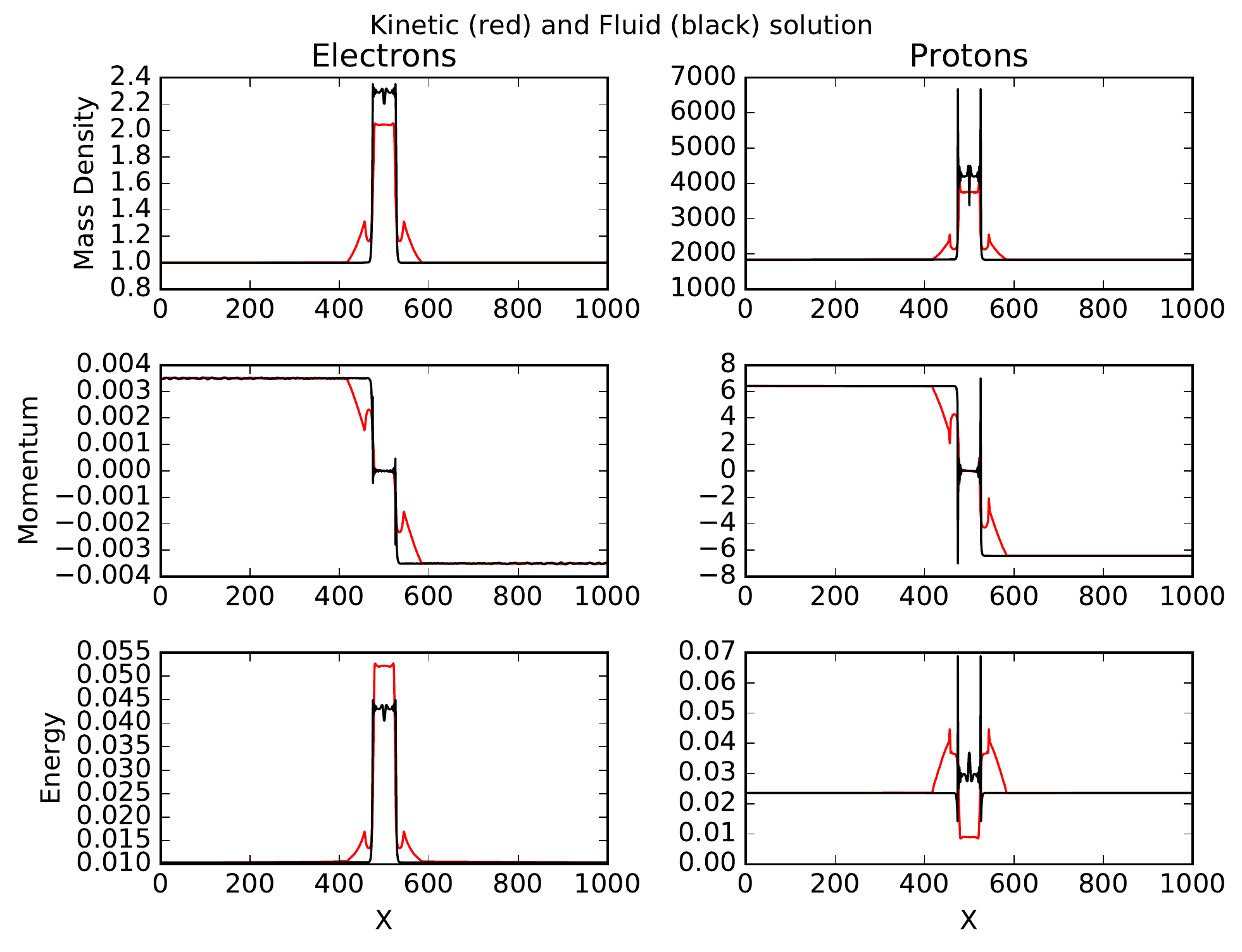}\\
      \includegraphics[width=.75\linewidth]{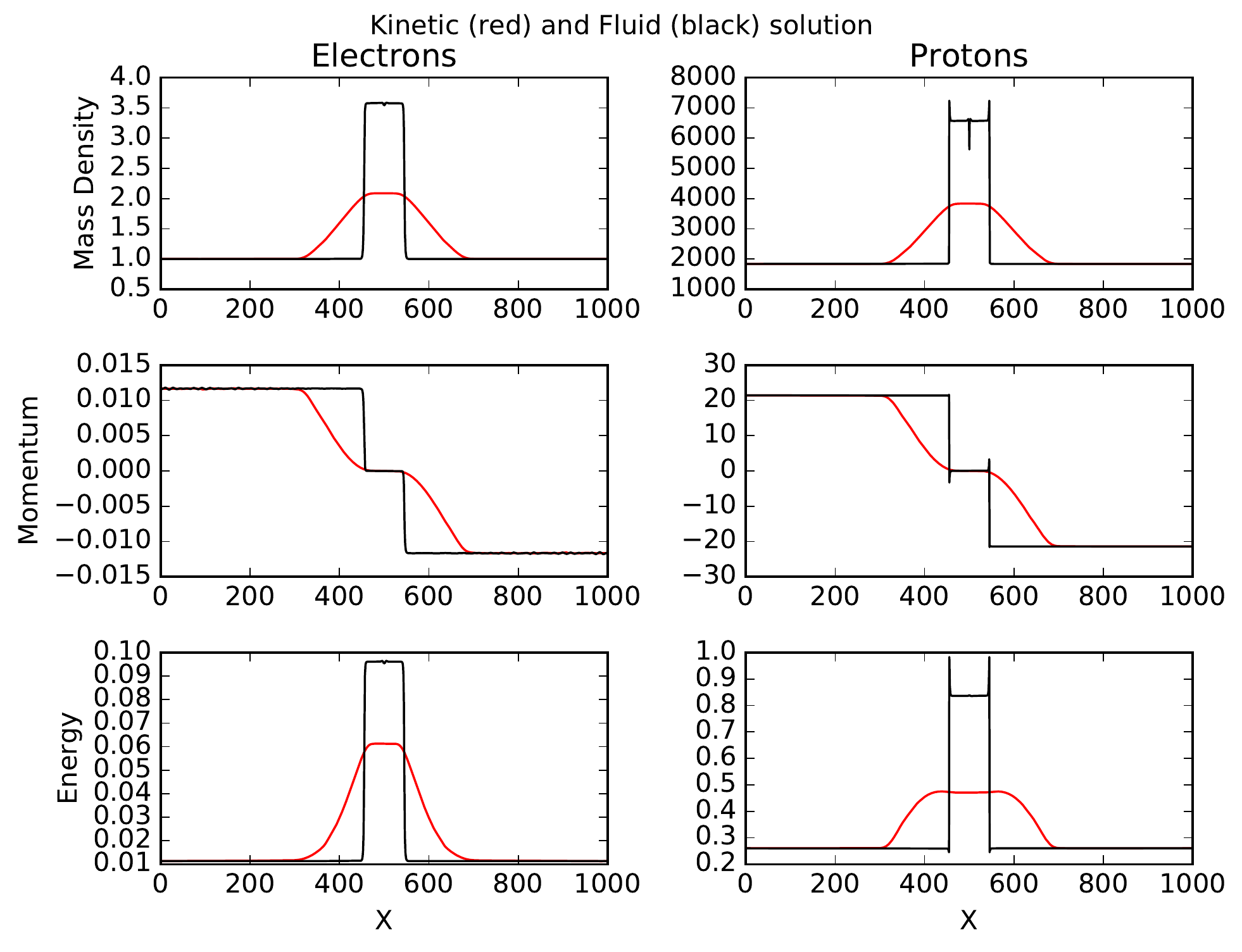}
      \end{center}
        \caption{Comparisons of the five moment two fluid moments (black) and the kinetic (red) moments for the $M_c = 1.5$ (top) and $M_c = 5$ (bottom) simulations. The mass densities and momenta of the protons and electrons in the $M_c = 1.5$ simulation compare reasonably well; however, the $M_c = 5$ simulation only shocks for the two fluid model, as expected \citep{Forslund:1970}. Note the density and momentum include the relevant mass factors for their definition, and the energy is defined by \eqr{\ref{eq:energyDefinition}}.}
        \label{fig:FluidMach}
  \end{figure}
  \begin{figure}
        \subfloat{\includegraphics[width=.5\linewidth]{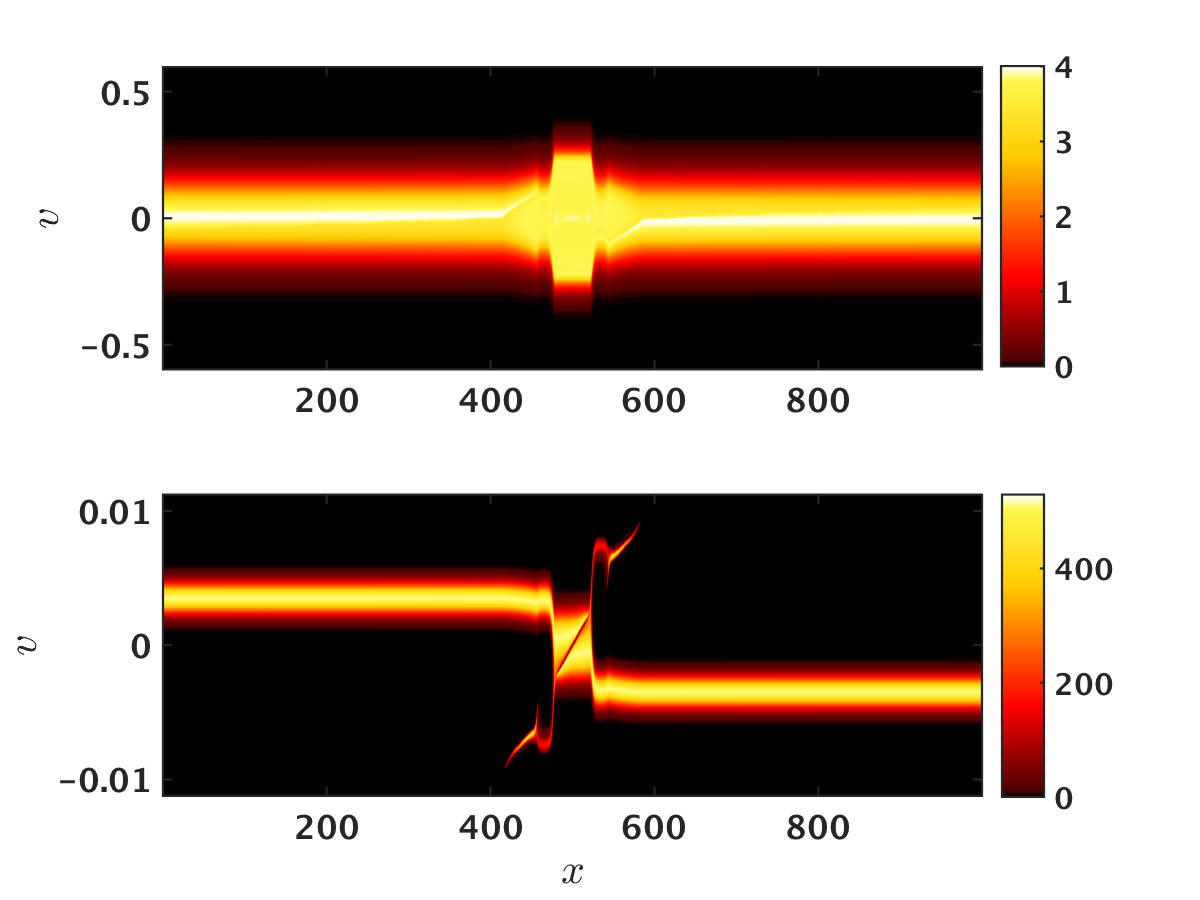}}
        \subfloat{\includegraphics[width=.5\linewidth]{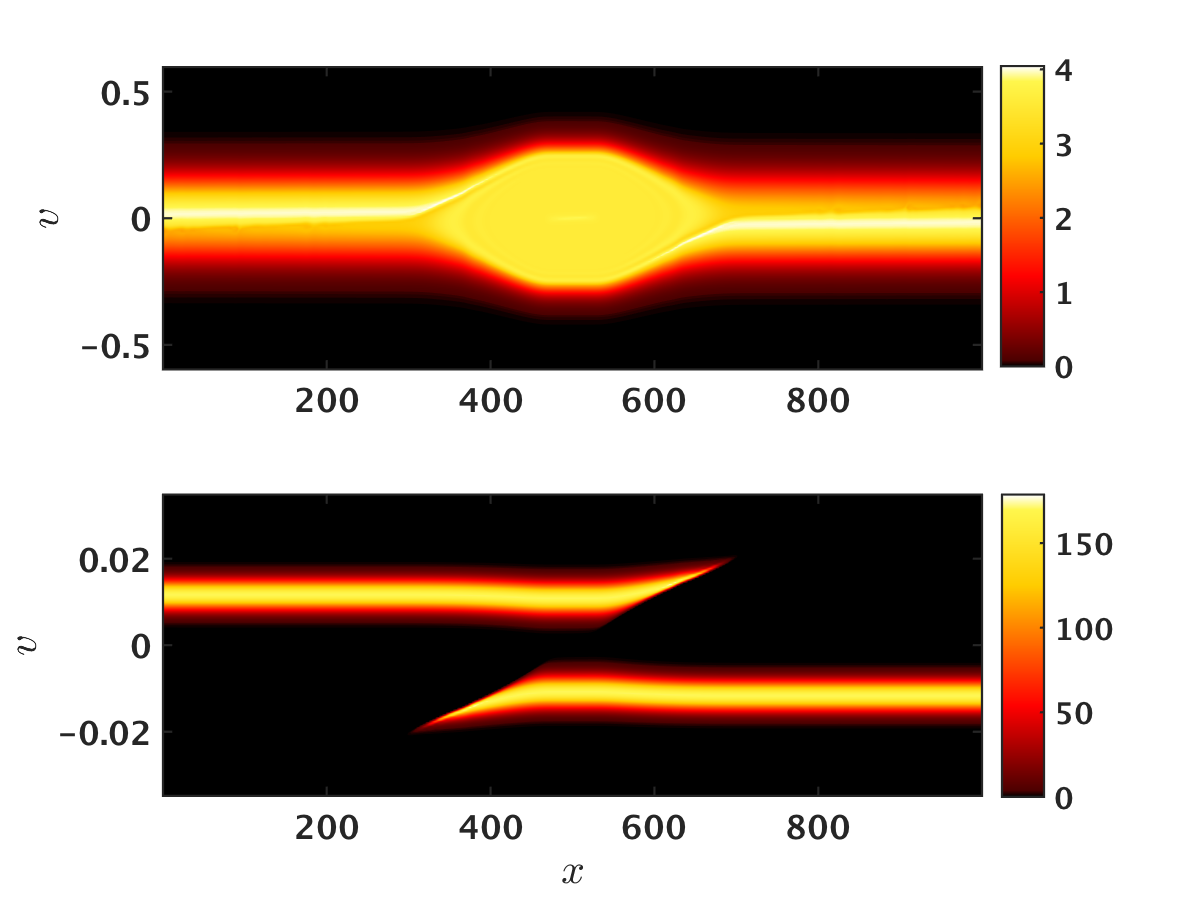}}
        \caption{Comparisons of electron (top) and proton (bottom) distribution functions for the $M_c = 1.5$ (left) and $M_c = 5$ (right) simulations. Note that because we use different velocity space grids for each species, the phase space density of a cell for each distribution function will be different, thus accounting for the difference in color bars between the electrons and protons.}
        \label{fig:ESShockDistributionFunction}
  \end{figure}
  \end{par}
\subsection{Electromagnetic applications}
\begin{par}
We now move on to more general applications, which involve both the electric and magnetic fields in Maxwell's equations, Eqns.\thinspace(\ref{eq:dbdt})-(\ref{eq:divB}). For all benchmarks the distribution function is again initialized as a drifting Maxwellian as in \eqr{\ref{eq:initMaxwell}}.
\end{par}
\subsubsection{Vlasov-Maxwell Riemann problem}
\begin{par}
A common robustness test of algorithms for fluid equations involves so-called ``Riemann'' problems, of which the Brio-Wu shock tube is a famous example \citep{BrioWu1988}. Here we perform a 1X3V, $(x, v_x, v_y, v_z)$, simulation of a Riemann problem with our Vlasov-Maxwell solver and compare with a case previously studied with a two-fluid algorithm. We have again employed the five moment two-fluid algorithm described in Hakim 2006 \citep{Hakim:2006} to compare the fluid moments between the two simulations. The initial condition is given by
\begin{align}
\begin{bmatrix} n_e \\ u_{e_x} \\ u_{e_y} \\ u_{e_z} \\ \epsilon_e \\ n_p \\ u_{p_x} \\ u_{p_y} \\ u_{p_z} \\ \epsilon_p \\ E_x \\ E_y \\ E_z \\ B_x \\ B_y \\ B_z \end{bmatrix} = \begin{bmatrix} 1.0 \\ 0.0 \\ 0.0 \\ 0.0 \\ 5 \times 10^{-5} \\ 1.0 \\ 0.0 \\ 0.0 \\ 0.0  \\ 5 \times 10^{-5} \\ 0.0 \\ 0.0 \\ 0.0  \\ 0.75 \times 10^{-2} \\ 0.0 \\ 1.0\times 10^{-2} \end{bmatrix} 
\qquad
\begin{bmatrix} n_e \\ u_{e_x} \\ u_{e_y} \\ u_{e_z} \\ \epsilon_e \\ n_p \\ u_{p_x} \\ u_{p_y} \\ u_{p_z} \\ \epsilon_p \\ E_x \\ E_y \\ E_z \\ B_x \\ B_y \\ B_z \end{bmatrix} = \begin{bmatrix} 0.125 \\ 0.0 \\ 0.0 \\ 0.0 \\ 5 \times 10^{-6} \\ 0.125 \\ 0.0 \\ 0.0 \\ 0.0  \\ 5 \times 10^{-6} \\ 0.0 \\ 0.0 \\ 0.0  \\ 0.75 \times 10^{-2} \\ 0.0 \\ -1.0\times 10^{-2} \end{bmatrix},
\end{align}
where the left values are the fluid moments initialized on the left half of the domain, and the right values are the fluid moments initialized on the right half of the domain. We note that $\epsilon_s$, the energy, is defined in \eqr{\ref{eq:energyDefinition}}. The full mass ratio, $m_p/m_e = 1836$, is employed, and the domain size is $L_x = d_p$ where $d_p = c/\omega_{pp}$ is the proton skin depth. All parameter definitions are with respect to values on the left half of the domain, i.e., $L_x = d_p$ is defined using the proton plasma frequency $\omega_{pp}$ on the left half of the domain. The simulation is run for $t = 0.125 \Omega_{cp}^{-1}$, where
\begin{align}
\Omega_{cp} = \frac{e |\mvec{B}|}{m_p} 
\end{align}
is the proton cyclotron frequency. We note that the cyclotron frequency is the same on both the left and right halves of the domain since the magnetic field has the same magnitude on each side of the interface. The velocity space extents of the Vlasov-Maxwell run are $\pm 6 v_{th_s}$, where $v_{th_s}$ is defined by the higher temperature on the left side of the domain for each species, and zero-flux boundary conditions are employed in velocity space. As in the electrostatic shock problem, open boundary conditions are employed in configuration space. The speed of light is chosen such that $v_{th_e}/c = 1/7$, so that the edge of velocity space for the electrons is sub-luminal. The resolution of the Vlasov-Maxwell simulation is $64 \times 32^3$ with polynomial order 2. We can compare some of the aforementioned fluid moments between the Vlasov-Maxwell and the same five moment two fluid model employed in Section \ref{sec:ESShock}. Results are shown in Figure~\ref{fig:VlasovRiemannFluidMoments}.
  \begin{figure}
        \subfloat{\includegraphics[width=\linewidth]{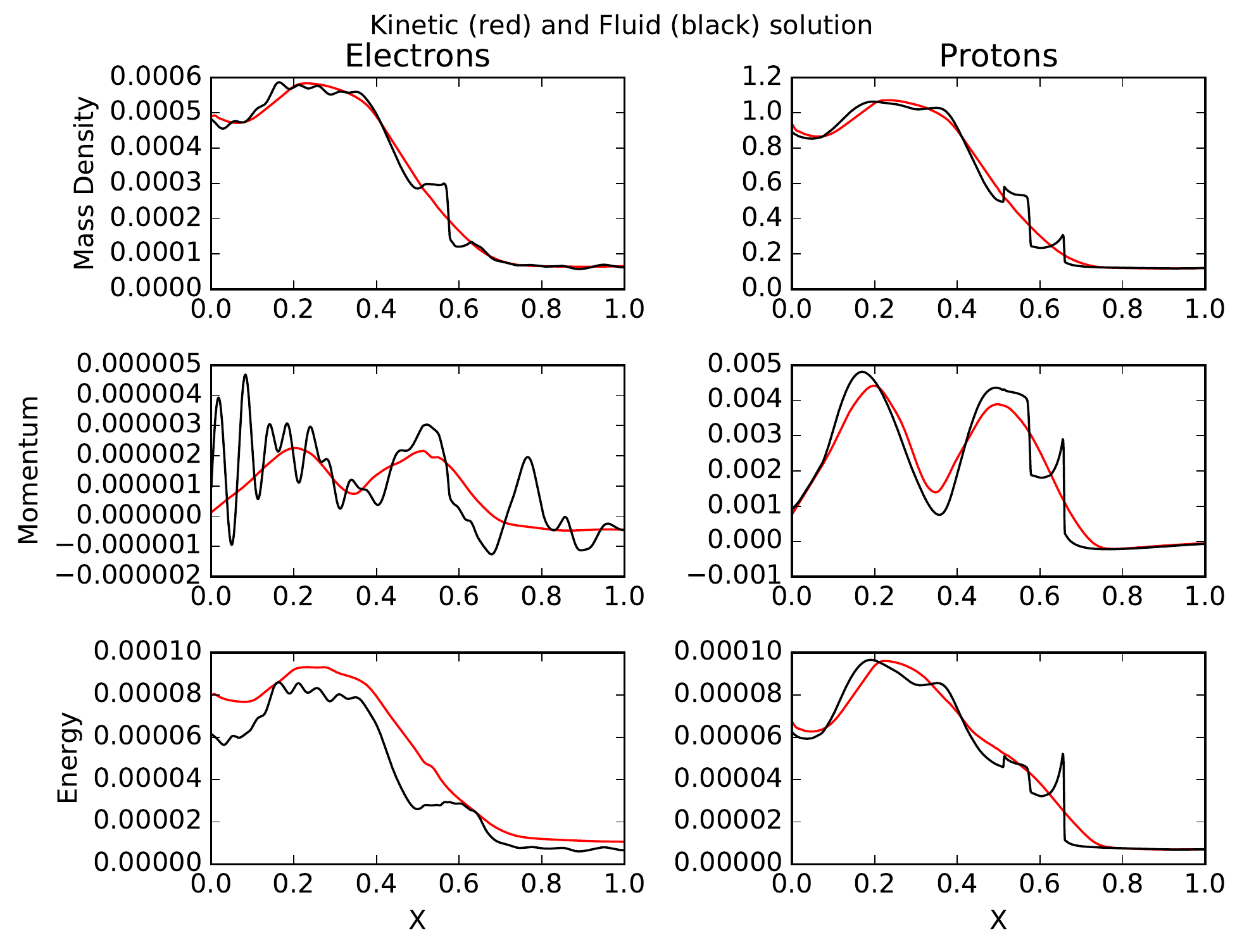}}        
        \caption{Comparison of electron mass density (top left) and ion mass density (top right), electron momentum (middle left), ion momentum (middle right), electron particle energy (bottom left) and ion particle energy (bottom right) between two-fluid and kinetics for the Riemann problem.}
        \label{fig:VlasovRiemannFluidMoments}
  \end{figure}
We note some key differences between the two models due to the more complete description of the plasma the Vlasov-Maxwell system provides. The presence of a larger amount of structure in the two-fluid model is owed to the fact that in the fluid model, the dispersive waves launched by the sharp gradient in the density and the temperature can only propagate and nonlinearly interact with each other. However, these dispersive waves in the Vlasov-Maxwell system can also damp on the particles in the plasma, giving their energy to the plasma particles via a variety of wave-particle resonances, such as Landau damping. The presence of wave-particle resonance processes is likely the reason for the striking difference between the electron particle energy between the two models, where the profiles are similar in terms of global structure, but the Vlasov-Maxwell energy density profile is shifted up. We note that this shift is indeed due to the energization of the electrons and not due to any sort of energy conservation violations, as energy is conserved for this run to within 1 percent, with the only losses due to the fact that the configuration space boundary conditions are open. Further evidence for this damping can be seen by examining the electron distribution function. Isosurfaces of the electron distribution function are plotted in Figure~\ref{fig:VlasovRiemannIso}. The ripples on the distribution function where the layer is initialized are a characteristic of significant phase-mixing, a key component of the damping of these waves. Evidence for phase-mixing is especially clear when considering 2D and 1D cuts of the distribution function, plotted in Figures~\ref{fig:VlasovRiemann2DCuts} and \ref{fig:VlasovRiemann1DCuts}.
  \begin{figure}
         \subfloat{\includegraphics[width=0.5\linewidth]{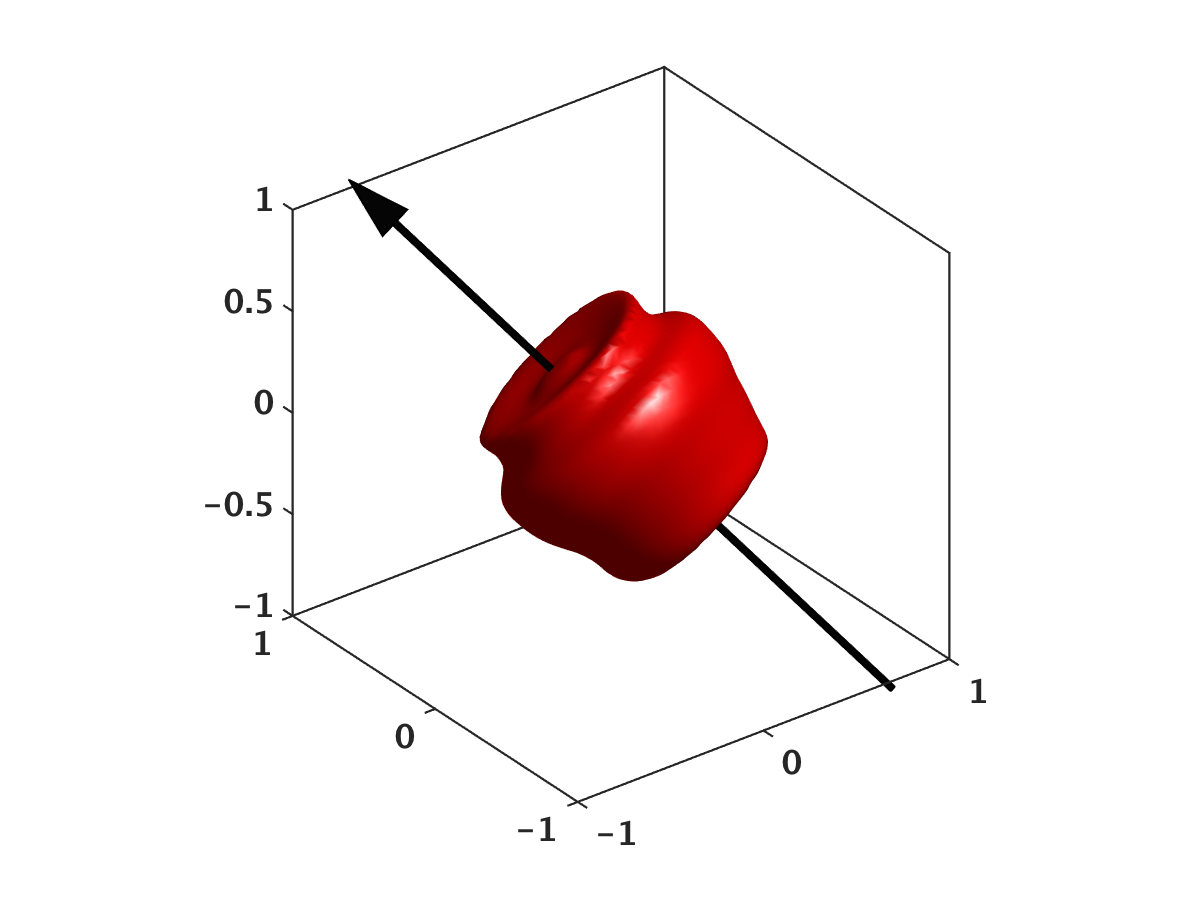}}
         \subfloat{\includegraphics[width=0.5\linewidth]{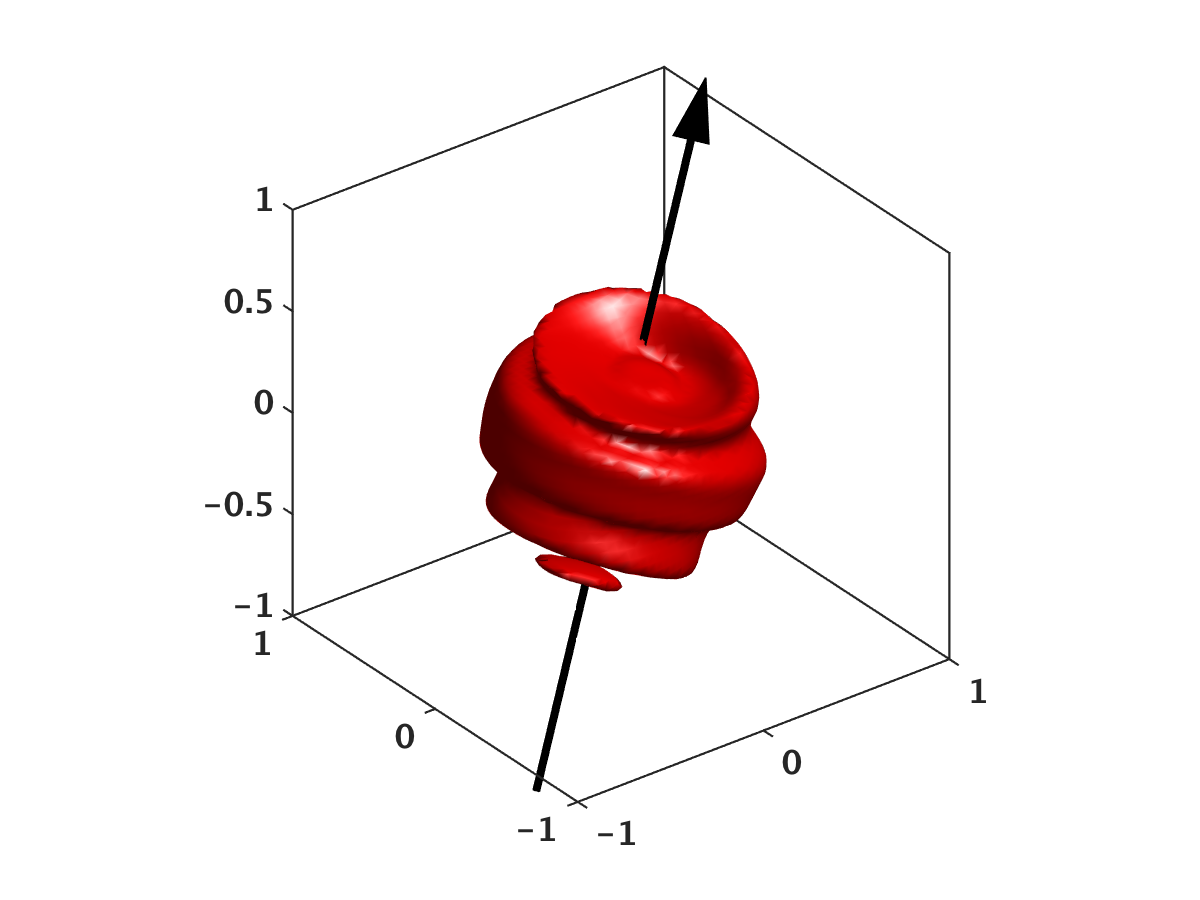}}
        
        \caption{Isosurfaces of the electron distribution function at $x = 0.3d_p$ (left) and $x = 0.5d_p$ (right) from the Riemann problem at $t = 0.125 \Omega_{cp}^{-1}$. The arrow points in the direction of the total magnetic field. We note the that the distribution function in the left plot is gyrotropic, i.e., symmetric about the total magnetic field, unlike where we initialize the interface, which has significant agyrotropic features.}
        \label{fig:VlasovRiemannIso}
  \end{figure} 
  \begin{figure}
        \subfloat{\includegraphics[width=0.33\linewidth]{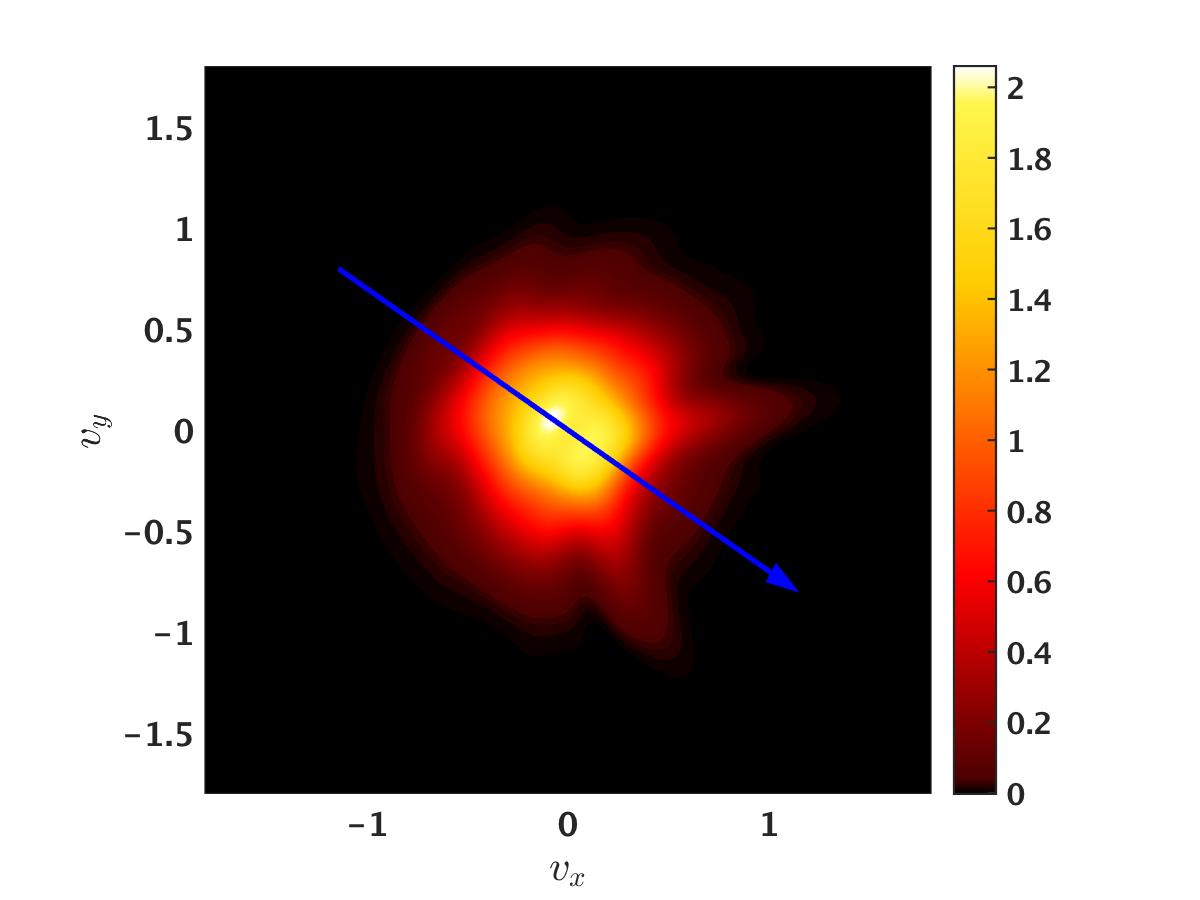}}
        \subfloat{\includegraphics[width=0.33\linewidth]{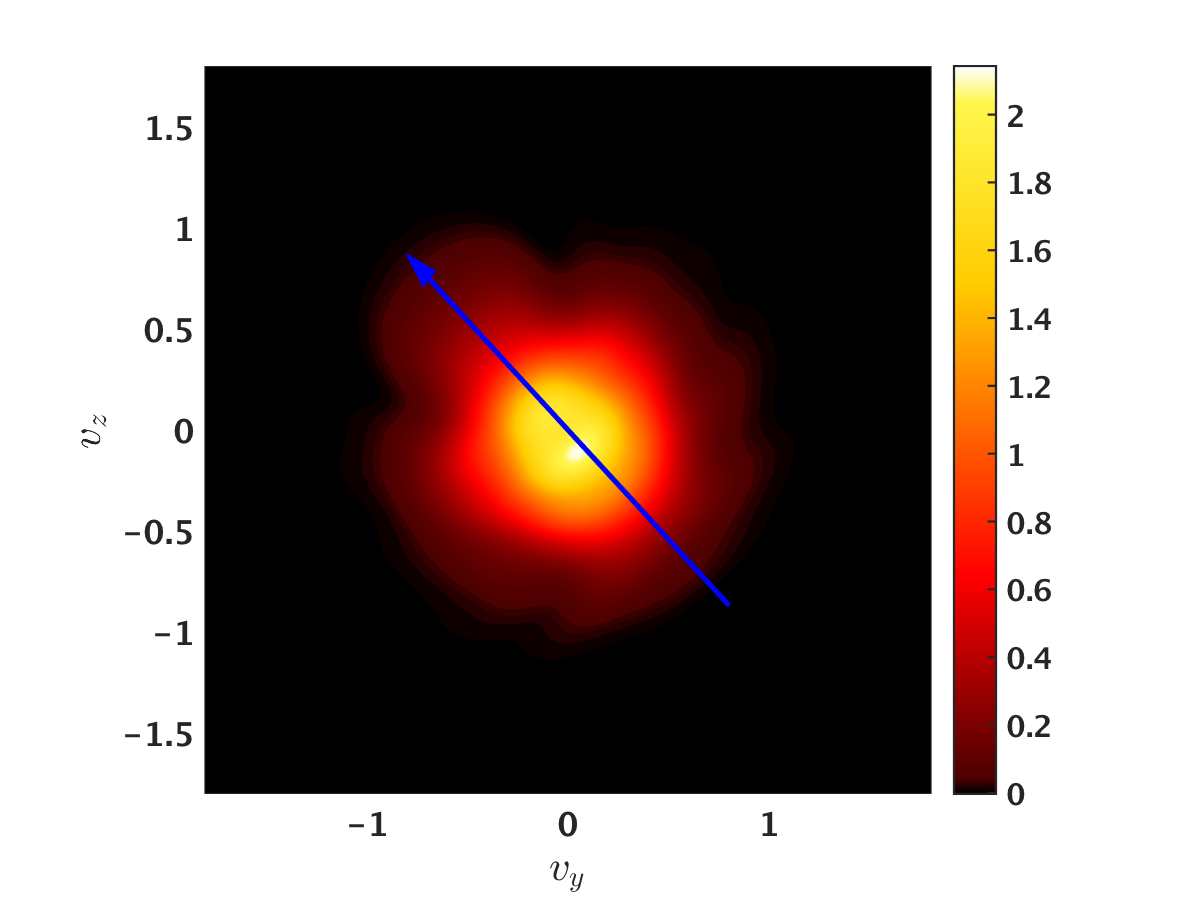}}
        \subfloat{\includegraphics[width=0.33\linewidth]{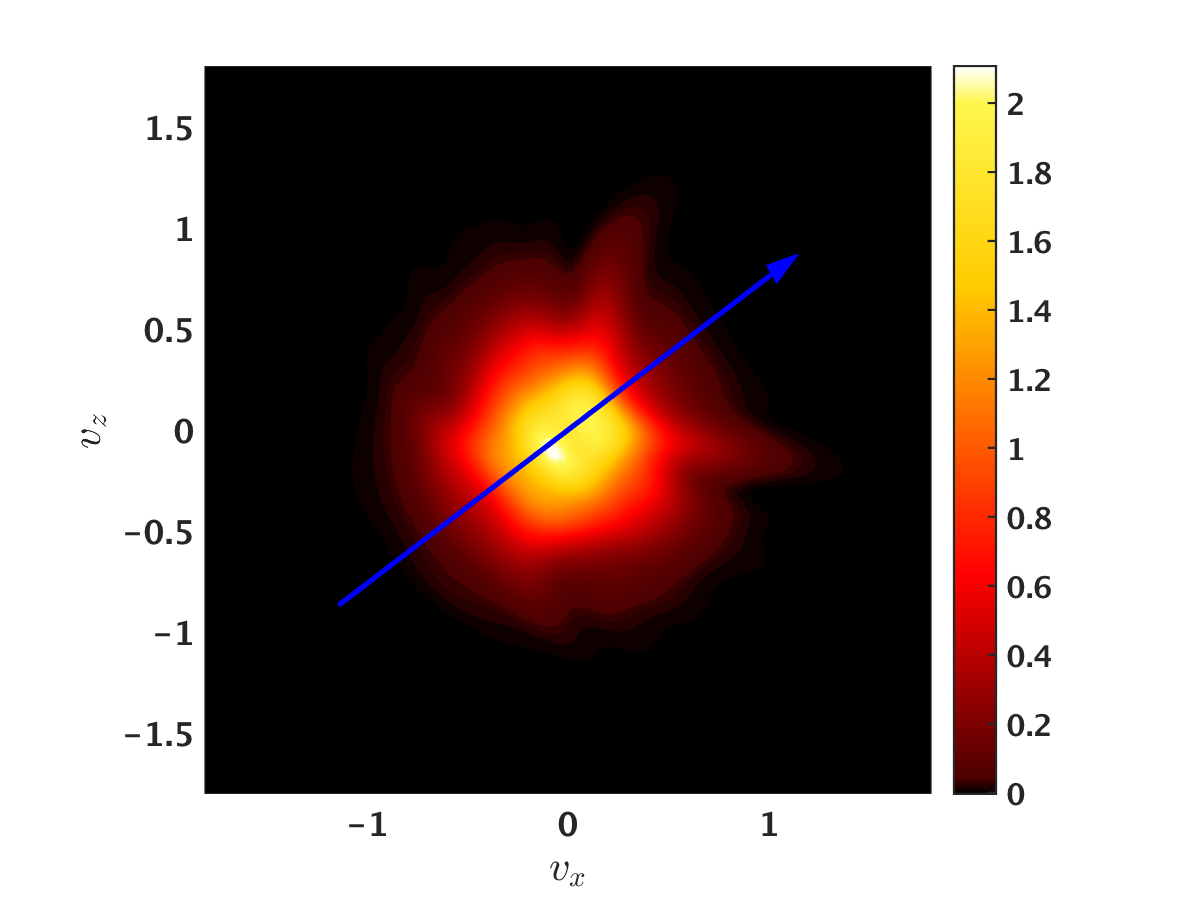}}\\
        \subfloat{\includegraphics[width=0.33\linewidth]{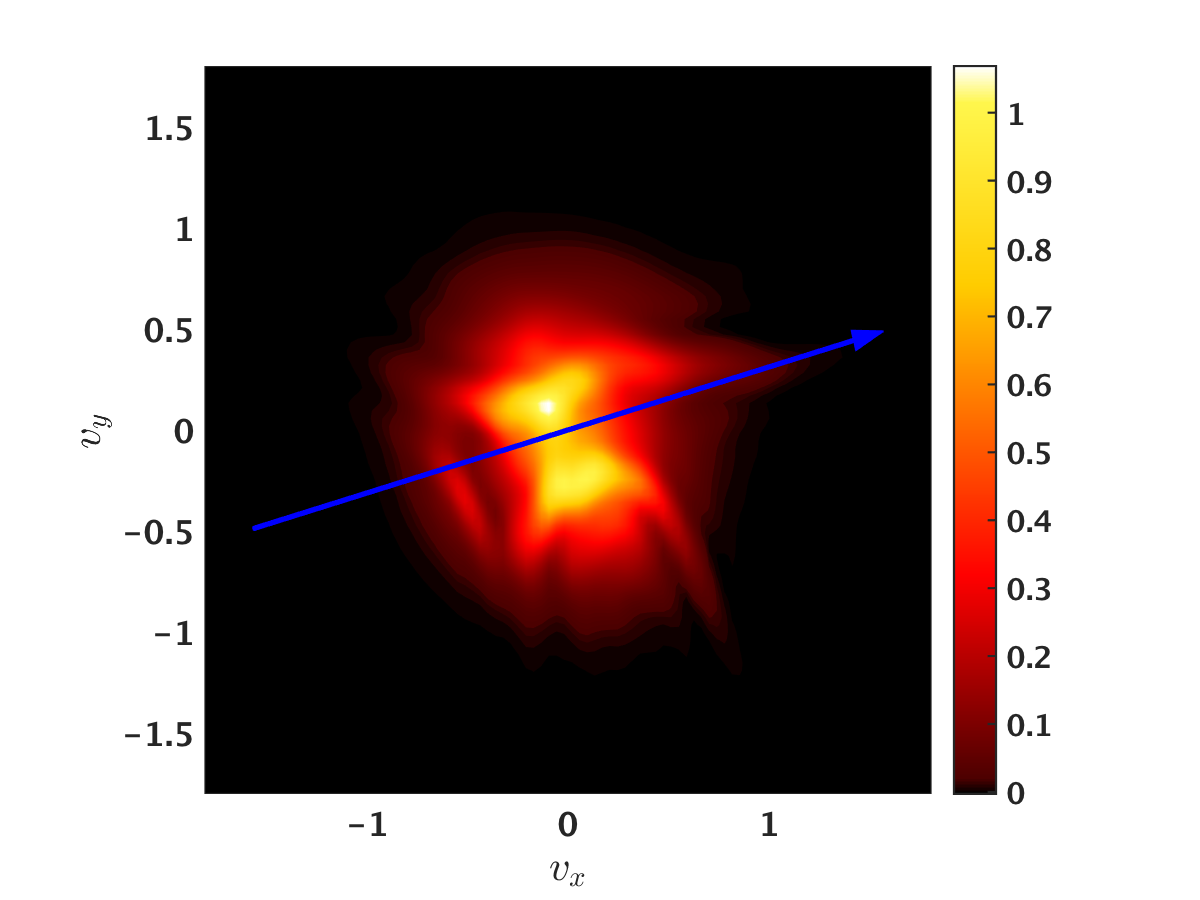}}
        \subfloat{\includegraphics[width=0.33\linewidth]{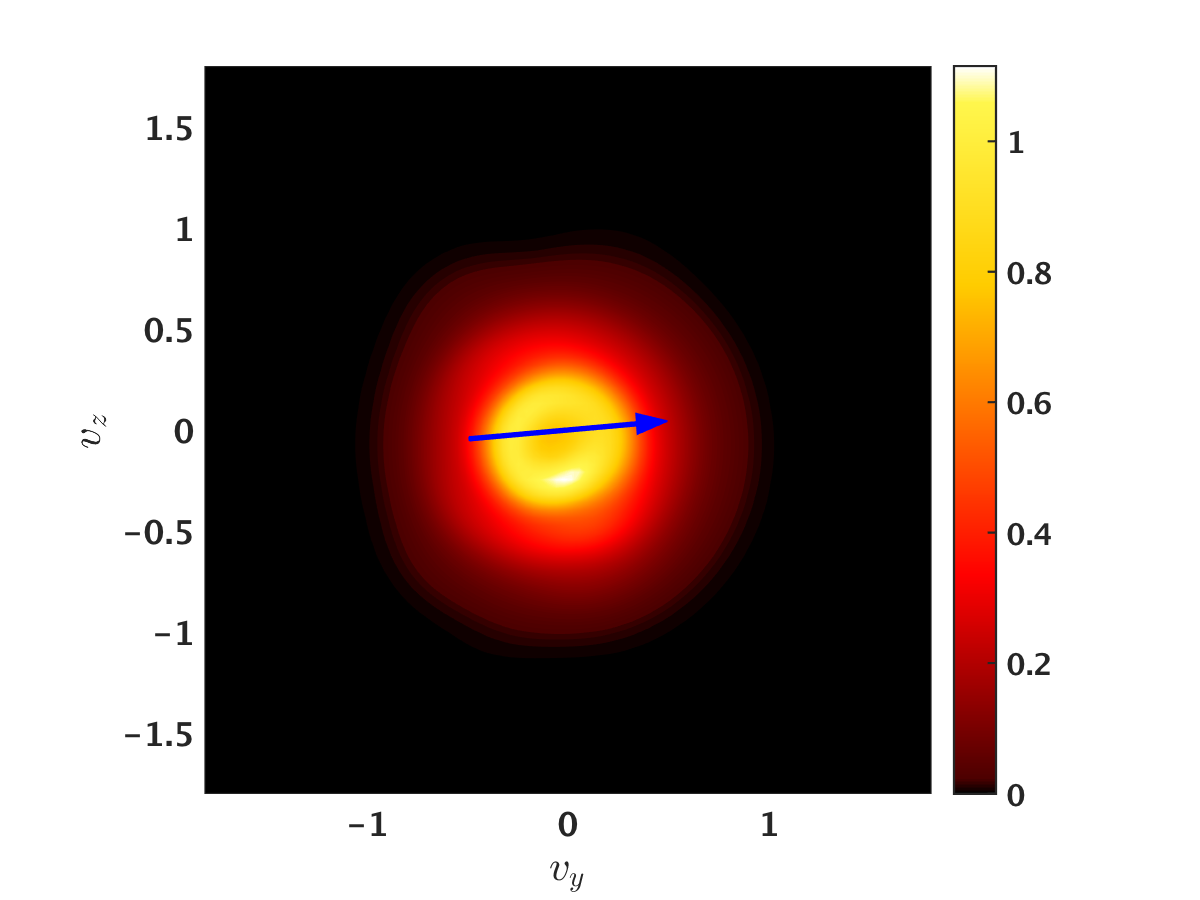}}
        \subfloat{\includegraphics[width=0.33\linewidth]{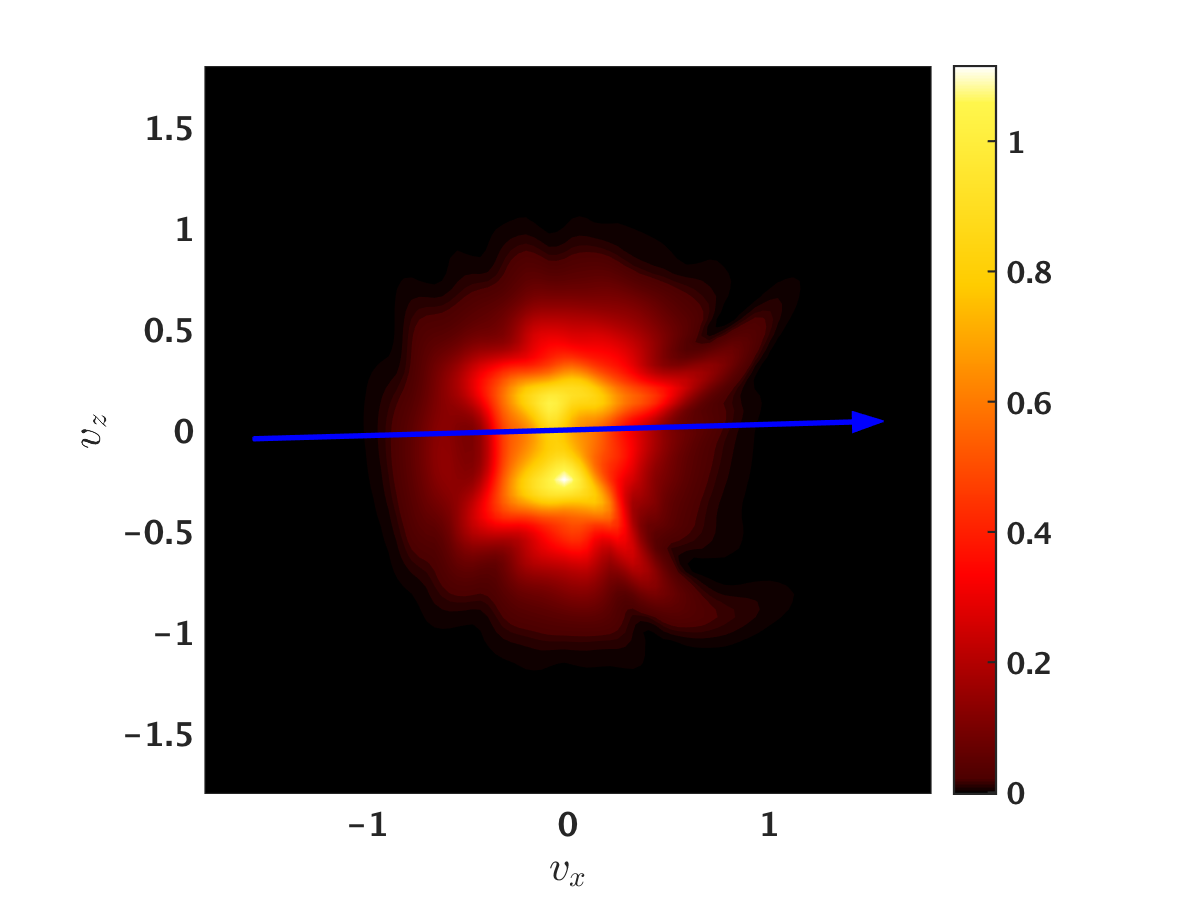}}
        \caption{2D cuts of the 3D electron distribution function at $x = 0.3d_p$ (top) and $x = 0.5d_p$ (bottom) through $v_z = 0$ (left), $v_x = 0$ (middle), and $v_y = 0$ (right) from the Riemann problem at $t = 0.125 \Omega_{cp}^{-1}$. The symmetry about the magnetic field in the region downstream of the interface is even more striking, confirming the distribution function is indeed gyrotropic in this region of configuration space. However, while a similar symmetry is present in the $v_y, v_z$ contour plot at the interface, the contours through $v_y = 0$ and $v_z = 0$ show significant structure along $v_x$. This is noteworthy because this is the dimension we have the configuration space component of the phase space advection, or the dimension in which the $\mvec{v} \cdot \nabla$ term in the Vlasov equation is non-zero, which is principally responsible for phase-mixing.}
        \label{fig:VlasovRiemann2DCuts}
  \end{figure}
  \begin{figure}
         \subfloat{\includegraphics[width=0.5\linewidth]{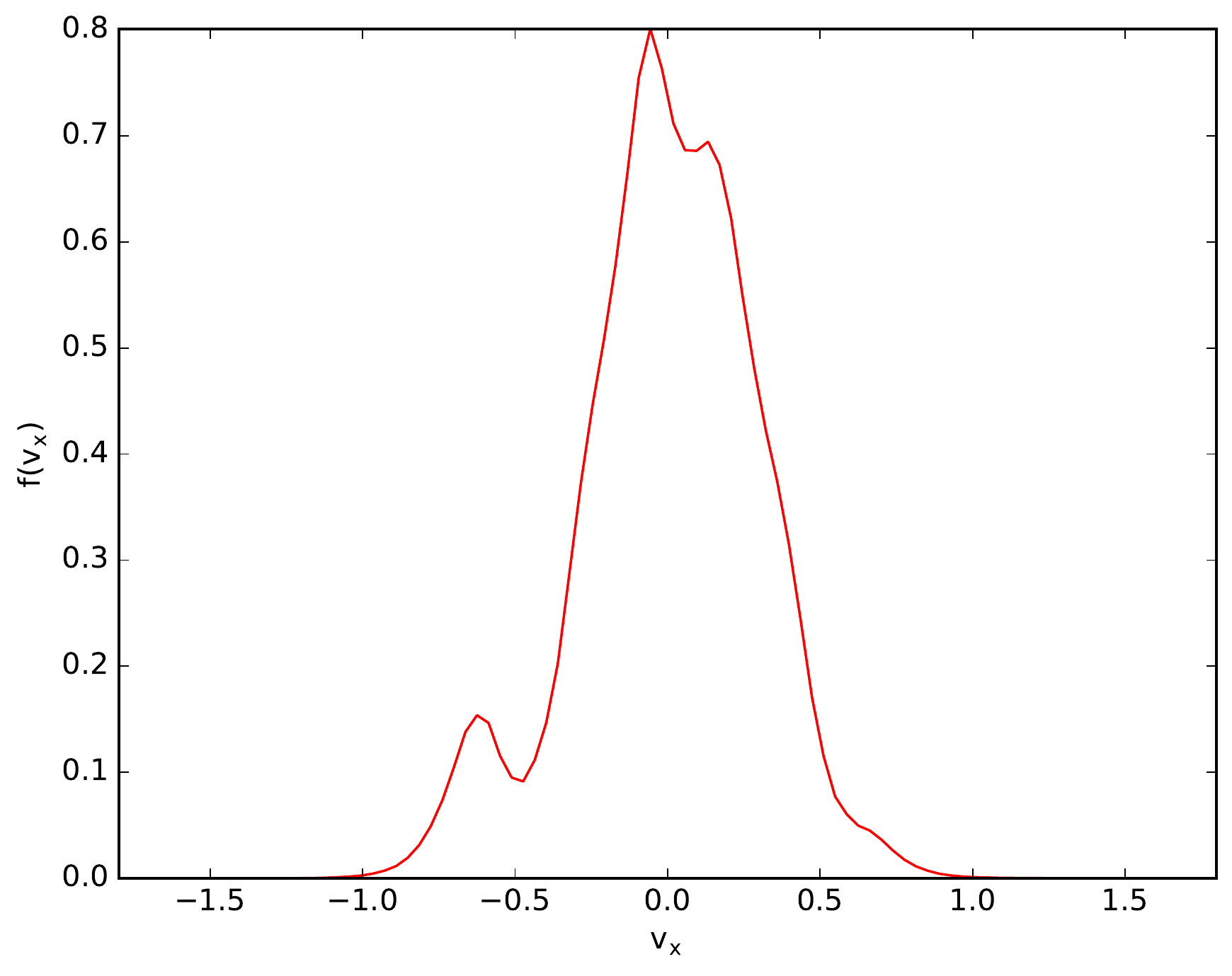}}
         \subfloat{\includegraphics[width=0.5\linewidth]{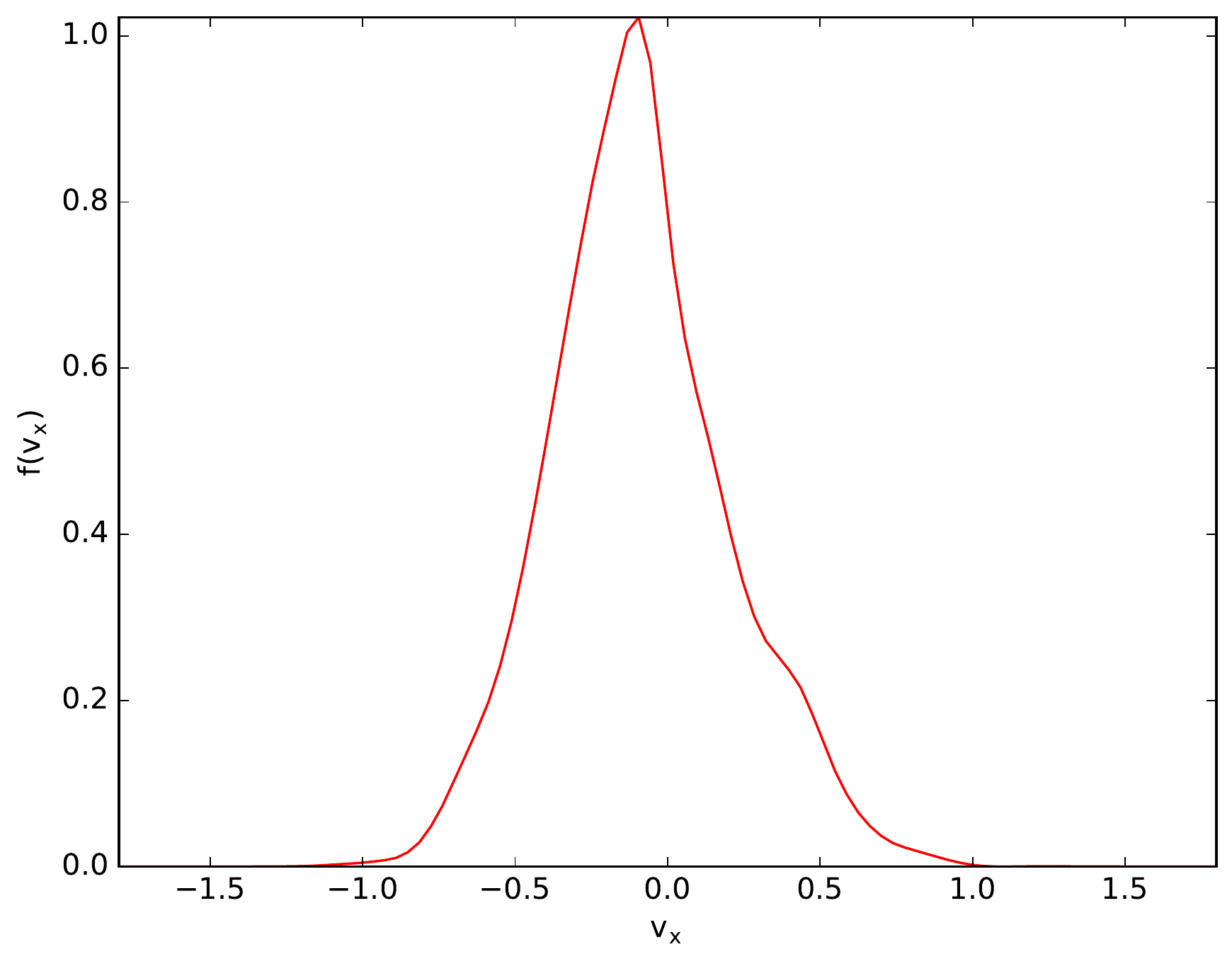}}
        
        \caption{Cuts through $v_y = -1/3 v_{th_e}$ (left) and $v_y = 1/3 v_{th_e}$ (right) of the $v_z = 0$ electron distribution function contour plot at $x =  0.5d_p$ in Figure~\ref{fig:VlasovRiemann2DCuts} from the Riemann problem at $t = 0.125 \Omega_{cp}^{-1}$. The structure present, especially in the $v_y = -1/3 v_{th_e}$ plot on the left, provides further credibility to the intuition that significant amounts of collisionless damping are occurring in the Vlasov simulation that are not present in the fluid simulation.}
        \label{fig:VlasovRiemann1DCuts}
  \end{figure} 
\end{par}
\subsubsection{The Orszag-Tang vortex}\label{sec:OT}
\begin{par}
Our final benchmark is a 2X3V, $(x, y, v_x, v_y, v_z)$, calculation of the Orszag-Tang vortex \citep{Orszag:1979}, a famous benchmark of magnetohydrodynamic (MHD) codes, and an extensively used simulation setup to study the physics of turbulence in 2D \citep{Parashar:2009, Li:2015}. The Orszag-Tang Vortex initial condition is given by:
\begin{align}
\mvec{B} & = -B_0 \sin\left(\frac{2 \pi y}{L_y} \right ) \mvec{\hat{x}} + B_0  \sin\left(\frac{4 \pi x}{L_x} \right ) \mvec{\hat{y}} + B_g \mvec{\hat{z}}, \\
\mvec{u}_p & = -u_0  \sin\left(\frac{2 \pi y}{L_y} \right ) \mvec{\hat{x}} + u_0  \sin\left(\frac{2 \pi x}{L_x} \right ) \mvec{\hat{y}}, \\
\mvec{J} & = \frac{1}{\mu_0} \nabla \times \mvec{B} = e n_0 (\mvec{u}_p - \mvec{u}_e), \\
\mvec{E} & = - \mvec{u}_e \times \mvec{B},
\end{align}
where we note the slight modification from the standard MHD initial condition is required in the equation for the electric field since we are now evolving a multi-species plasma. $B_g$ here is an out-of-plane guide field. The following parameters are employed: $\beta_p = 0.08, m_p/m_e = 25, \Omega_{ce}/\omega_{pe} = v_{Ae}/c = 0.1$, where $\beta_p$ is the proton plasma beta, $\beta_p = 2v_{th_p}^2/v_{Ap}^2$ and $v_{As}$ is the Alfv\'en speed for species $s$, $v_{As} = B_g/\sqrt{\mu_0 n_0 m_s}$. The amplitudes for the initial condition are, $B_0 = 0.2 B_g, u_0 = 0.2 v_{Ap}$. $n_0$ is taken to be 1.0. The size of our box in configuration space is $(L_x, L_y) = (20.48 d_p, 20.48 d_p)$ and the velocity space extents are $\pm 5.6 v_{th_s}$. Periodic boundary conditions are employed in configuration space and zero flux boundary conditions are employed in velocity space. The resolution of the simulation is $40^2 \times 8^3$ with polynomial order 2. Our resolution is such that, with polynomial order 2, we resolve the electron skin depth, $d_e$. Schematically, the initial condition is given in Figure~\ref{fig:OTIC}. For comparison, an analogous 10 moment two fluid calculation is run, though with larger resolution to avoid grid scale diffusion in the fluid algorithm. By 10 moments, we mean that the algorithm solves for all six independent components of the pressure tensor for both species. The algorithm description is given in Hakim 2008 \citep{Hakim:2008}. After some time, the flow and magnetic field become turbulent due to the twisting of the field and flow. In Figure~\ref{fig:OTFluidMoments}, the out-of-plane current and magnetic field contours are compared between the two models, and qualitative agreement is found.
\begin{figure}
\centerline{
    \setkeys{Gin}{width=0.5\linewidth}
	\includegraphics{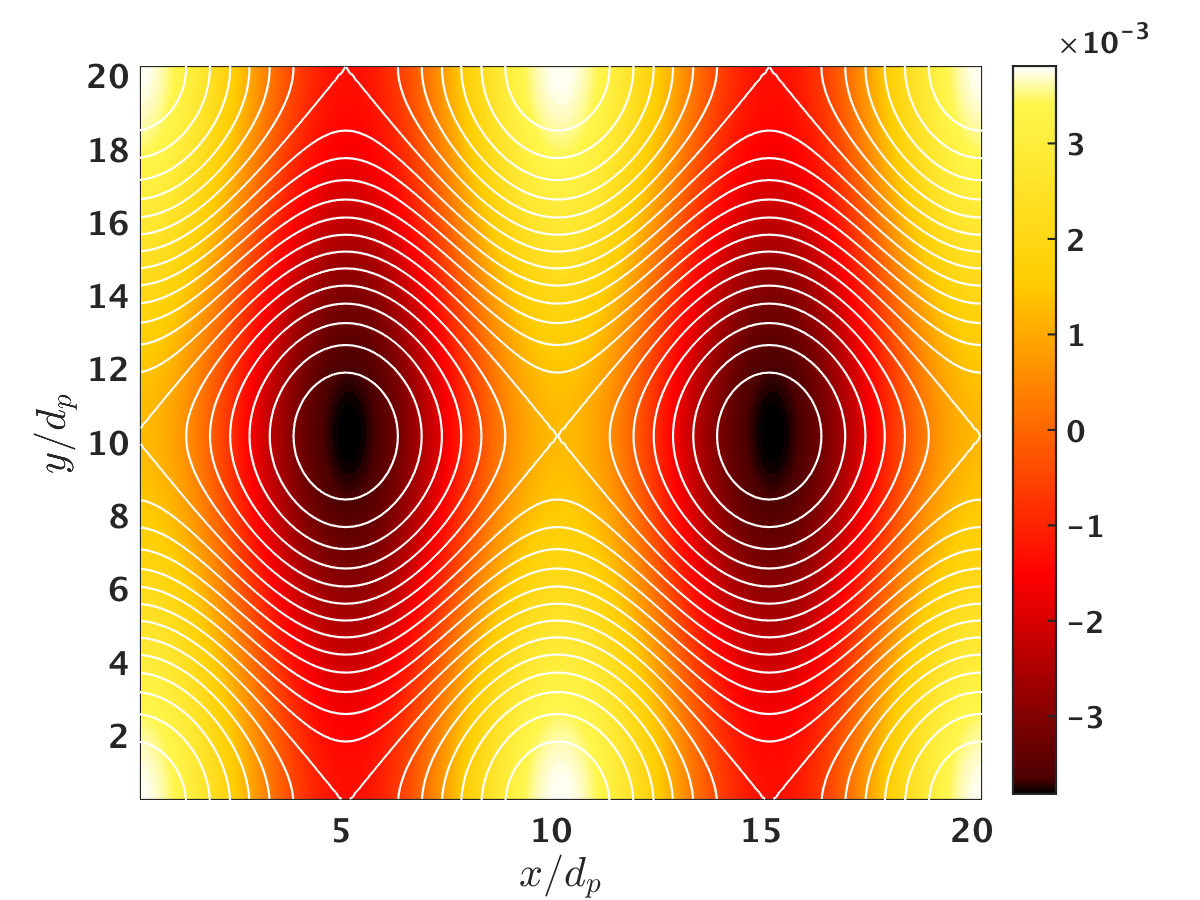}
	\includegraphics{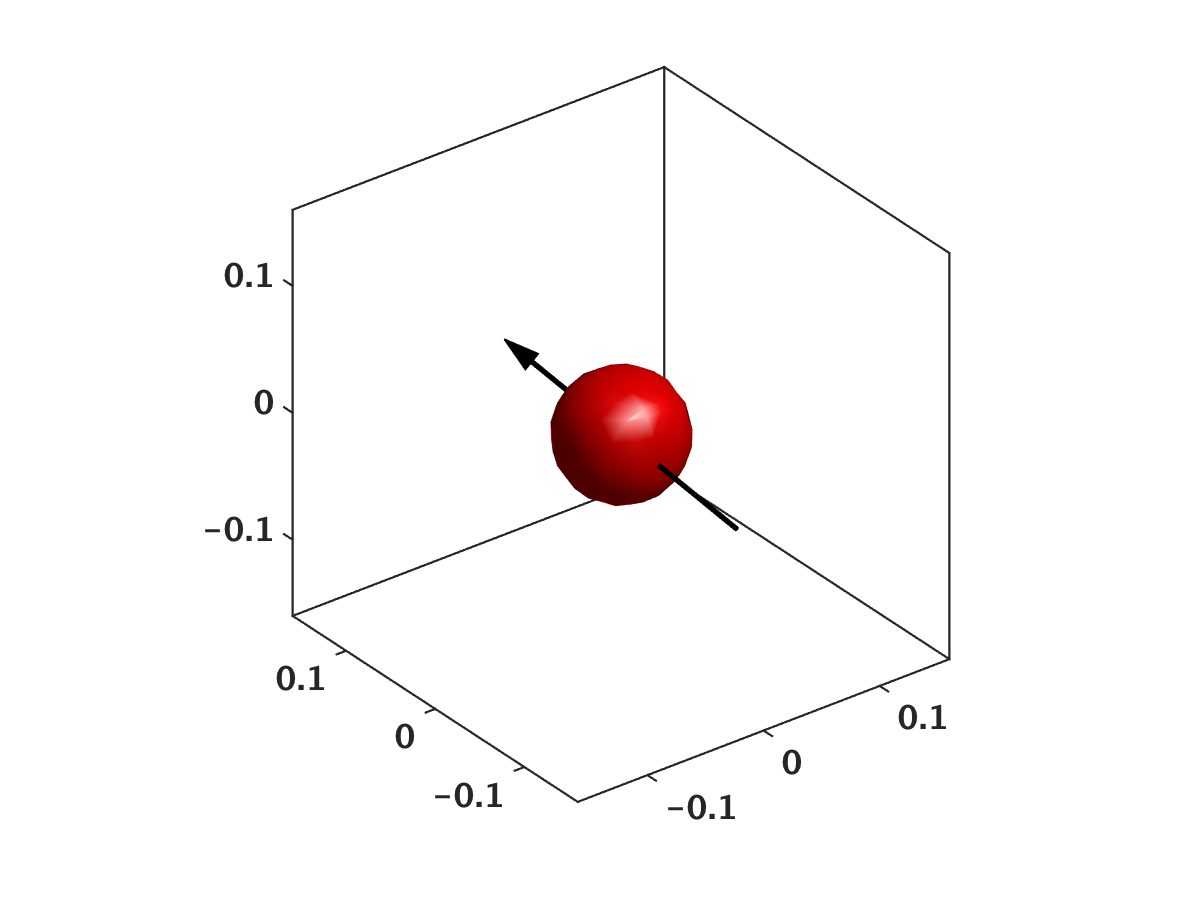}
	}
	\caption{The out of plane current $J_z$ with magnetic field contours superimposed (left) and an example initial proton distribution function (right) at a single point in configuration space for the initial condition of the Orszag-Tang vortex.}
	\label{fig:OTIC}
\end{figure}
\begin{figure}
	\subfloat{\includegraphics[width=0.5\linewidth]{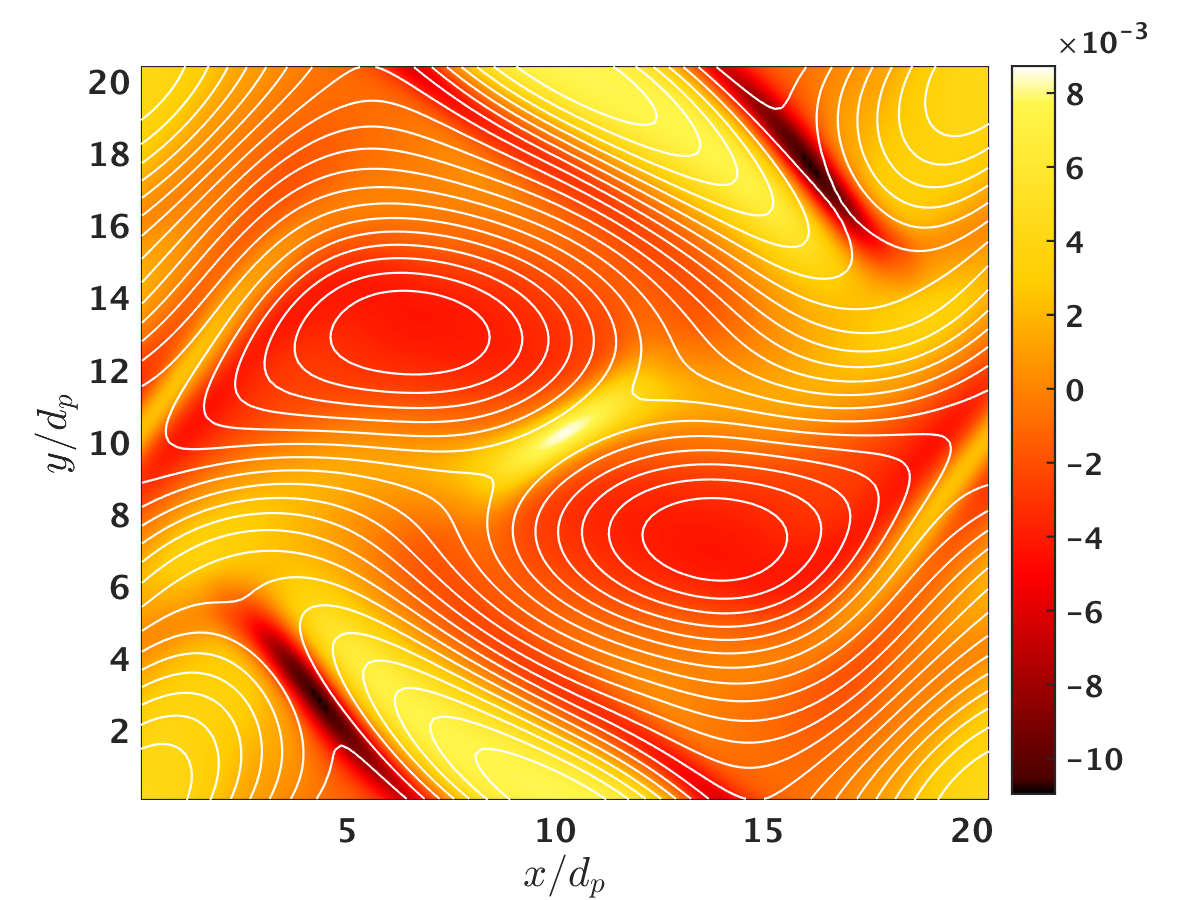}}
	\subfloat{\includegraphics[width=0.5\linewidth]{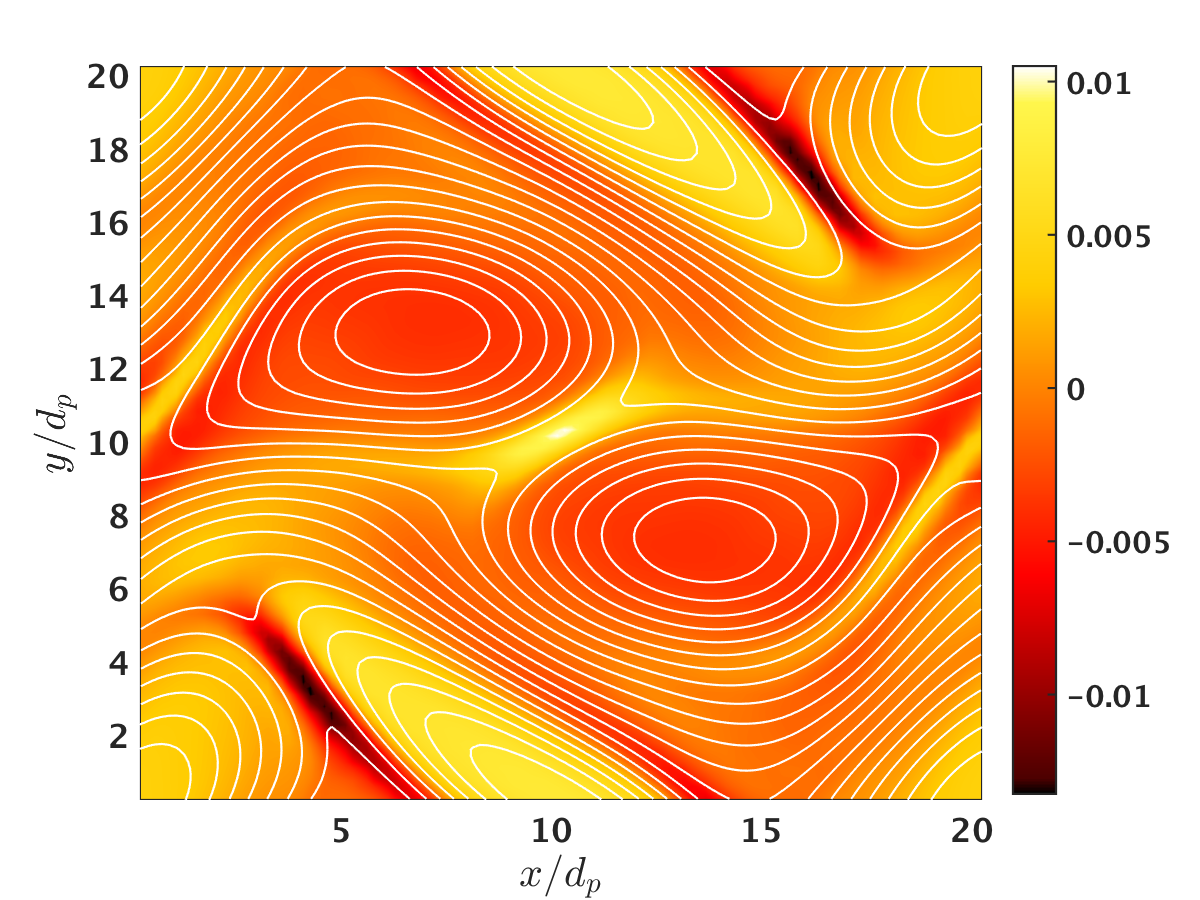}}\\
	\subfloat{\includegraphics[width=0.5\linewidth]{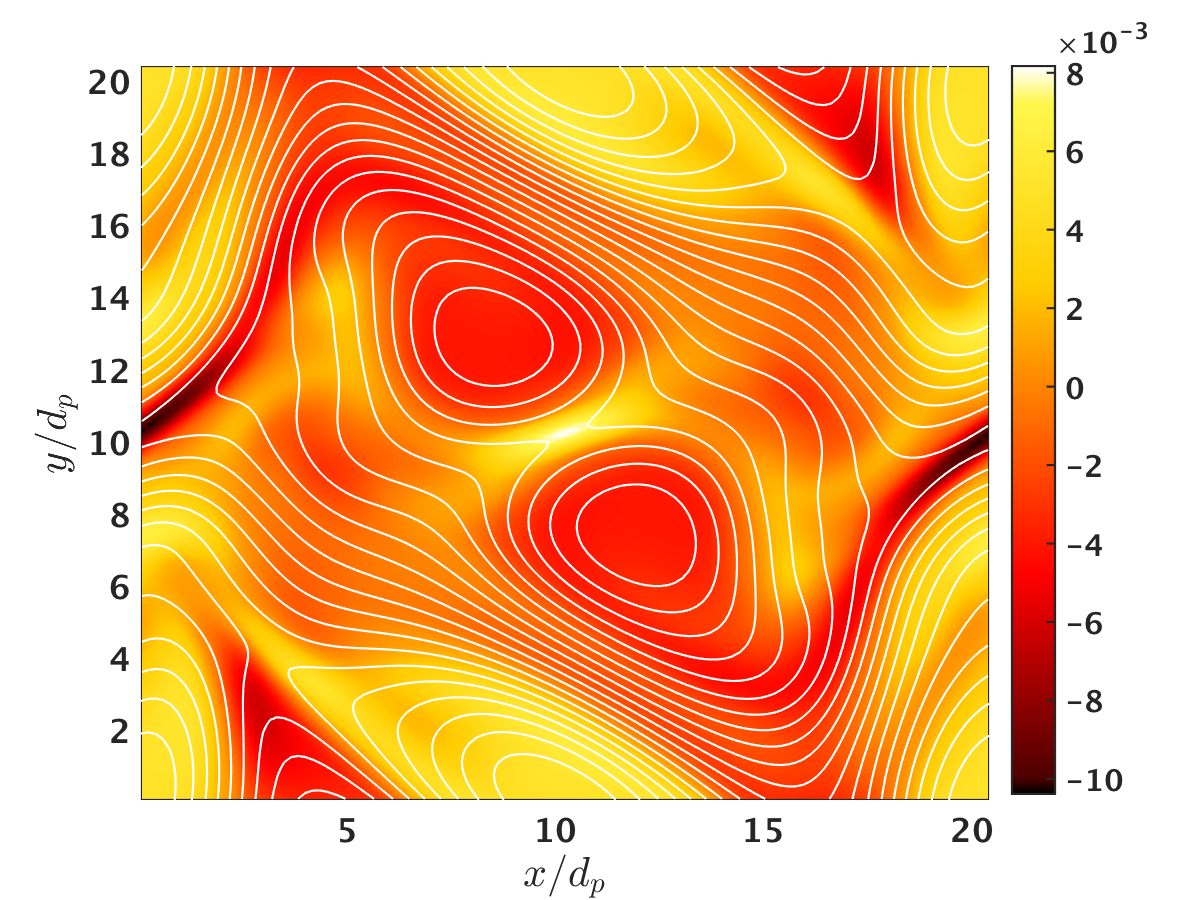}}    
	\subfloat{\includegraphics[width=0.5\linewidth]{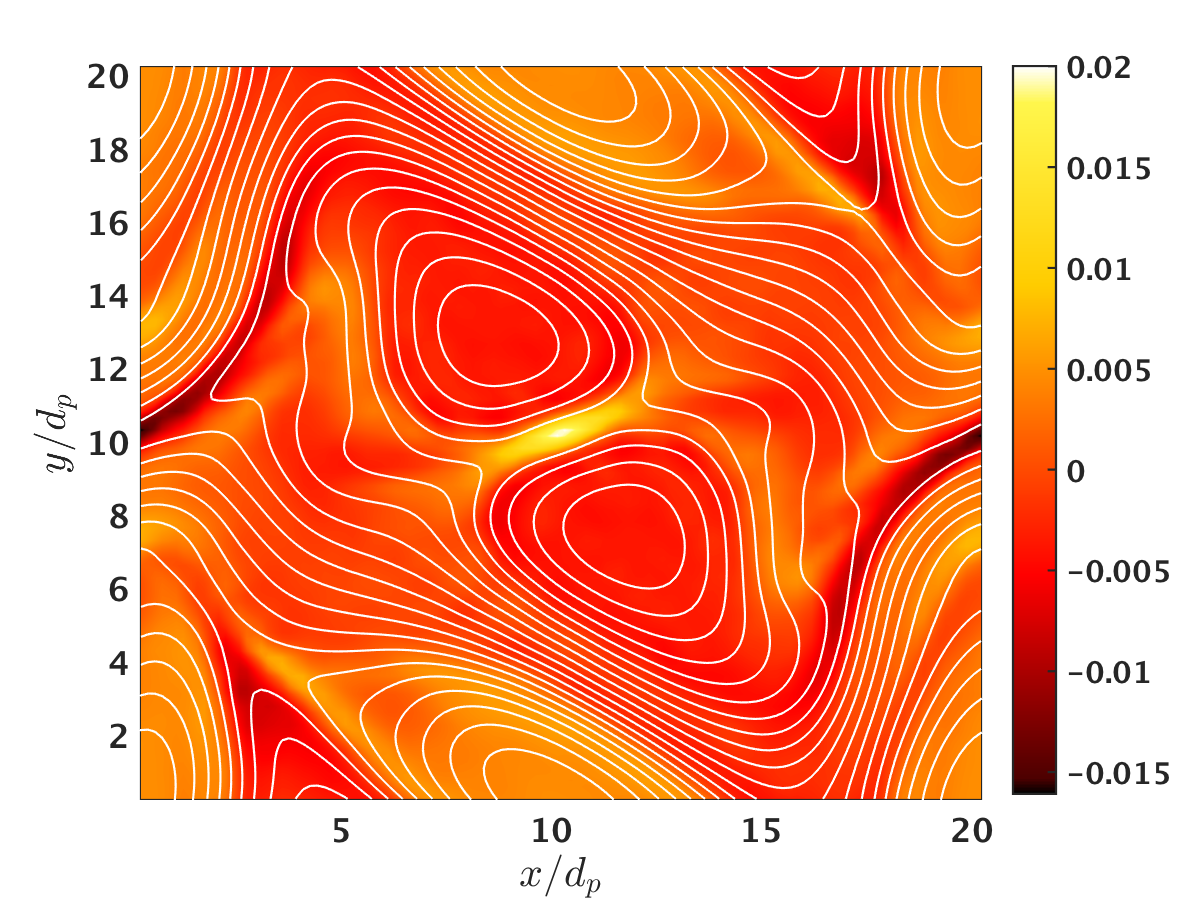}}
	\caption{The out of plane current $J_z$ with magnetic field contours superimposed at 20 $\Omega_{cp}^{-1}$ (top) and 40 $\Omega_{cp}^{-1}$(bottom) for 10 moment two fluid (left) and Vlasov-Maxwell (right) models from the Orszag-Tang vortex simulation. We note the difference in color bar for the Vlasov-Maxwell simulation, which shows clearly the currents are stronger by $\sim 20$ percent in the current layers which develop, for example, in the middle of the domain via the twisting motion of the flow and field at 20 $\Omega_{cp}^{-1}$. The currents in the Vlasov-Maxwell continue to grow in comparison to the fluid model at 40 $\Omega_{cp}^{-1}$.}
	\label{fig:OTFluidMoments}
\end{figure}
Though many large scale features between the two models agree as they should, there are key physical differences between the two models which highlight the more complete description of the plasma the Vlasov-Maxwell system provides. Since the 10 moment two fluid model can be thought of as a truncation of the Vlasov-Maxwell model, it is apparent from the larger currents in the Vlasov-Maxwell simulation that there is physics missing from the fluid model, either due to the higher fluid moments beyond the pressure tensor, or the closure of the fluid system, which causes the thinning current sheets in certain regions of the domain. It is also instructive to consider the energetics and spectra of the two simulations, plotted in Figures~\ref{fig:OTEnergetics} and \ref{fig:OTSpectra} respectively.
\begin{figure}
	\subfloat{\includegraphics[width=\textwidth]{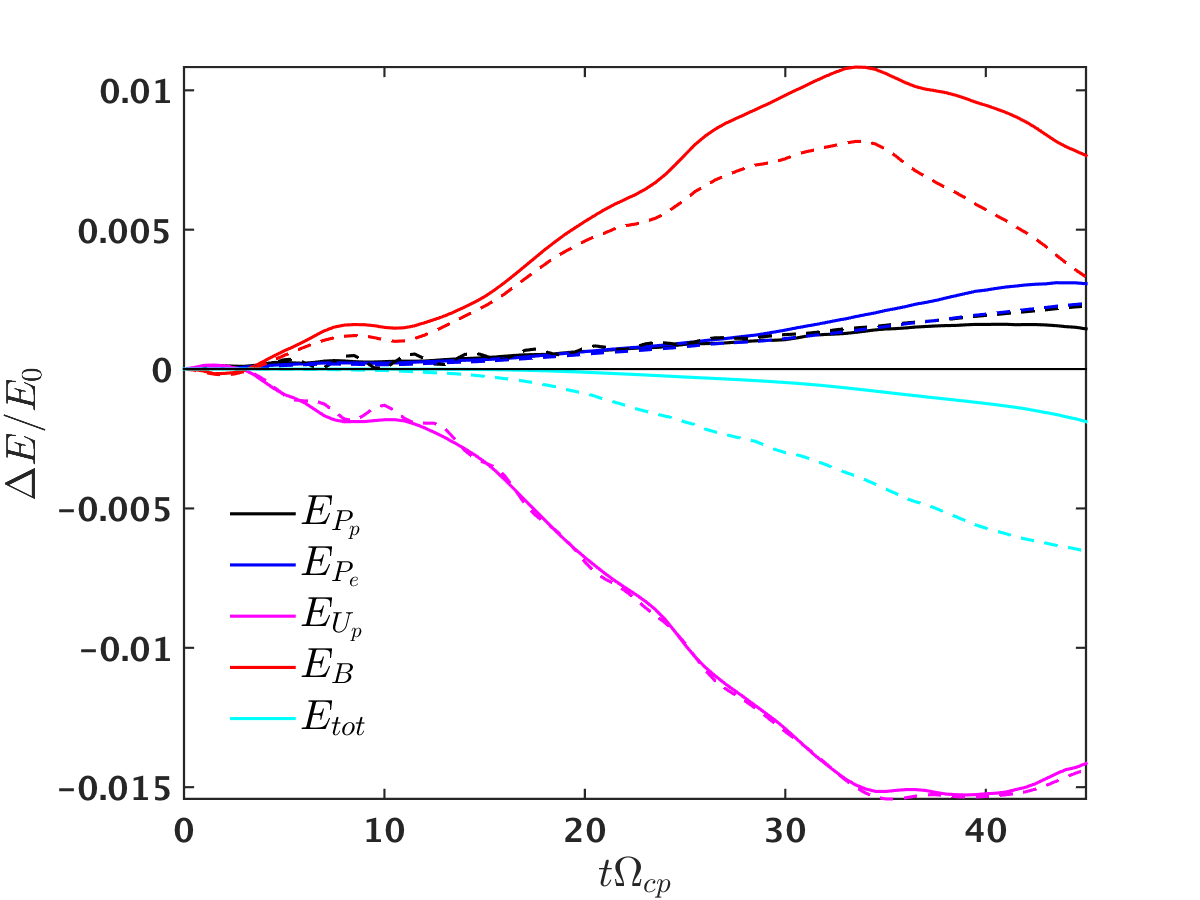}}
  	\caption{The evolution of various components of the energy, the proton thermal energy, $E_{P_p} = \frac{1}{2} m_p n_p v_{th_p}^2 $, the electron thermal energy, $E_{P_e} =  \frac{1}{2} m_e n_e v_{th_e}^2$, the proton flow energy, $E_{U_p} =  \frac{1}{2} m_p n_p |\mvec{u}_p|^2$, and the magnetic energy, $E_B = \frac{1}{2 \mu_0} |\mvec{B}|^2 $, along with the total energy, in the Orszag-Tang vortex simulation. Note that the quantities plotted are the relative changes in each of the energy components compared to the total energy at $t=0$, i.e, $\Delta E_{comp}/E_0$. Solid lines are the Vlasov-Maxwell simulation and dashed lines are the 10 moment two fluid simulation.}
	\label{fig:OTEnergetics}
\end{figure}
\begin{figure}
	\subfloat{\includegraphics[width=\textwidth]{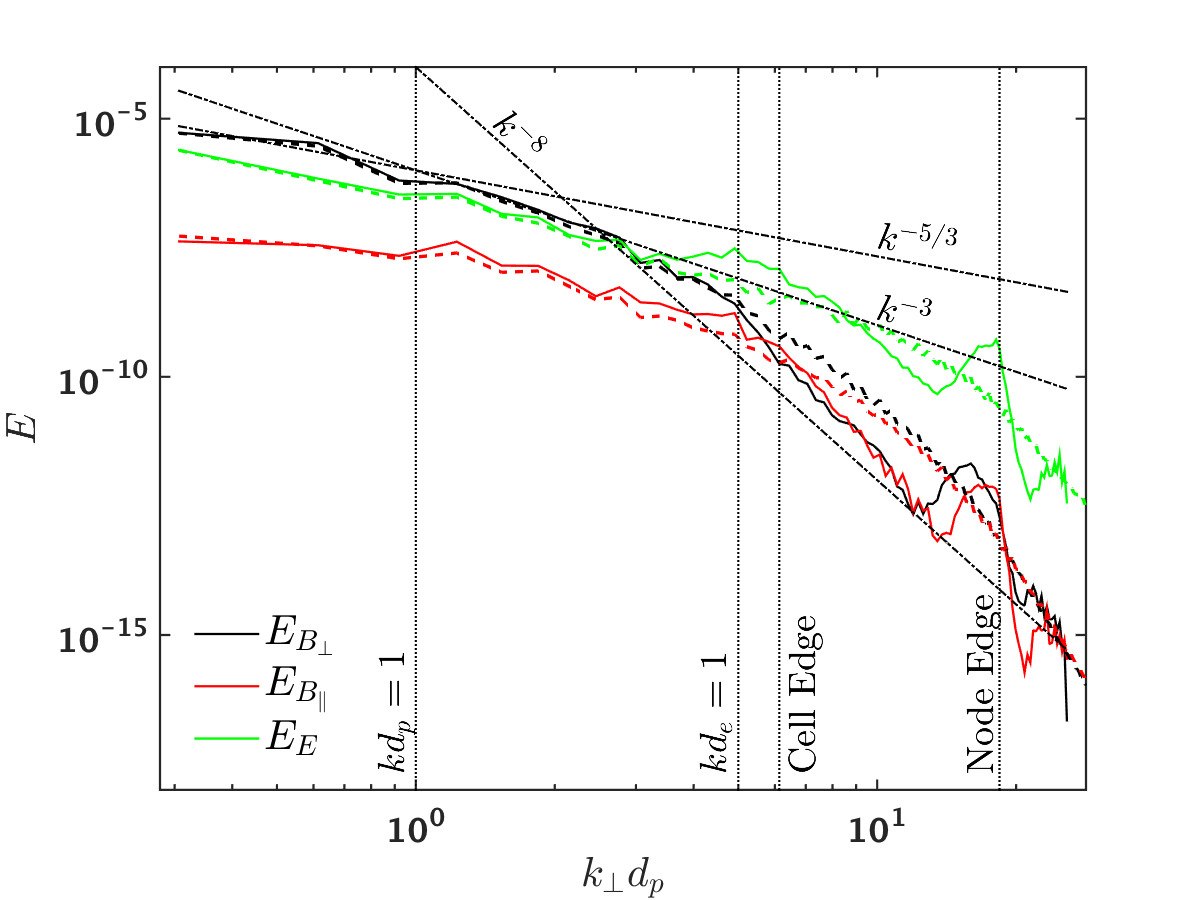}}
	\caption{The energy spectra for Vlasov (solid lines) and 10 moment two fluid (dashed lines) at $t = 35 \Omega_{cp}^{-1}$ of the Orszag-Tang vortex simulation. The perpendicular and parallel magnetic energy spectra are directions with respect to the magnetic field. The vertical dashed lines define a variety of scales where the spectra noticeably steepens, such as the proton inertial length, $d_p$, and the electron inertial length, $d_e$. In addition, we note the two limits of our resolution: the grid scale, and where our resolution cuts off from the polynomial representation in each cell, which we call the node edge. It is apparent that the polynomial representation in each cell does give resolution beyond the grid scale. The various slopes denote the fitted spectra at each of the spectral breaks. These compare well with previous results which suggest that at scales $k d_p < 1$, the spectra in $k_\perp$ is fitted by a -5/3 slope, before steepening at $k d_p > 1, k d_e < 1$, and steepening further at $k d_e > 1$. The build up in the Vlasov spectra at small scales is due to the fact that the collisionless system has no means of dissipating energy.}
	\label{fig:OTSpectra}
\end{figure}
Of course, one of the most interesting diagnostics remains: plots of the distribution function. Using a k-filter, where the distribution function is Fourier transformed, all Fourier coefficients except a single $k_x$ and $k_y$ are set to zero, and the distribution function is inverse Fourier transformed back to real space, one can consider how the distribution function is evolving at various scales. It is particularly noteworthy in Figure~\ref{fig:DistributionFunctionPlots} that a large amount of structure is contained in the Vlasov-Maxwell system. Such features could easily be smoothed away in a standard PIC simulation due to noise and coarse graining.
\begin{figure}
        \subfloat{\includegraphics[width=0.33\linewidth]{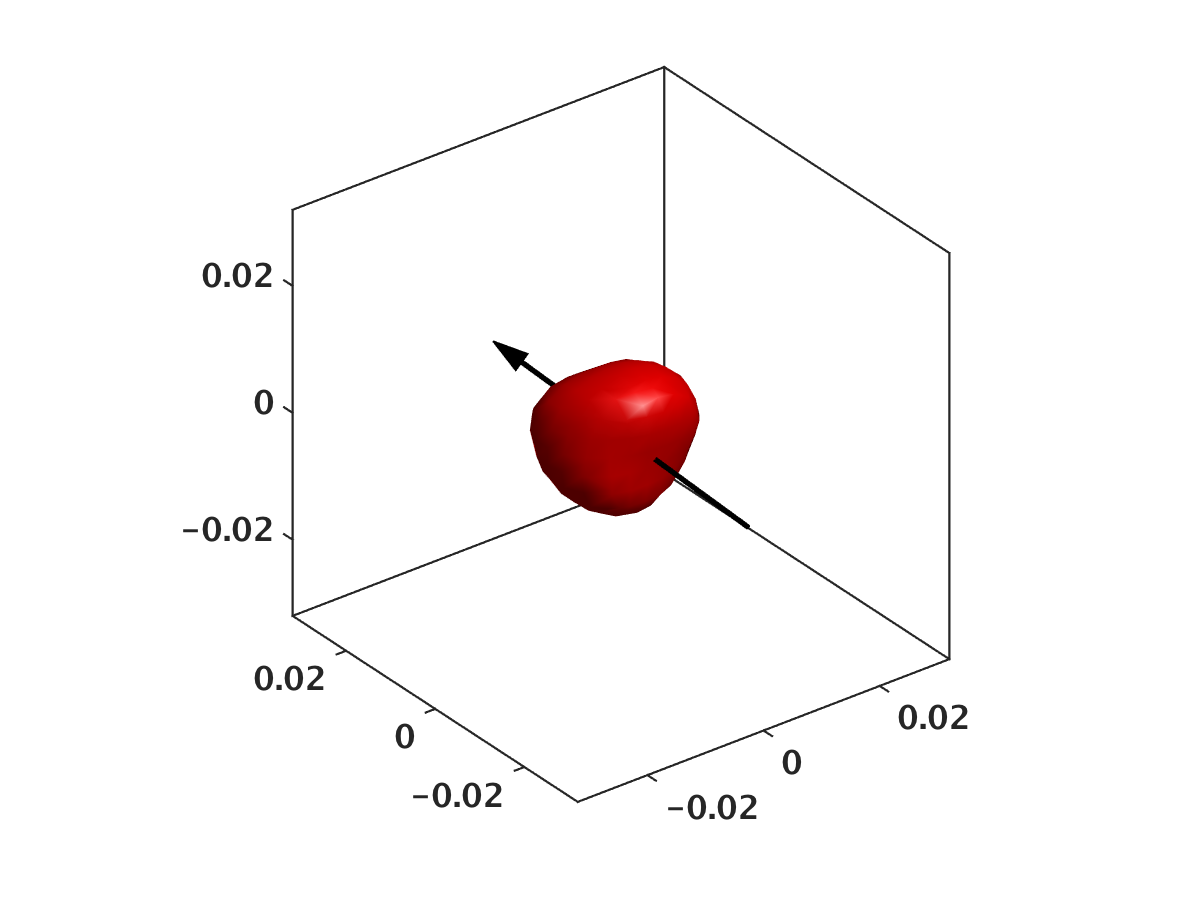}}
        \subfloat{\includegraphics[width=0.33\linewidth]{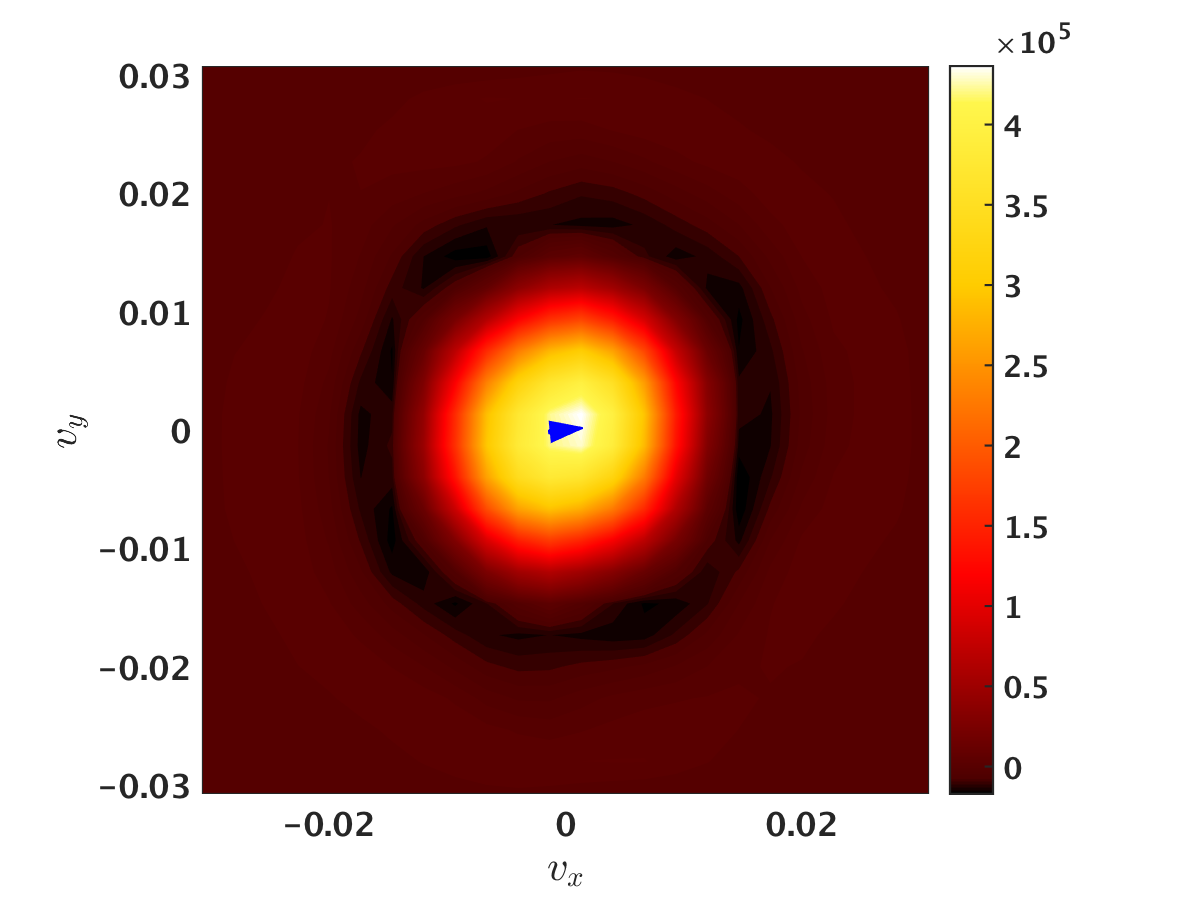}}
        \subfloat{\includegraphics[width=0.33\linewidth]{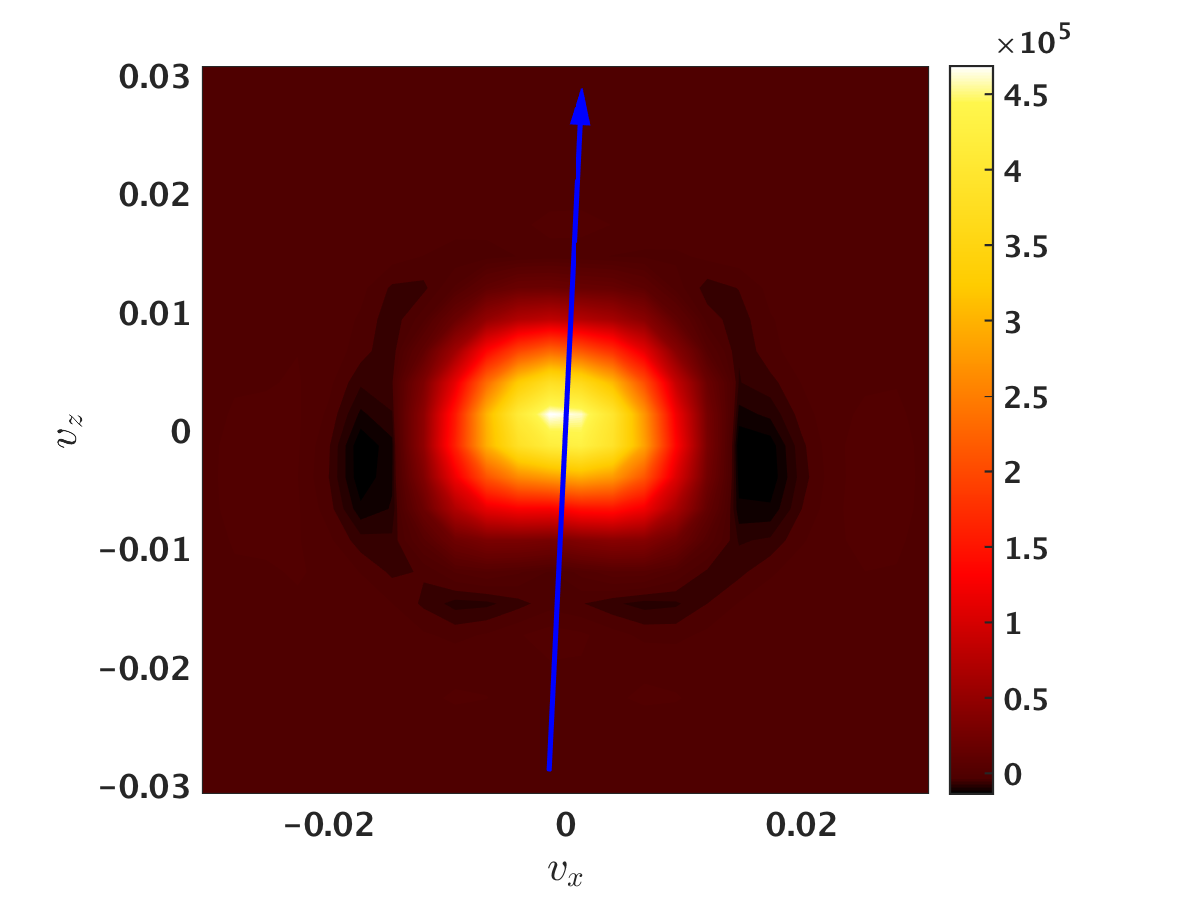}}\\
        
        \subfloat{\includegraphics[width=0.33\linewidth]{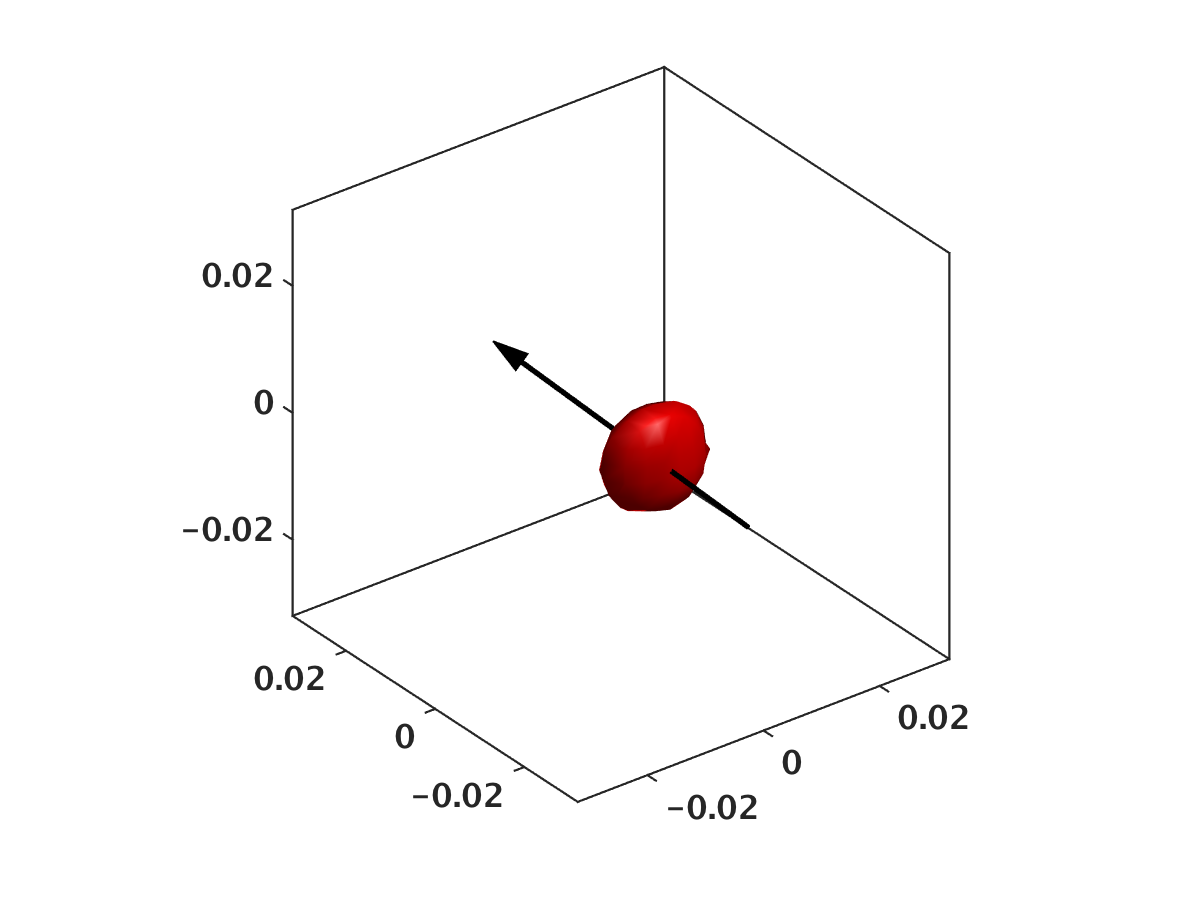}}
        \subfloat{\includegraphics[width=0.33\linewidth]{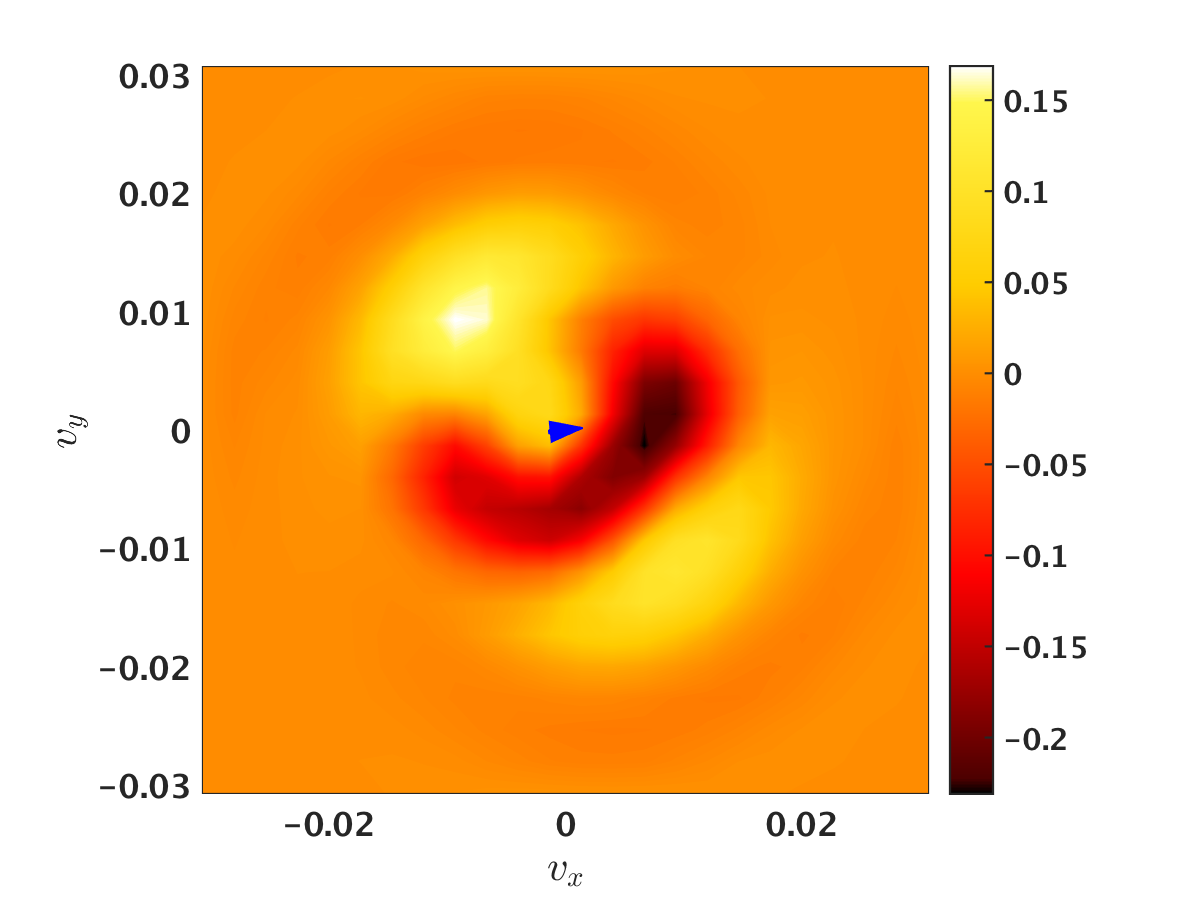}}
        \subfloat{\includegraphics[width=0.33\linewidth]{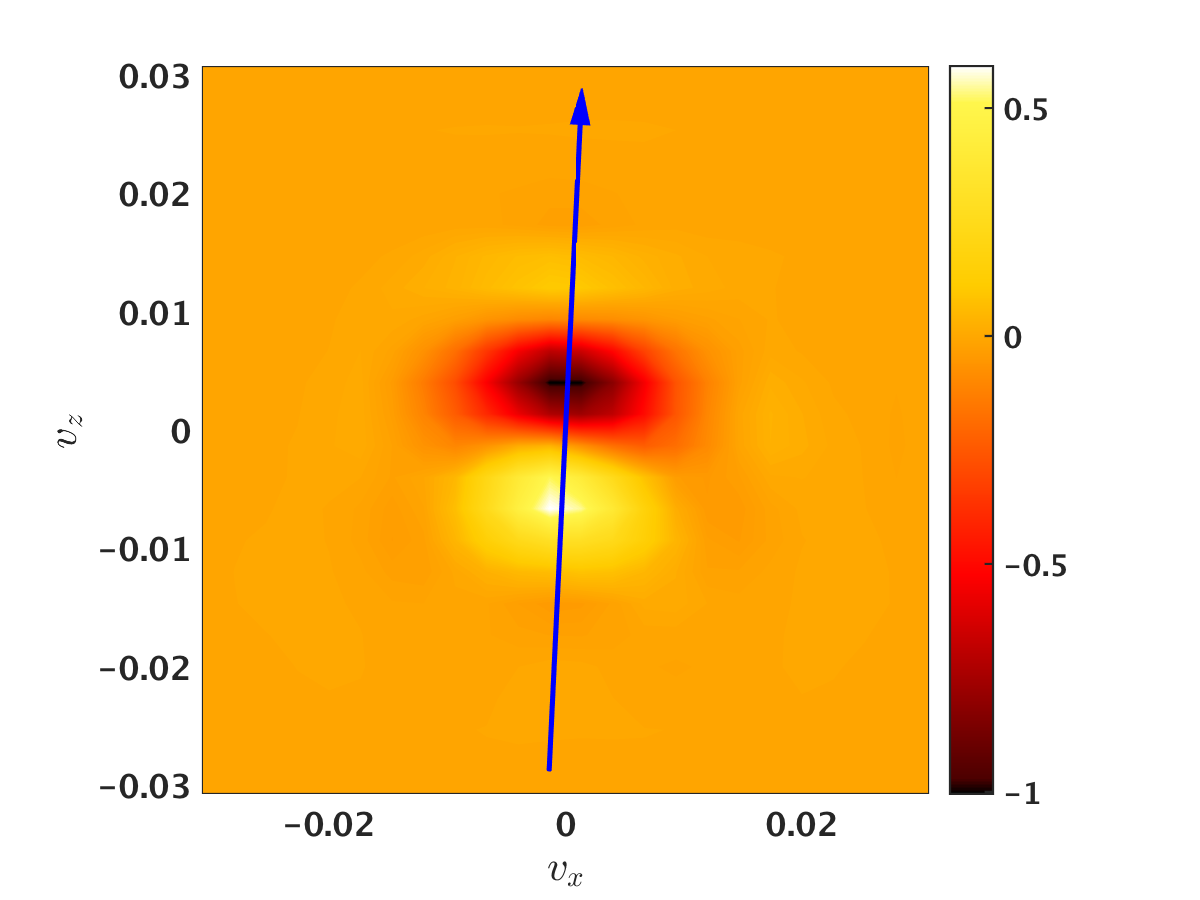}}\\
        
        \subfloat{\includegraphics[width=0.33\linewidth]{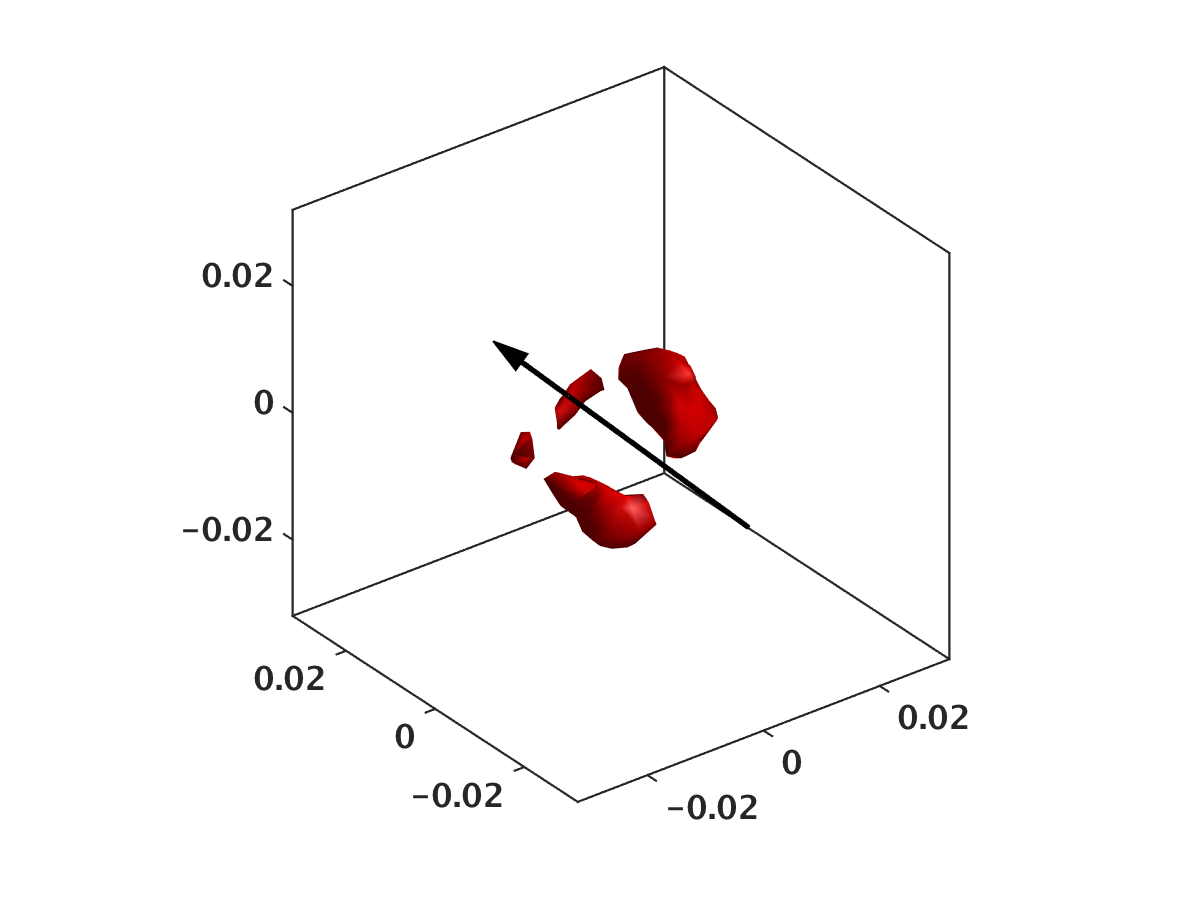}}
        \subfloat{\includegraphics[width=0.33\linewidth]{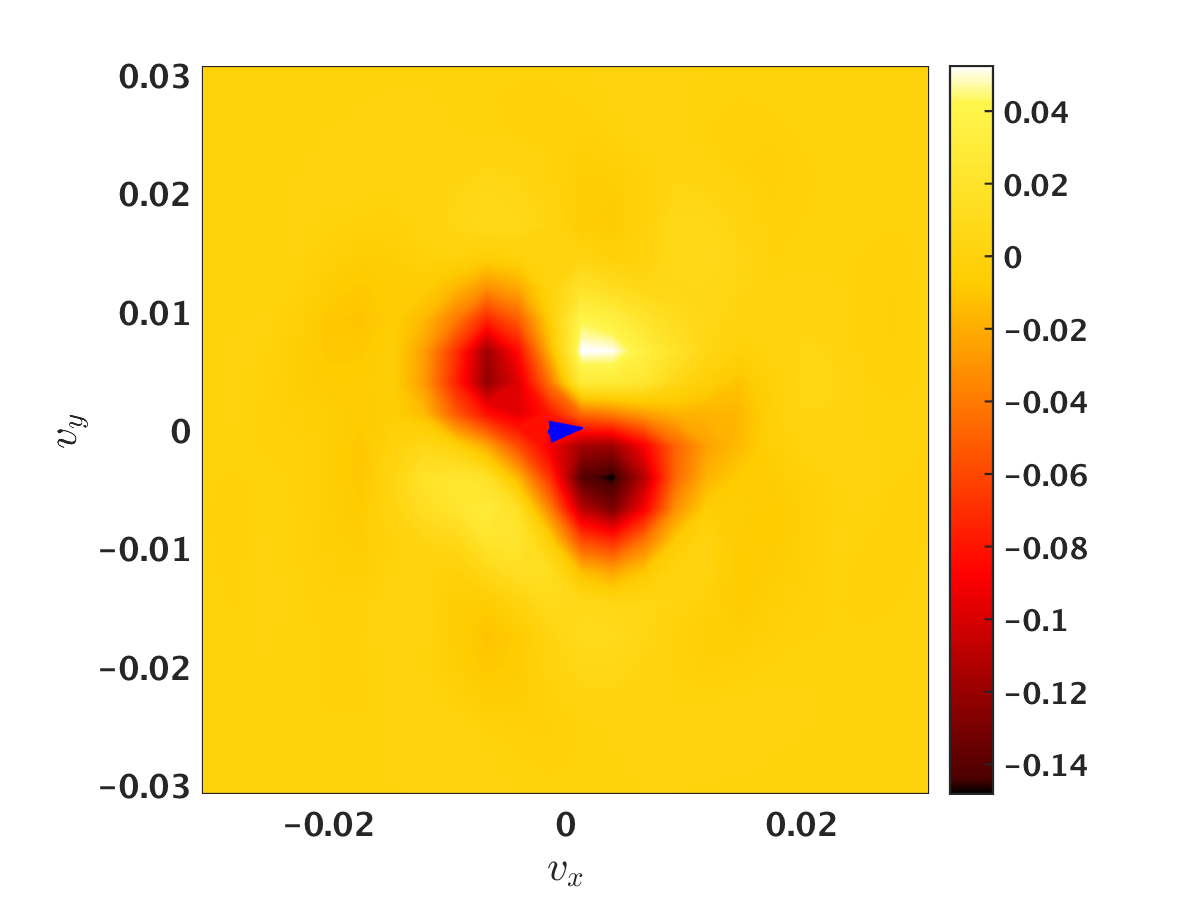}}
        \subfloat{\includegraphics[width=0.33\linewidth]{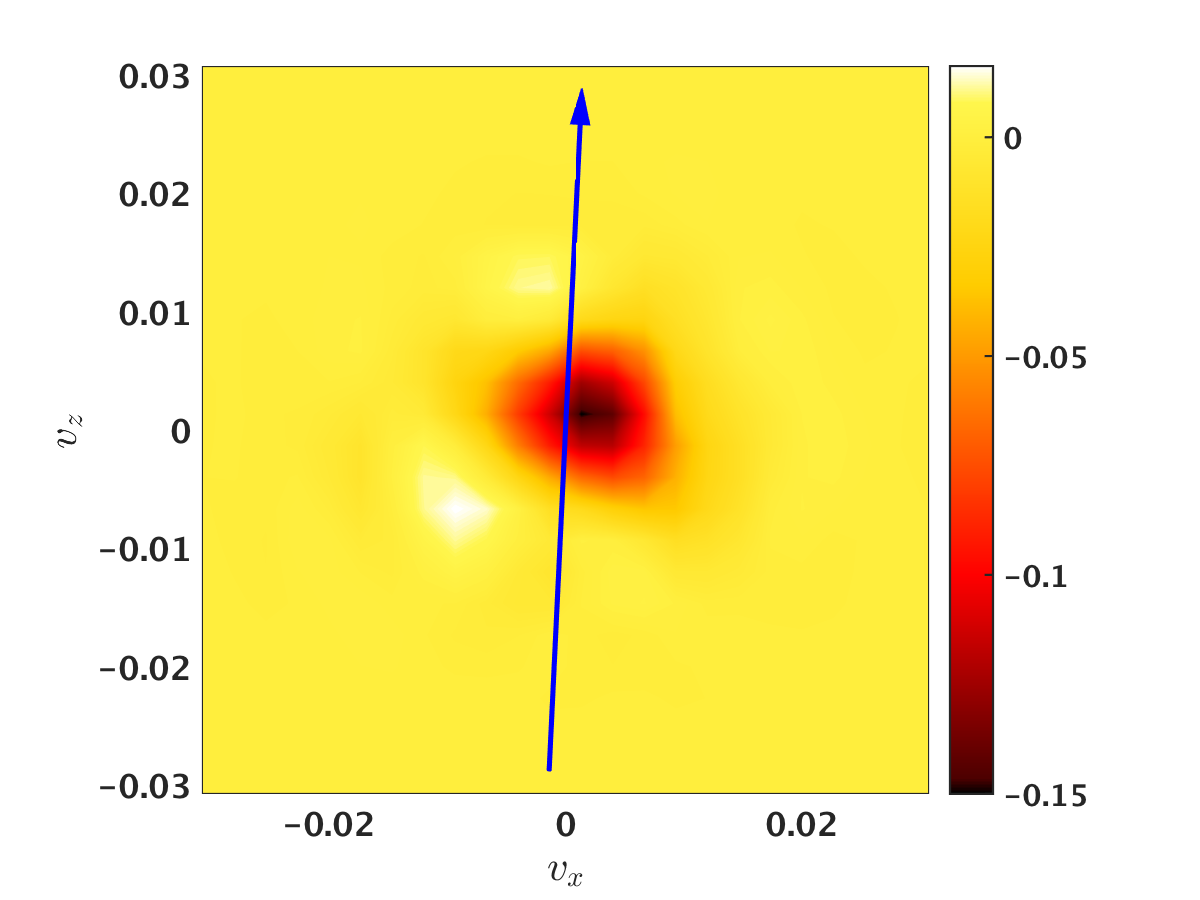}} \\

        \subfloat{\includegraphics[width=0.33\linewidth]{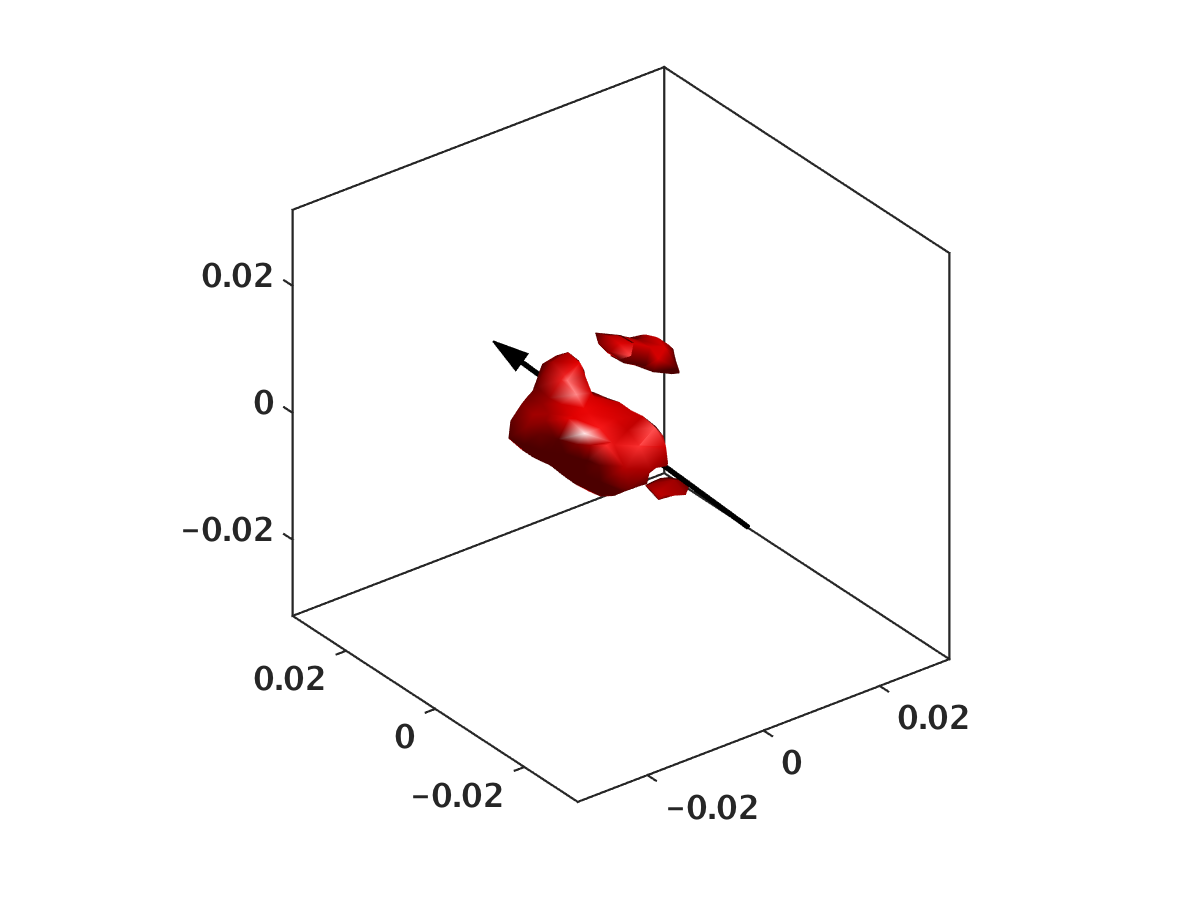}}
        \subfloat{\includegraphics[width=0.33\linewidth]{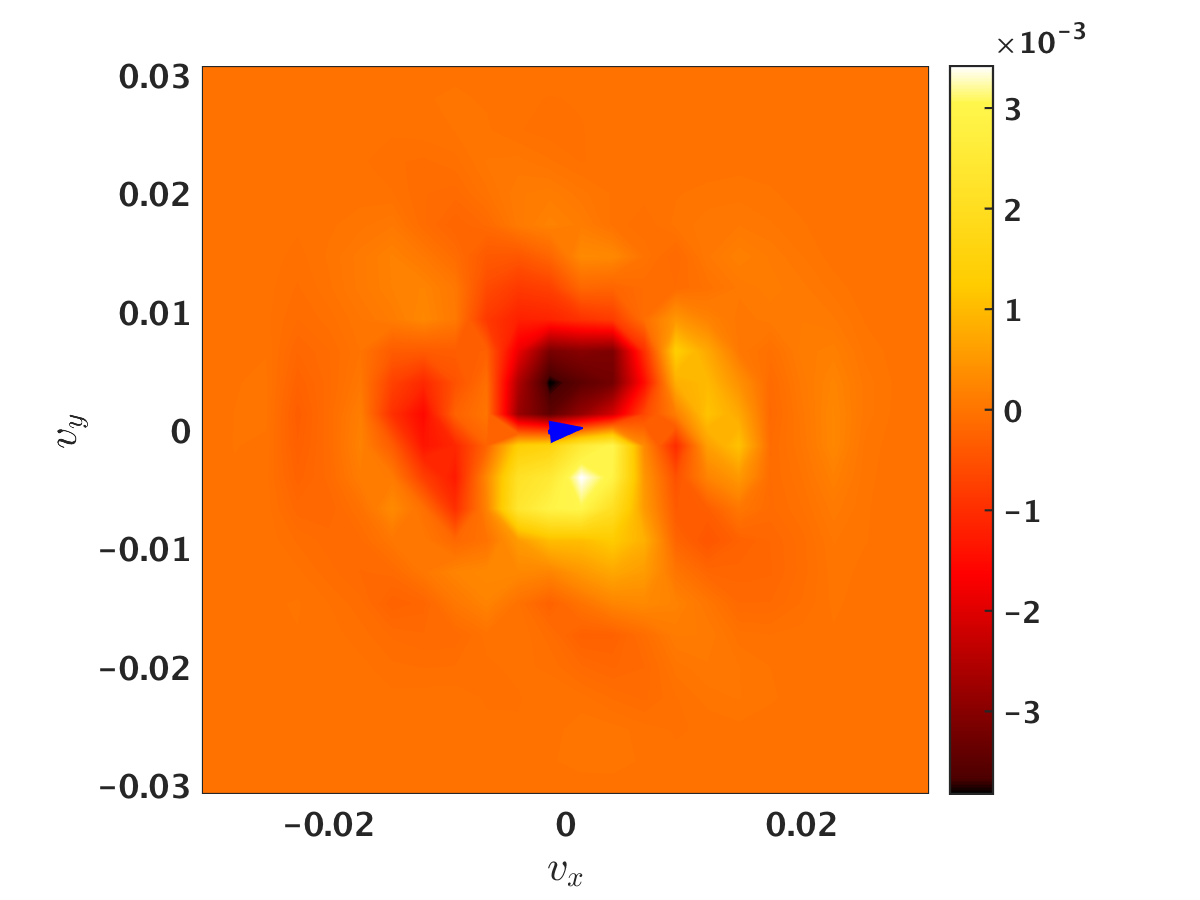}}
        \subfloat{\includegraphics[width=0.33\linewidth]{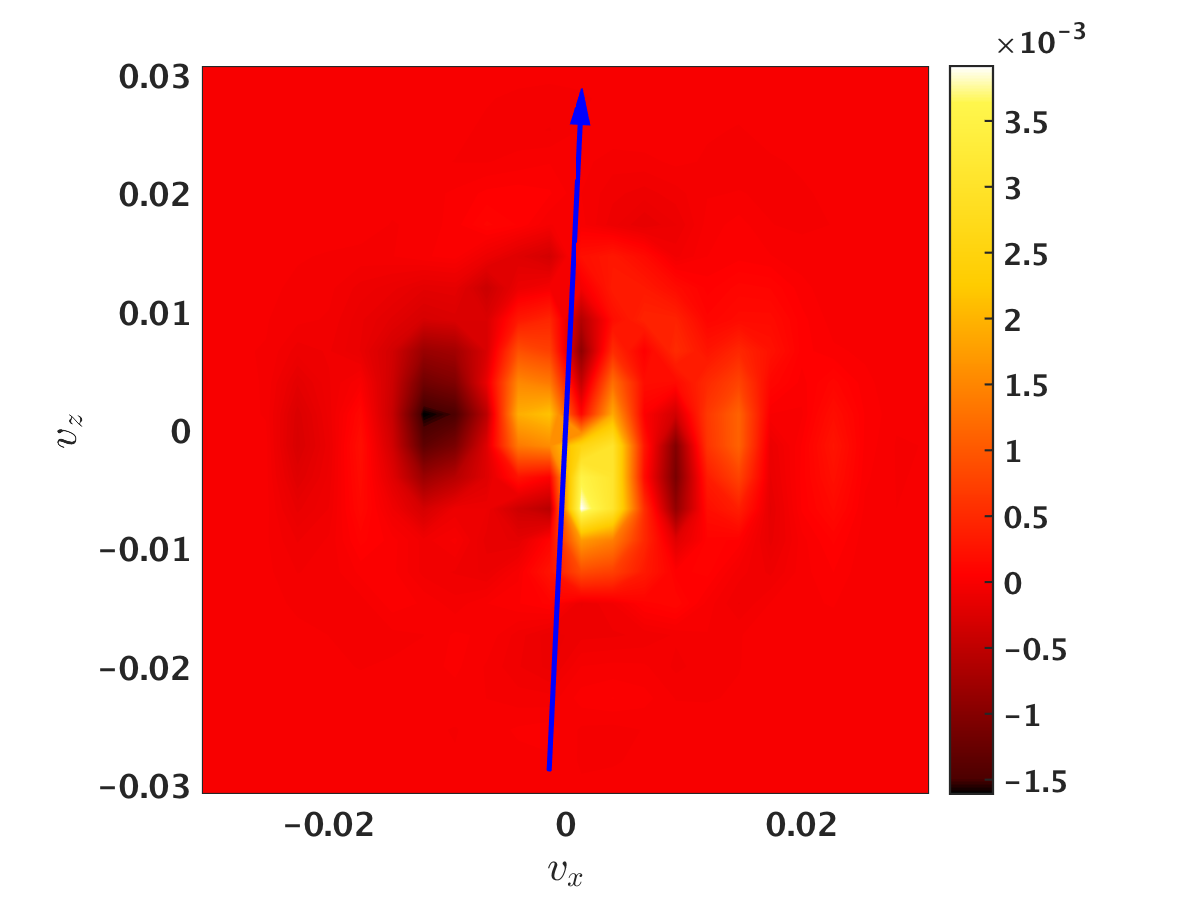}}
        \caption{Isosurfaces (left) and cuts through $v_z = 0$ (middle) and $v_y = 0$ (right) for the full proton distribution function (top) and the k-filtered distribution function at $k d_p \sim 1$ (top middle), $k d_p \sim 5$ (bottom middle), and $k d_p \sim 10$ (bottom) from the Orszag-Tang vortex simulation at $t= 40 \Omega^{-1}_{cp}$ and $(x,y) = (10.24 d_p, 10.24 d_p)$. The most striking feature of the k-filtered distribution functions is the level of signal displayed by the Vlasov-Maxwell solver, with eight orders of magnitude from the full distribution function to the smallest scale k-filter. There are small positivity violations in the full distribution function, but note that the distribution function being negative after k-filtering is not alarming, as k-filtering is a proxy for $\delta f$.}
        \label{fig:DistributionFunctionPlots}
\end{figure}
\end{par}
\begin{par}
A natural follow up question is then, what does such signal provide? Beyond the ability to more easily compare with observations and experiments if performing verification studies of the solver, having noise-free distribution functions increases the number of tools available for analysis of a simulation. For example, consider a diagnostic called a field-particle correlation \citep{Howes2017, Klein2016, Klein2017}. There are three field-particle-correlations which we will consider here, defined as follows:
\begin{align}
C_1(\mvec{v}, t, \tau) & = C(q_s v_\parallel f(\mvec{x}_0, \mvec{v}, t), E_\parallel(\mvec{x}_0, t)), \label{eq:LandauProxy} \\
C_2(\mvec{v}, t, \tau) & = C\left(-v_\parallel \frac{m_s v_\perp^2}{2 |\mvec{B}|} f(\mvec{x}_0, \mvec{v}, t), B_\parallel(\mvec{x}_0, t)\right), \label{eq:TransitProxy} \\
C_3(\mvec{v}, t, \tau) & = C(q_s v_\perp f(\mvec{x}_0, \mvec{v}, t), E_\perp(\mvec{x}_0, t)) \label{eq:CyclotronProxy},
\end{align}
where $\parallel, \perp$ are directions with respect to the local magnetic field, $\tau$ is the correlation time, and $\mvec{x}_0$ is the configuration space point at which the correlation is being performed. The function $C(\cdot, \cdot)$ is defined as,
\begin{align}
C(f(\mvec{x}_0, \mvec{v}, t_i), g(\mvec{x}_0, t_i)) = -\frac{1}{N} \sum_{j=i}^{i+N} f(\mvec{x}_0, \mvec{v}, t_j) g(\mvec{x}_0, t_j).
\end{align}
The correlation time $\tau$ determines the summation extents, i.e., $\tau = N dt$, $dt = t_j - t_{j-1}$, while the index $i$ determines the specific time at which the correlation is being measured. In other words, we determine the correlation at time $t_i$ over a specified correlation time $\tau$ by partitioning the correlation time $\tau$ into $N dt$ chunks. The power of this diagnostic lies in its ability to ascertain secular energy exchange between the particles and fields, which may be occurring amidst complicated dynamics such as the broadband turbulence found in the Orszag-Tang set up. For example, the first correlation, \eqr{\ref{eq:LandauProxy}}, provides a proxy for determining whether Landau damping or parallel acceleration may be taking place at the specified point. Landau damping of a particular wave mode can be determined with a judicious choice of correlation time. The other two correlations, Eqns.\thinspace(\ref{eq:TransitProxy})-(\ref{eq:CyclotronProxy}), provide proxies for two other types of wave particle resonances: transit time damping and cyclotron damping. There are two ways to represent this data, as stack plots in $t-v_\parallel$, for analyzing how much parallel phase mixing may be occurring, or in $v_\perp-v_\parallel$ at a specified time as a way of ascertaining where the wave-particle resonances may be occurring, and thus narrowing down what wave modes may be present in the simulation. These two ways are plotted in Figures~\ref{fig:OTStackPlots} and \ref{fig:OTvperpvparallel}. The level of signal is striking, as determining this amplitude of a correlation could be challenging in a standard particle-in-cell simulation without a sizable number of particles per cell, especially given the gradients in velocity space of the distribution function present in the field-particle-correlation definitions, which tend to exacerbate counting noise issues. A future publication will explore the significance of the correlations presented here.
\begin{figure}
	\subfloat{\includegraphics[width=\textwidth]{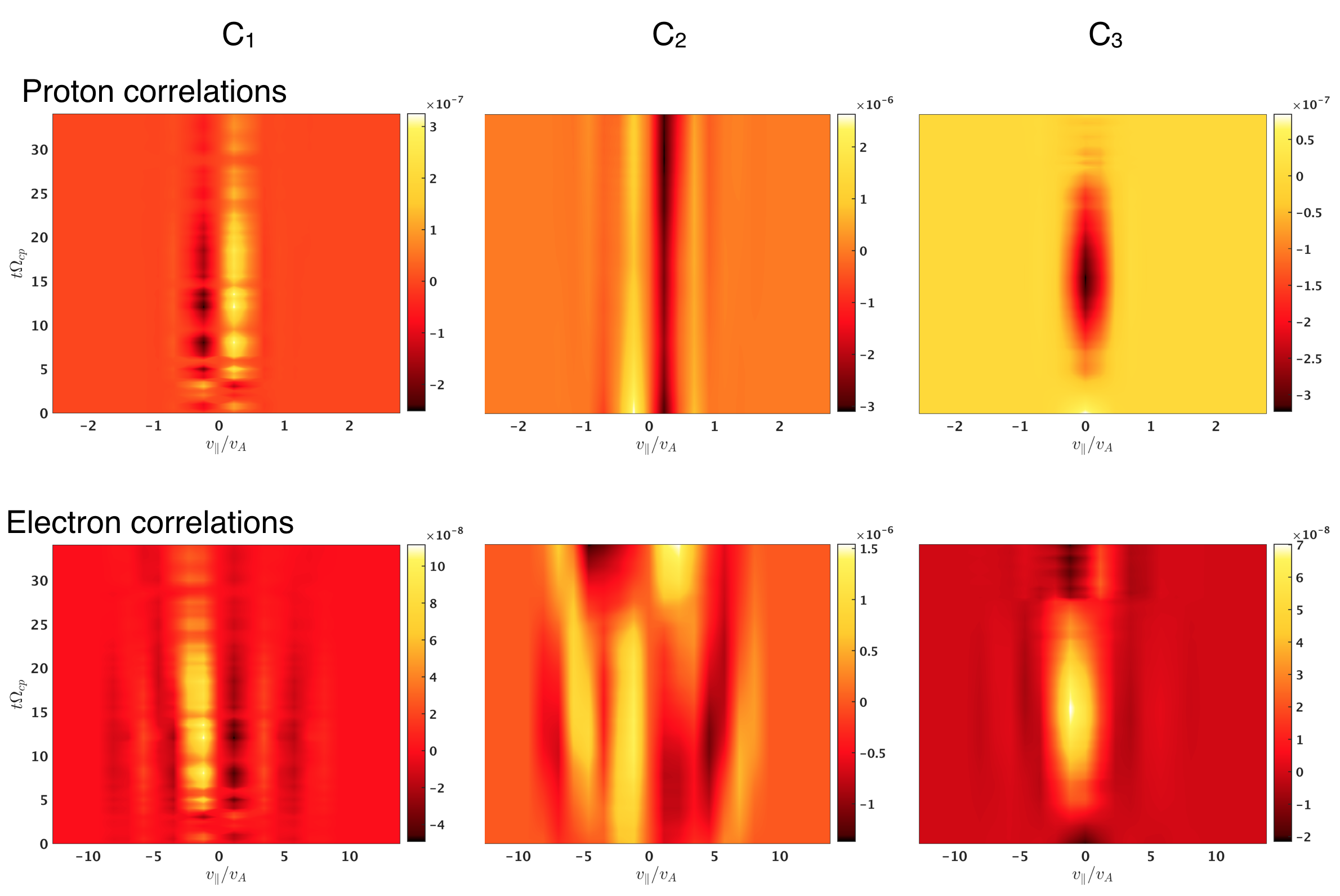}}
	\caption{Stack plots of $C_1$, \eqr{\ref{eq:LandauProxy}}, (left), $C_2$, \eqr{\ref{eq:TransitProxy}}, (middle), and $C_3$ , \eqr{\ref{eq:CyclotronProxy}}, (right) for the protons (top) and electrons (bottom) at $(x,y) = (10.24 d_p, 10.24 d_p)$ and $v_\perp = 0$ for a correlation time of $\tau = 15 \Omega_{cp}^{-1}$ in the Orszag-Tang vortex simulation.}
	\label{fig:OTStackPlots}
\end{figure}
\begin{figure}
	\subfloat{\includegraphics[width=\textwidth]{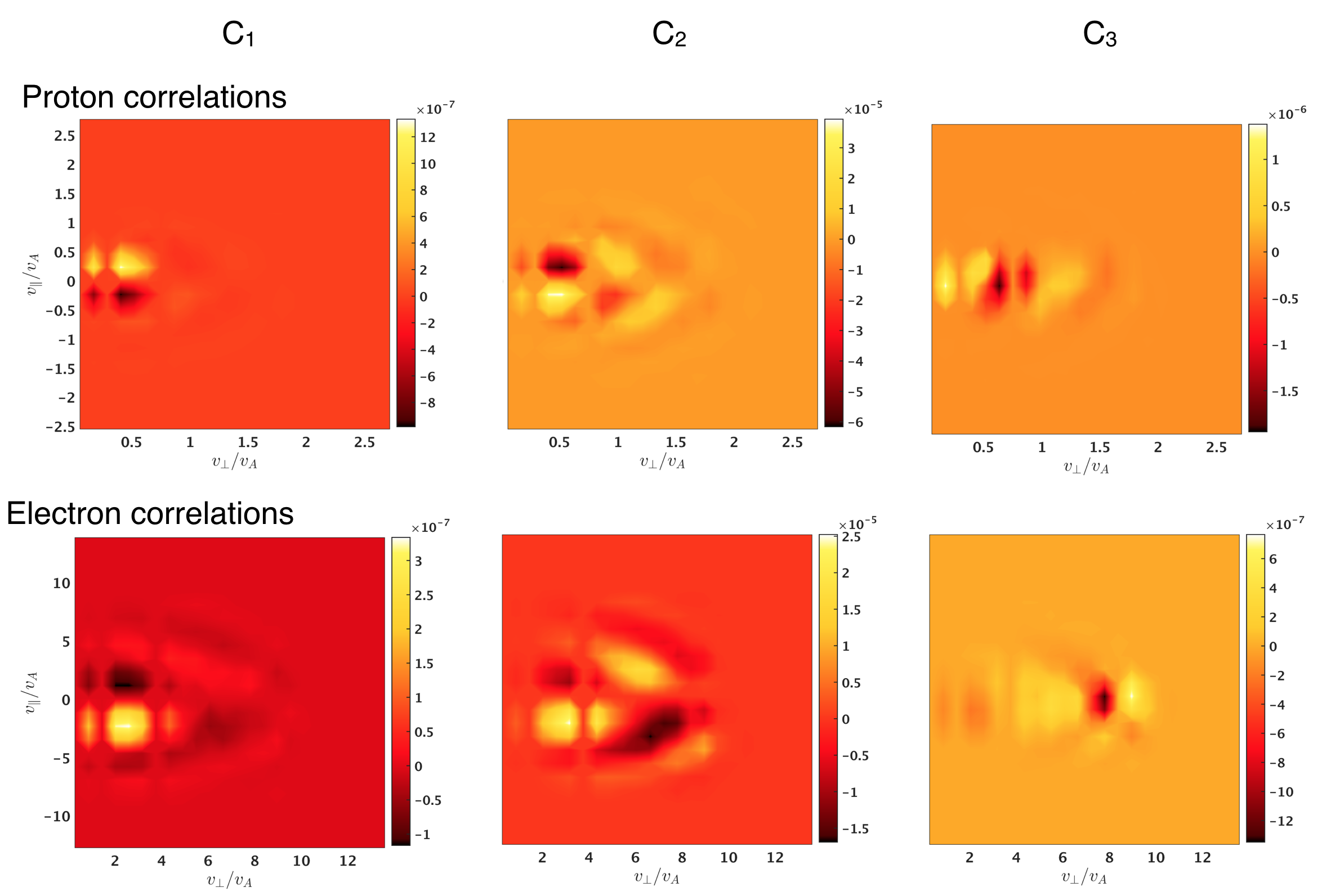}}
	\caption{$v_\perp-v_\parallel$ of $C_1$, \eqr{\ref{eq:LandauProxy}}, (left), $C_2$, \eqr{\ref{eq:TransitProxy}}, (middle), and $C_3$ , \eqr{\ref{eq:CyclotronProxy}}, (right) for the protons (top) and electrons (bottom), at $(x,y) = (10.24 d_p, 10.24 d_p)$ and $t = 30 \Omega_{cp}^{-1}$ for a correlation time of $\tau = 15 \Omega_{cp}^{-1}$ in the Orszag-Tang vortex simulation.}
	\label{fig:OTvperpvparallel}
\end{figure}
\end{par}
\begin{par}
We wish to make a couple final notes here about the Orszag-Tang simulation. One is that, out of all the benchmarks presented, it is perhaps most readily apparent that the small positivity violations described as inevitable in Section \ref{sec:ConservationTests} are present in this simulation, e.g. see the upper panels of Figure~\ref{fig:DistributionFunctionPlots}. We reiterate that this behavior is, of course, unphysical. We also want to emphasize again that the fluid moments such as the density and energy remain positive definite for the entire simulation in spite of these positivity violations in the distribution function, thus further adding credibility to our claim that small positivity violations do not pose as severe of a numerical threat as positivity violations do in standard fluid codes. Nonetheless, the unphysical nature of a negative distribution function, if sufficiently large, is enough to give one pause about any analysis one wishes to perform. Unfortunately, as we move to higher and higher dimensions, increasing the velocity resolution dramatically to reduce these positivity violations becomes a more intractable solution. Limiters which preserve the positivity of particular quantities in fluid algorithms are incredibly common, but due to the fact that conservation relations are implicit in the Vlasov-Maxwell system, fixing positivity in a kinetic system has additional subtleties. In Cagas \emph{et al.} \citep{Cagas:2017b}, they employ the Moe-Rossmanith-Seal limiter \citep{Moe2015} with the Vlasov-Maxwell solver,
\begin{align}
f^*(\mvec{z}) = f_0 + \theta (f(\mvec{z}) - f_0),
\end{align}
where $f_0$ is the cell average of the function $f(\mvec{z})$, and $\theta$ is defined as,
\begin{align}
\theta = \min\left(1, \frac{f_0}{f_0 - \min_{\mvec{z}} f(\mvec{z})} \right ). 
\end{align}
Such a limiter shifts the distribution function upward in a cell to guarantee positivity while keeping the cell average, and thus the number of particles, constant. However, energy conservation will be broken by this operation, and it could additionally make momentum non-conservation worse. In fact, Cagas \emph{et al.} \citep{Cagas:2017b} notes that the operation of fixing positivity gives rise to one to two percent energy conservation errors due to the increase in the particle energy. Conserving momentum and energy after the application of such a limiter can be done by scaling the distribution function by an additional function
\begin{align}
\pfrac{f}{t} = \alpha \gvs \cdot ((\mvec{v} - \mvec{u}) f),
\end{align}
where 
\begin{align}
\alpha = -\frac{1}{2} \ln \left (\frac{\epsilon}{\epsilon^*} \right ),
\end{align}
and $\epsilon$ is the second moment of the distribution function prior to the application of the limiter, while $\epsilon^*$ is the second moment of the rescaled distribution function after fixing positivity violations. Unfortunately, such an operation is nonlocal in velocity space, and it is as yet unclear what effect such an operation has on the numerical stability of the algorithm. A more natural limiter may involve modification of the differential operator, as opposed to the solution, to prevent the differential operator from ever generating unphysical positivity violations. By limiting the solution via changes to the differential operator, we could retain our conservation relations, but additional testing would be needed to determine if a modified operator still has the same convergence properties and produces the correct physical behavior of the Vlasov-Maxwell system. We defer a systematic study of limiters for kinetic systems to a future publication and for now comment that the physics of the Orszag-Tang simulation does not appears to be significantly impacted by the presence of positivity violations in the distribution function.
\end{par}
\begin{par}
The second note is that, on the coarse resolution explored here, it is even more important to discuss the divergence errors which will arise from how we have discretized Maxwell's equations. We focus in particular on the errors in the divergence of the magnetic field, since for the plasma parameters employed, the dynamics of the Orszag-Tang turbulence are magnetically dominated. In Figure \ref{fig:divBErrors}, we plot the divergence of the magnetic field normalized to the maximum of the curl of the magnetic field. The choice of normalization is due to how important the curl of the magnetic field is for the development of the current layers seen in the evolution of the out-of-plane current plotted in Figure \ref{fig:OTFluidMoments}.
\begin{figure}
\centerline{
    \setkeys{Gin}{width=0.5\linewidth}
	\includegraphics{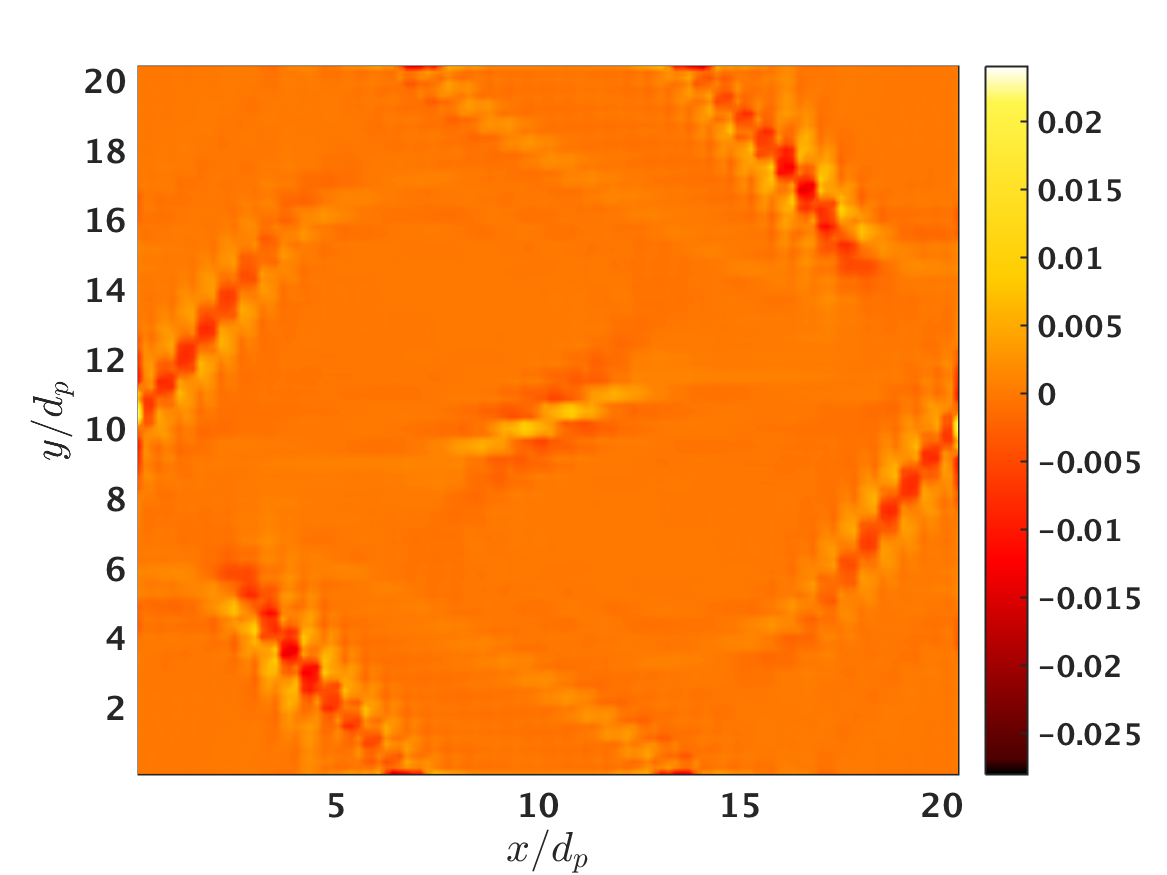}
	\includegraphics{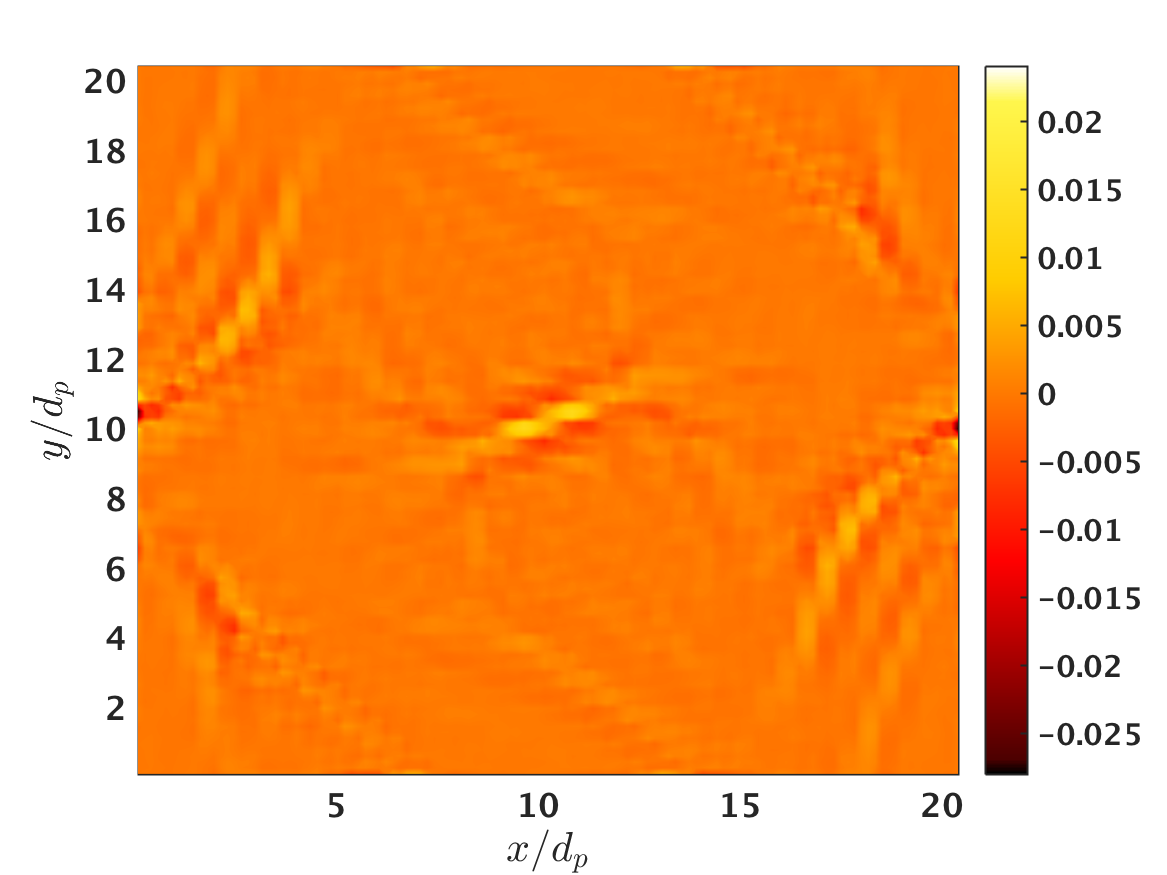}
	}
	\caption{The divergence of the magnetic field normalized to the maximum of the curl of the magnetic field at 20 $\Omega_{cp}^{-1}$ (left) and 40 $\Omega_{cp}^{-1}$ for the Orszag-Tang vortex. We note that the divergence errors appear to be small compared to the dynamics of interest and thus likely do not affect the physics.}
	\label{fig:divBErrors}
\end{figure}
Since the errors in Figure \ref{fig:divBErrors} appear to be small, at least when compared to a quantity of dynamical interest, we thus expect that these divergence errors, like the small positivity violations, have not significantly impacted the physics present in this turbulence simulation. However, it is worth acknowledging how we might eliminate, or at least further suppress, these sorts of errors since $\sim$ 2 percent errors are not completely trivial. One method of mitigating divergence errors is through so-called divergence cleaning schemes, such as Maxwell's equations in perfectly hyperbolic form \citep{munz_2000, munz_2000b, munz_2000c}. This system of equations adds the two divergence relations as dynamical equations, with their time evolution represented by correction potential which contain the accumulated errors in the divergence of the electric and magnetic fields,
 \begin{align}
   \frac{\partial \mvec{B}}{\partial t} + \nabla\times\mvec{E} + \gamma \nabla \psi &= 0, \label{eq:PHdbdt} \\
   \epsilon_0\mu_0\frac{\partial \mvec{E}}{\partial t} - \nabla\times\mvec{B} + \chi \nabla \phi &= -\mu_0\mvec{J}, \label{eq:PHdedt} \\
   \frac{1}{\chi} \pfrac{\phi}{t} + \nabla\cdot\mvec{E} &= \frac{\varrho_c}{\epsilon_0}, \label{eq:PHdivE} \\
   \frac{\epsilon_0 \mu_0}{\gamma} \pfrac{\psi}{t} + \nabla\cdot\mvec{B} &= 0, \label{eq:PHdivB}
 \end{align}
where $\gamma, \chi$ are dimensionless parameters which determine the speed at which errors propagate in the system. While this scheme is perfectly reasonable in vacuum, where the charge density, $\rho_c$, is zero, for a plasma, the correction potentials may modify the total energy, and thus the conserved quantity, or worse, they may modify the Lorentz force in a non-trivial way. Alternatively, reconstruction methods have been quite fruitful in magnetohydrodynamics \citep{Balsara2001614, Balsara2004, Balsara20095040}, where one only worries about the divergence of the magnetic field, and they have been recently extended to finite volume discretizations of relativistic two-fluid equations, as long as one explicitly evolves an equation for the charge density  \citep{Balsara2016357}. Finally, extensions to curl-conforming methods such as Nedelec finite elements, popular in the continuous Galerkin community, have been very recently achieved for both the discontinuous Galerkin method and equation systems with explicit sources, such as the Vlasov-Maxwell system \citep{Campos2016403, campospinto:hal-01303852, campospinto:hal-01303861}. These latter methods are particularly exciting, as they present a possible method of completely avoiding divergence errors, as opposed to just mitigating them, by virtue of the DG discretization containing a discrete version of the continuous vector identity $\nabla \cdot (\nabla \times \mvec{A}) = 0$, where $\mvec{A}$ is some vector field.
\end{par}
\section{Summary and Future Work}\label{sec:Summary}
\begin{par}
In this paper, we have described a new algorithm for the numerical solution of plasma kinetic equations. We have demonstrated that the semi-discrete discontinuous Galerkin (DG) discretization of the Vlasov-Maxwell system conserves particles exactly and conserves energy exactly when central fluxes are used for Maxwell's equations, independent of the choice of numerical flux function for the Vlasov equation. We have discussed at length an efficient implementation of the kinetic algorithm, which significantly reduces the number of multiplications required in a nodal basis expansion. This economic formulation of the algorithm scales incredibly well on standard computing architecture and has been used to verify the solver against common benchmarks, both electrostatic and fully electromagnetic, from only one configuration and velocity dimension, to a five dimensional calculation of turbulence.
\end{par}
\begin{par}
There remain many avenues of future work to further improve the Vlasov-Maxwell solver. As discussed in Section \ref{sec:TimeStepping}, the kinetic equation's CFL condition can be thought of as two distinct restrictions on the time step, a maximum velocity and a maximum acceleration, where the latter CFL can quite easily be more restrictive than the former. For instance, in the five dimensional turbulence run, because the magnetic field's amplitude is so large, the acceleration CFL is as much as two times more restrictive than the speed of light CFL with the given parameters. This constraint on the time-step arises despite being significantly non-relativistic, since $v_{Ae}/c = 0.1$. Splitting the update into a configuration space and velocity space update, and treating just the velocity space update implicitly, would pay huge dividends. Further, as we hinted at in the summary of the algorithm, a reformulation of the DG update in terms of a \emph{modal} update is an interesting path to pursue. While the nodal expansion often provides the simplest means of generalizing to complicated geometries and very sophisticated boundary conditions, for many astrophysical systems, simplicity is usually easily justified. As such, a modal expansion on structured, Cartesian grids, where a Gram-Schimdt orthogonalization procedure could be employed to produce an orthonormal basis expansion would produce a dramatic reduction in the number of operations required to perform the update of the kinetic equation. In doing so, it may be possible to reduce the scaling of the algorithm from $\mathcal{O}(N_p^2)$ to $\mathcal{O}(N_p)$ for the configuration space component of the update and $\mathcal{O}(N_q N_p)$ to $\mathcal{O}(N_c N_p)$ for the velocity space component of the update. In this case, $N_c$ is the number of basis functions in the expansion in configuration space for the DG discretization of Maxwell's equations. In fact, we find in a preliminary test of a modal basis expansion of the Vlasov-Maxwell system that these scaling do hold, encouraging further study of this alternative formulation. Traditionally, this approach is not justified because of how expensive the modal surface updates can become in 2D and 3D, but our preliminary tests additionally reveal a turnover from 3D to 4D, wherein even the surface integrals are much cheaper with the orthogonal modal basis. Other improvements to the solver can come from hardware considerations. Current computing paradigms are quickly moving towards many-core devices, which dramatically increase the number of floating point operations the architecture can do per second at the cost of available memory. This paradigm favors arithmetically intense algorithms like DG, as long as one embraces hybrid parallel programming models so that memory can be shared amongst cores and thus shared amongst floating point units.
\end{par}
\begin{par}
Finally, we consider further additions to the solver beyond the core algorithm which would further improve the tool's capabilities. As mentioned in the introduction, the algorithm derived and implemented here is in both the non-relativistic and collisionless limit of the Vlasov-Maxwell system of equations. Of course, there exist many plasma systems, especially in astrophysics, where relativistic effects are important, such as gamma ray bursts or the magnetospheres of neutron stars. And many plasma systems, such as the solar wind, are not in fact collisionless, but weakly collisional. The inclusion of the Lorentz boost factors and a discretized collision operator, either reduced such as the Bhatnagar-Gross-Krook\citep{Bhatnagar:1954}, \eqr{\ref{eq:BGKOperator}}, or Lenard-Bernstein, \eqr{\ref{eq:LBOperator}}, operators \citep{Kirkwood1947, Lenard1958, Dougherty1967}, or the full Landau-Fokker-Planck, \eqr{\ref{eq:FPOperator}}, operator \citep{Rosenbluth1957}, would be useful extensions to the model. It should be noted that the inclusion of relativistic effects in the Vlasov equation necessitates the inclusion of relativistic effects in the collision operator, e.g., the implementation of the Beliav-Budker operator \citep{Beliav:1956} instead of the Landau-Fokker-Planck operator, or some relativistic generalization of the example reduced collision operators. While the inclusion of the physics of relativistic effects would be vital for simulating extreme astrophysical environments such as pulsar wind nebulae, just as the study of dissipation and entropy production requires a discrete collision operator, there are challenges to the inclusion of these model extensions which are worth acknowledging. For example, when adding relativistic effects, one key element to our algorithm that must be respected is that the integrals in the discrete-weak form must be evaluated exactly. Unlike other systems in which DG has found success, such as the solutions of fluid equations, where aliasing errors from under-integration are often tolerated, no such concession can be made for a DG discretization of the Vlasov equation, lest one is willing to risk numerical instabilities from a lack of energy and particle conservation. Unfortunately, as can be seen in \eqr{\ref{eq:relativisticVlasov}} and \eqr{\ref{eq:LorentzBoost}}, the Lorentz boost factors in the relativistic Vlasov equation will inevitably lead to aliasing errors in our current algorithm due to the fact that standard Gauss-Legendre quadrature will fail to exactly evaluate the square roots and rational functions which arise from a DG discretization of the relativistic Vlasov equation. Even an algorithmic change to pre-evaluate the integrals, such as the one we are pursuing with the modal basis, may not work as closed form expressions for the integrals may not exist. The inclusion of a discrete collision operator has its own challenges, namely that collisions tend to manifest as a stiff source term in our equations, similar to the issues seen in the solution of advection-diffusion equations \citep{Pareschi2005, LiuKun:2016}. This challenge would normally motivate an operator split so as to treat the collision term implicitly, but since collision operators such as the Lenard-Bernstein and Landau-Fokker-Planck operators are nonlinear, such a split is nontrivial and may require very sophisticated implicit solves. These issues, while challenging, do have solutions in specific, non-general cases. For example, a worthwhile avenue to pursue for including relativistic effects would be to Taylor expand the Lorentz Boost factors locally in a momentum cell, thus allowing the treatment of some relativistic effects while maintaining an alias-free algorithm. And since many plasma systems are weakly collisional, a low enough collisionality such that the collision operator is not as stiff as the advection term could be justified for certain applications. These solutions would very likely break down when attempting to model highly relativistic, highly collisional plasmas; however, a path forward exists for including the physics of collisions and special relativity without solving the fully general problem.
\end{par}
\begin{par}
Beyond the addition of a physical model for collisions and the inclusion of relativistic effects, there is much to be gained in added flexibility of our velocity space discretization. Currently, velocity space boundary conditions remain an unphysical characteristic of the discrete system. An assumption that no particles will reach velocities at the edge of velocity space is a dangerous one, as sufficient particle acceleration or heating can quickly lead to numerical issues when the plasma interacts with the velocity space ``wall.'' Extending velocity space to $\pm \infty$ via something like a ``super-cell'' wherein the volume term is integrated with Gauss-Laguerre instead of Gauss-Legendre quadrature, as some gyrokinetic codes currently employ \citep{Numata:2010, Jenko:2000a}, would allow for the treatment of plasma systems which experience non-trivial amounts of particle acceleration and/or heating. Non-uniform grids in velocity space would allow one to concentrate velocity resolution around expected resonance processes. Adaptive mesh refinement in velocity space \citep{Arslanbekov2013, Zabelok2015455} would prove invaluable for highly nonlinear simulations such as the turbulence calculation presented in this paper, as velocity space requirements vary dramatically throughout the domain, from minimal in quiescent parts of the plasma to quite demanding in the middle of reconnecting current sheets such as the one which forms in the center of the domain in the Orszag-Tang vortex. Ideally, such mesh refinement would be anisotropic, as a regridding procedure in five or six dimensions would be quite demanding, and configuration space mesh refinement is not strictly necessary; for many problems of interest, the configuration space scales one wants to resolve are well-defined. Such flexibility in velocity space grids is not guaranteed to provide significant advantage over uniform, fixed, grids, as processes such as phase-mixing can create structure everywhere in velocity space at arbitrarily small velocity space scales; however, for simulations which include significant particle acceleration such as the electrostatic shock in Section \ref{sec:ESShock}, the generation of the highly energetic tail in the proton distribution function would be much cheaper. And relativistic laser-plasma Vlasov codes in 1X1P have demonstrated significant gains with AMR \citep{Wettervik2016}. Most importantly, as discussed in Sections \ref{sec:ConservationTests} and \ref{sec:OT}, there is currently nothing inhibiting the distribution function from going negative, nor is there anything preventing the equation system from accumulating errors in the divergence relations in Maxwell's equations, Eqns. \ref{eq:divE} and \ref{eq:divB}. Although the fluid moments such as density and temperature are positive definite for the extent of the run, the distribution function likewise should also be positive definite, and although the divergence errors remain small, charge should be conserved and the divergence of the magnetic field should remain zero. The implementation of a positivity, and perhaps monotonicity, enforcing limiter, either to the solution or the differential operator, which preserves particles and energy, and the implementation of a curl-conforming DG method for Maxwell's equations, such as the one outlined in Pintos \emph{et al.} \citep{campospinto:hal-01303861}, would be the two most valuable improvements of all the algorithmic enhancements discussed here. 
\end{par}
\begin{par}
Nonetheless, despite many possibilities for future development, we have demonstrated the robustness of a DG Vlasov-Maxwell solver through a variety of benchmarks. The level of signal present in the continuum discretization of the Vlasov-Maxwell system is quite substantial. The Orszag-Tang run exhibits greater than eight orders of magnitude of signal in the distribution function. Given cutting edge diagnostics such as the field particle correlations \citep{Howes2017, Klein2016, Klein2017}, the value of having a noise-free distribution function is made manifest. Many exciting problems can be tackled with this tool, spanning the space of plasma systems found in the lab, such as radio-frequency heating problems, and space and astrophysical plasmas, such as ongoing attempts to quantify dissipation in plasma turbulence in the heliosphere.
\end{par}

\vspace{1em}
\noindent {\bf Acknowledgements}
\nopagebreak
\vspace{1em}
\nopagebreak

The authors acknowledge fruitful discussions with Greg Hammett, Kristopher Klein, and Petr Cagas. The work of J. Juno and J. TenBarge was supported by the National Science Foundation SHINE award No. AGS-1622306 and the UMD DOE grant DE-FG02-93ER54197. The work of W. Dorland was also supported by the UMD DOE grant DE-FG02-93ER54197. The work of A. Hakim and E. Shi was supported by the U.S. Department of Energy under Contract No. DE-AC02-09CH11466, through the Max-Planck/Princeton Center for Plasma Physics and the Princeton Plasma Physics Laboratory. A. Hakim was also supported by the Air Force Office of Scientific Research under Grant No. FA9550-15-1-0193. This work used the Extreme Science and Engineering Discovery Environment (XSEDE), which is supported by National Science Foundation grant number ACI-1548562.

\vspace{1em}
\noindent {\bf References}
\bibliographystyle{elsarticle-num} 
\bibliography{abbrev2.bib,all.bib}


\end{document}